\documentclass[letterpaper,10pt]{article}

\usepackage{amsmath,amssymb,amsfonts}
\usepackage{algorithmic}
\usepackage{graphicx}
\usepackage{textcomp}
\def\BibTeX{{\rm B\kern-.05em{\sc i\kern-.025em b}\kern-.08em
    T\kern-.1667em\lower.7ex\hbox{E}\kern-.125emX}}

\usepackage{color}
\usepackage[usenames, dvipsnames]{xcolor}
\usepackage{xspace}
\usepackage{mathpartir}
\usepackage{listings}
\usepackage{amssymb, amsthm,stmaryrd} 
\usepackage{wrapfig}
\usepackage{float}
\usepackage{authblk}

\usepackage{anysize}
\marginsize{1in}{1in}{1in}{1in}

\usepackage[pdftex,colorlinks]{hyperref}
\hypersetup{
    bookmarks=true,         
    unicode=false,          
    pdftoolbar=true,        
    pdfmenubar=true,        
    pdffitwindow=true,      
    pdftitle={},    
    pdfauthor={},     
    pdfsubject={},   
    pdfnewwindow=true,      
    pdfkeywords={}, 
    colorlinks=true,       
    linkcolor=blue,          
    citecolor=blue,        
    filecolor=magenta,      
    urlcolor=cyan           
}

\newcommand{\annotationcolor}{\color{ForestGreen}}
\newcommand{\codebox}[1]{\vspace{0.15cm} \\ #1 \vspace{0.15cm} \\}

\newcommand{\m}[1]{\mathtt{#1}}
\ifx\lvec\undefined
	\newcommand{\lvec}[1]{\overrightarrow{#1}}
\fi

\makeatletter
\newsavebox{\@brx}
\newcommand{\llangle}[1][]{\savebox{\@brx}{\(\m@th{#1\langle}\)}%
  \mathopen{\copy\@brx\kern-0.5\wd\@brx\usebox{\@brx}}}
\newcommand{\rrangle}[1][]{\savebox{\@brx}{\(\m@th{#1\rangle}\)}%
  \mathclose{\copy\@brx\kern-0.5\wd\@brx\usebox{\@brx}}}
\makeatother

\newcommand{\Paragraph}[1]{\vspace{5pt}\noindent{\bf #1}}
\newcommand{\langname}{\textit{polC}\xspace}
\newcommand{\minic}{$\mu$C\xspace}
\newcommand{\sysname}{$\m{FlowNotation}$\xspace}

\newcommand{\tJoin}[2]{#1 @ #2}

\newcommand{\flnCap}[1]{\textit{#1}}
\newcommand{\flnTag}[1]{\textit{#1}}
\newcommand{\flnKey}[1]{\mbox{{\tt\#}}#1}

\newcommand{\bnfdef}{::=}
\newcommand{\bnfalt}{\,|\,}
\newcommand{\rulename}[1]{\textsc{#1}}

\newcommand{\elet}{\mathsf{let}}
\newcommand{\ein}{\mathsf{in}}

\newcommand{\eif}{\mathsf{if}}
\newcommand{\ethen}{\mathsf{then}}
\newcommand{\eelse}{\mathsf{else}}

\newcommand{\enew}{\mathsf{new}}
\newcommand{\eloc}{\mathit{loc}}
\newcommand{\erelab}{\mathsf{reLab}}
\newcommand{\epair}[2]{\langle #1\bnfalt #2\rangle}
\newcommand{\eupdate}{\mathsf{upd}}
\newcommand{\eread}{\mathsf{rd}}

\newcommand{\bop}{\mathrel{\mathsf{bop}}}
\newcommand{\proj}[2]{\lfloor{#1}\rfloor_{#2}}
\newcommand{\dne}{\mathsf{d\&e}}

\newcommand{\bv}{\mathit{v}}
\newcommand{\lv}{\mathit{lv}}
\newcommand{\lexp}{\mathit{le}}

\newcommand{\extV}{v^{+}}
\newcommand{\extE}{e^{+}}

\newcommand{\tunit}{\mathsf{unit}}
\newcommand{\tpint}{\mathsf{int}}
\newcommand{\tptr}{\mathsf{ptr}}
\newcommand{\tstruct}{\mathsf{struct}}

\newcommand{\lab}{\ell}
\newcommand{\labof}{\mathit{labOf}}
\newcommand{\tpof}{\mathit{tpOf}}
\newcommand{\tmof}{\mathit{tmOf}}
\newcommand{\pc}{\mathit{pc}}
\newcommand{\pol}{\rho}

\newcommand{\join}{\sqcup}

\newcommand{\storetp}{\Sigma}
\newcommand{\tpdefctx}{D}

\newcommand{\hole}{[\,]}
\newcommand{\stepsto}{\longrightarrow}
\newcommand{\sepidx}[1]{\mathrel{/_{#1}}}

\newcommand{\at}{\mathrel{\mathsf{at}}}
\newcommand{\un}{\mathit{U}}
\newcommand{\trans}[1]{\llbracket #1 \rrbracket}
\newcommand{\inscast}[1]{\llangle #1 \rrangle}
\newcommand{\genname}{\mathit{genName}}

\newcommand{\ee}{\mathcal{E}}

\theoremstyle{plain}
\newtheorem{thm}{Theorem}
\newtheorem{lem}[thm]{Lemma}

\newtheorem{defn}[thm]{Definition}

\newenvironment{proofsketch}{\noindent{\it Proof (sketch):}\hspace*{0.25em}}{ \hspace*{\fill} \qed}

\newcommand{\eat}[1]{}
\begin{document}

\title{FlowNotation: Uncovering Information Flow Policy\\ Violations in C Programs}
\date{}
\author[1]{Darion Cassel}
\author[2]{Yan Huang}
\author[1]{Limin Jia}
\affil[1]{Carnegie Mellon University}
\affil[2]{Indiana University}

\maketitle

\begin{abstract}
Programmers of cryptographic applications written in C need to avoid
common mistakes such as sending private data over public
channels, modifying trusted data with untrusted functions, or
improperly ordering protocol steps.
These secrecy, integrity, and sequencing policies can be cumbersome to 
check with existing general-purpose tools.
We have developed a novel means of specifying and uncovering violations of these
policies that allows for a much lighter-weight approach than previous
tools.
We embed the policy annotations in C's type system via a source-to-source
translation and leverage existing C compilers to check for 
policy violations, achieving high performance and scalability.
We show through case studies of recent cryptographic libraries and applications
that our work is able to express detailed policies for large bodies of C
code and can find subtle policy violations.
To gain formal understanding of our policy annotations, we show
formal connections between the policy annotations
and an information flow type system and prove a noninterference
guarantee.
\end{abstract}

\maketitle

\section{Introduction}
\label{sec:intro}

Programs often have complex data invariants and API usage policies
written in their documentation or comments. The ability
to detect violations of these invariants and policies is key
to the correctness and security of programs. This is particularly
important for cryptographic protocols and libraries as the
security of a large system depends on its underlying secure
protocols and primitives. 
%
%
As a result, there has been much interest in checking implementations of cryptographic
protocols~\cite{Jasmine:Almeida:2017,Cryptol:Lewis:2003,
  Blanchet:2001, Vale:Bond:2017, Bhargavan:2010, Bhargavan:2008, Scyther:Cremers:2008, Cortier:2005}.
 These verification systems, while comprehensive in their scope, require expert knowledge
 of both the cryptographic protocols and the verification tool to be used effectively.

 What remains missing is a lightweight and developer-friendly tool to help
 programmers identify programming errors at compile time that violate
 high-level policies on cryptographic libraries and protocols written
 in C.
The policies that
are particularly important are secrecy (e.g., sensitive data is not
given to untrusted functions), integrity (e.g., trusted data is not
modified by untrusted functions), and API call sequencing (e.g., the
ordering of cryptographic protocol steps is maintained). These
policies can be viewed as information flow policies.
 

In this paper, we present a framework called \sysname where C programmers can
add lightweight annotations to their programs to express policy
specifications. These policies are then automatically checked using a C
compiler's type checker, potentially revealing policy violations in the implementation.  
Our annotations are in the same family as {\em type
  qualifiers} (e.g. CQual~\cite{Zhang:2002, Chin:2005, Thesis:Foster:2002}),
where qualifiers such as tainted and trusted are used to identify violations of
integrity properties of C programs; supplying tainted inputs to a function
that requires a trusted argument will cause a type error.  Our work
extends previous results to support more complex and refined sequencing
properties. Consider the following policy: a data object is
initially tainted, then it is sanitized using a \lstinline{encodeURI}
API, then serialized using a \lstinline{serialize} API, and finally written to
disk using a \lstinline{fileWrite} API. Such API sequencing patterns are quite common, but cannot be
straightforwardly captured using previous type qualifier systems.

\sysname extends type qualifiers to include a sequence of labels for specifying
policies similar to the above example. However, rather than implement a new type
system, we develop a source-to-source transformation tool, which translates an
annotated C program to another C program, through which a C compiler's type
checker (indirectly) checks the annotated policies. The key insight is that
qualified C types can be translated to C structures whose fields are the
original C types. For instance, ``\lstinline{trusted int}'' and
``\lstinline{tainted int}'' can be translated to ``\lstinline[]!typedef struct {int x;} int_trusted!''  
and ``\lstinline[]!typedef struct {int x;} int_tainted!'', respectively.  
Even though these two types are structurally
equivalent, C's struct types are nominal types, and thus, attempts to use data
of one type as the other will be reported as a compile-time error by a C type
checker. Consequently, we can directly use C type checkers for policy checking.
The benefit of this approach is that we can leverage performant C
compilers to quickly type-check our policies over large codebases.

To gain a formal understanding of the type of errors that we can uncover
with this system, we model the annotated types as {\em information flow
  types}, which augment ordinary types with security labels. 
%
We define a core language \langname and prove that its information flow type
system enforces noninterference. The novelty of \langname's type system is that
the security labels are sequences of secrecy and integrity labels, specifying
the path under which data can be relabeled. Relabeling corresponds to {\em
  declassification} (marking secrets as public) and {\em endorsement}
(marking data from untrusted source as trusted). 
The type system ensures that relabeling functions are called in the correct order.

We also define \minic, a core imperative language with nominal types
but without information flow labels in order to model a fragment of C. We then
formally define our translation algorithm based on \langname and
\minic. We prove the correctness of our translation algorithm:
If the translated program is accepted by the type checker in \minic,
then the original program is well-typed in \langname. The formalism
not only makes explicit assumptions made by our algorithm, but also
provides a formal account of the properties being checked by the annotations.

To demonstrate the effectiveness of \sysname we implement a prototype for a
subset of C and evaluate the prototype on several cryptographic libraries. Our
evaluation shows that we are able to check useful information flow
policies in a modular way and uncover subtle program flaws.
%

This paper makes the following technical contributions:
\begin{itemize}
\item We propose \sysname, a lightweight tool for finding errors that
  violate information flow policies in C programs.
\item We connect annotations in \sysname to the
  information flow type system \langname. We prove a noninterference theorem for
  \langname's type system, from which the property of correct API sequencing
is a corollary.
\item We define a translation algorithm from \langname types
  to nominal types (modeled by \minic) and prove it correct.
\item We implement a prototype and demonstrate the effectiveness of
  \sysname by evaluating it on several C
  cryptographic libraries and applications.
\end{itemize}

The rest of this paper is organized as follows: Section~\ref{sec:motivation}
presents a motivating example and describes the workflow of \sysname.  Next, we
define \langname (Section~\ref{sec:lang}) and \minic
(Section~\ref{sec:translation}) along with the algorithm for our translation
process.  In Section~\ref{sec:implementation}, we explain how the algorithms are
implemented in C. Our case studies and evaluation results are presented in
Section~\ref{sec:case-studies}.  Finally, we discuss related work in
Section~\ref{sec:related} and conclude in Section~\ref{sec:conclusion}.



\section{Overview and Motivating Examples}
\label{sec:motivation}
We illustrate how \sysname concretely works on the left side of Figure~\ref{fig:overview}.
First, to check an application-specific policy, a programmer writes the policy in C
pragma annotations. Then our source-to-source translator takes the annotated
program as input, and produces a translated C program. The resulting program is
then type-checked using an off-the-shelf C compiler. If the compiler returns a
type error, then this implies the policy is violated in the program.

\begin{figure}[b!]
\centering
\includegraphics[width=0.4\textwidth]{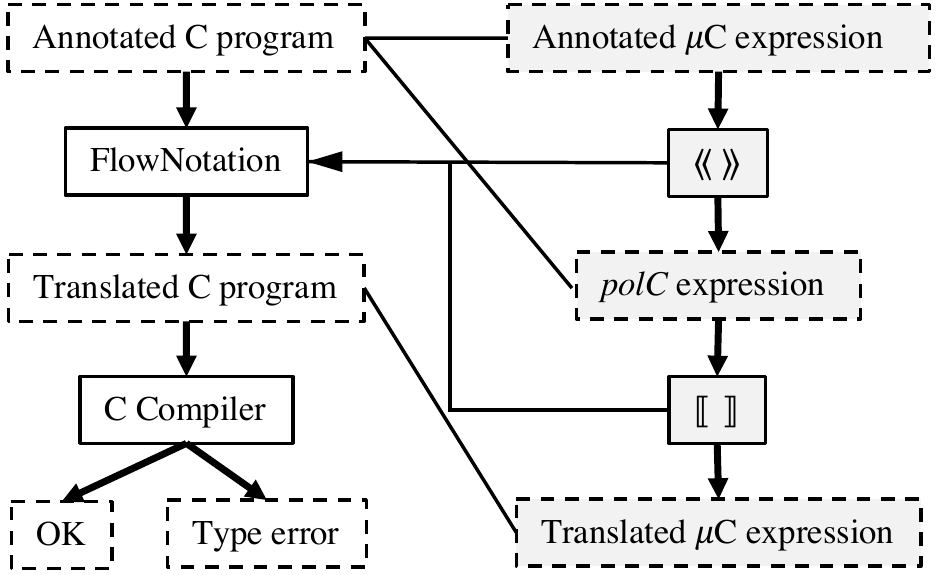}
\caption{Overview of \sysname and connections to the formal model.}
\label{fig:overview}
\end{figure}


Next we show example 
policies in the context of developing cryptographic
applications.

\subsection{Secrecy}
Suppose a team of software developers is working on a large C project that uses
customers' financial data. This project integrates a secure two-party
computation component that allows Alice and Bob to find out which of the two is
wealthier without revealing their wealth to the other or relying on a trusted
third party. Let us assume that the program obtains Alice's balance using the
function \lstinline{get_alice_balance}, then calls function
\lstinline{wealthierA} to see whether Alice is wealthier than
Bob. \lstinline{wealthierA}'s implementation uses a library that provides APIs
for secure computation primitives.

\begin{wrapfigure}[6]{l}{0.4\textwidth}
\begin{lstlisting}[aboveskip=-3pt]
	int bankHandler() {
		int balA;
		balA = get_alice_balance();
		...
		wealthierA(balA);
	}
\end{lstlisting}
\end{wrapfigure}
The variable \lstinline{balA} contains Alice's balance, and therefore should be
handled with care. In particular, the programmer wants to check that the secrecy of
\lstinline{balA} is maintained. One method is to use information flow types
(e.g.~\cite{Smith:1997}), where the information flow
type of \lstinline{balA}  is (\lstinline{int AlicePrivate}), indicating that it is an
integer containing an \lstinline{AlicePrivate} type of secret. In contrast,
variables that do not contain secrets can be given the type (\lstinline{int Public}). The information flow type system then makes sure that read and write
operations involving \lstinline{balA} are consistent with its secrecy label. For
instance, if a function \lstinline{postBalance(int Public)}, which is meant to
post the balance publicly, is called with \lstinline{balA} as the
argument, the type system will reject this program for violating the secrecy
policy.

Our annotations are {\em information flow labels}, each of which has a {\em secrecy}
  component and an {\em integrity} component. 
  Programmers can provide these annotations above the declaration of
  \lstinline{balA} to specify the secrecy policy as follows: %

\begin{wrapfigure}[3]{l}{0.4\textwidth}
\begin{lstlisting}[aboveskip=-3pt]
	#requires AlicePriv:secrecy
	int balA;
\end{lstlisting}
\end{wrapfigure}
In the annotation, \lstinline{#requires} is a directive that
allows our tool to parse this annotation (in practice, \lstinline{#pragma} prefaces it).
\lstinline{AlicePriv} is a secrecy label. Finally, \lstinline{secrecy} is a \textit{projection}; 
it specifies that we only care about the secrecy component of the
label. \lstinline{balA}'s integrity component is automatically 
assigned \lstinline{bot}, the lowest integrity. The information flow type of
\lstinline{balA} corresponding to this annotation is \lstinline{int(AlicePrivate, bot)}.
%
This annotation can be used to 
 check this program for violations of the following policy $P_1$. 
\[
\begin{array}{ll}
P_1: & \mbox{\lstinline{balA}}~\textit{should never be given as input}
\\& \textit{to an untrusted function}.
\end{array}
\]
Here, trusted functions are those trusted by the programmer not to 
 leak \lstinline{balA}. Next, we discuss how a
 programmer can annotate trusted functions. 

Our programmer trusts a secure computation library
that provides secure computation primitives. Let us assume 
the API \lstinline{encodeA} converts an integer argument into a bit 
representation similar to what is used in Obliv-C \cite{Zahur:2015} for use with a garbled circuit. 
The API \lstinline{yao_execA} takes a pointer to a function \lstinline{f} and
an argument for \lstinline{f}, and runs \lstinline{f} as a circuit with Yao's protocol \cite{Yao:1986}. 
Finally, at the end of the application's execution the API \lstinline{reveal}
is invoked to give the result of the function execution to both parties.
The programmer constructs the following code for Alice (Bob's program is symmetric, which we omit):
\begin{lstlisting}
	int compare(int a, int b) { return a > b; }
	int wealthierA(balA) {
		balA2 = encodeA(balA);
		int res = yao_execA(&compare, balA2);
		reveal(&res, ALICE);
	}
\end{lstlisting}
This program first encodes Alice's balance, and then calls
\lstinline{yao_execA} with the comparison function and  Alice's encoded
balance \lstinline{balA2} as arguments, and finally calls \lstinline{reveal}. 

The code as it stands will not type-check after being translated, unless the
programmer also appropriately annotates their trust in the secure computation APIs.
\begin{lstlisting}
	#param AlicePriv:secrecy
	int encodeA(int balA);
	#param(2) AlicePriv:secrecy
	int yao_execA(void* compare, int balA);
\end{lstlisting}
These two annotations state that the functions must accept parameters with 
the label \lstinline{AlicePriv}. In the second annotation,
\lstinline{#param(2)} specifies that the
annotation should only apply to the second parameter.
A violation of $P_1$ will be detected, when \lstinline{balA} is given
to a function that does not have this kind of annotation; e.g. that is
not allowed (by the programmer) to accept
\lstinline{AlicePriv}-labeled data.


\subsection{Integrity and Sequencing}
\label{sec:example-seq}
A programmer can also use \sysname to check the program
for violations of the following, more refined, policy $P_2$. 
\[\begin{array}{ll}
	P_2: & \mbox{\lstinline{balA}}~\textit{should be used by the encoding function}
\\& \textit{and then by the Yao protocol execution}.
\end{array}
\]
The annotation for \lstinline{balA} is as follows.
\begin{lstlisting}
	#requires AlicePriv:secrecy then EncodedBal:integrity
	int balA;
\end{lstlisting}
The keyword \lstinline{then} allows for the
sequencing of labels.  Corresponding changes
are made to the other annotations:

\begin{lstlisting}
	#param AlicePriv:secrecy
	#return EncodedBal:integrity
	int encodeA(int balA);
	#param(2) EncodedBal:integrity
	int yao_execA(void* compare, int balA);
\end{lstlisting}
The \lstinline{encodeA} function, as before, requires the argument to have the
\lstinline{AlicePriv} secrecy label. In addition, the return value from
\lstinline{encodeA} will have the integrity label \lstinline{EncodedBal},
stating that it is endorsed by the \lstinline{encodeA} function to be properly
encoded.  The \lstinline{yao_execA} function requires the argument to have the
same integrity label. If only programmer-approved encoding functions are
annotated with \lstinline{EncodedBal} at their return value, the type system
will check that an appropriate API call sequence (\lstinline{encodeA} followed by
\lstinline{yao_execA}) is applied to the value stored in \lstinline{balA}.

\section{A Core Calculus for Staged Release}
\label{sec:lang}

We formally define the syntax, operational semantics, and the
type system of \langname, which models annotated C programs that 
\sysname takes as input. We show that 
\langname's type system can 
enforce not only secrecy and integrity
policies, but also staged information release and data
endorsement policies. We prove that our type system enforces
noninterference, from which the property of staged information release
is a corollary.

\subsection{Syntax and Operational Semantics}
The syntax of \langname is summarized in Figure~\ref{fig:lang-syntax}. We write $\lab$ to denote
security labels, which consist of a secrecy tag $s$ and an integrity
tag $\iota$. We assume there is a security lattice
$(S, \sqsubseteq_S)$ for secrecy tags and a security lattice
$(I, \sqsubseteq_I)$ for integrity tags. The security lattice
$\mathcal{L} = (L, \sqsubseteq)$ is the product of the
above two lattices. The top element of the lattice is $(\top_S,\bot_I)$
(abbreviated $\top$), denoting data that do not contain any secret
and come from the most trusted source; and the bottom element is $(\bot_S,\top_I)$
(abbreviated $\bot$), denoting data that contain the most secretive
information and come from the least trusted source.
%
\begin{figure}[tbp]
\(
\begin{array}{lcll}
\textit{labels}   & \lab & \bnfdef & (s, \iota)
\\
\textit{policies}    & \pol & \bnfdef &  \bot\bnfalt \top \bnfalt  \lab::\rho
\\
\textit{1st order types} & b & \bnfdef & 
\tpint \bnfalt \tptr(s)  \bnfalt  T 
\\
\textit{simple sec. types} & t & \bnfdef & b \ \rho \bnfalt \tunit
\\
\textit{security types} & s & \bnfdef & t \bnfalt
                                        [\pc](t\rightarrow t)^\pol
\\
\textit{values} & v & \bnfdef & x \bnfalt n \bnfalt ()\bnfalt f \bnfalt T\{v_1,\cdots, v_k\} \bnfalt \eloc
\\
\textit{expressions} & e & \bnfdef &
 v \bnfalt e_1\m{bop}\; e_2\bnfalt v\, e 
\bnfalt \elet\, x=e_1\, \ein\, e_2
\\
& & \bnfalt &  v.i \bnfalt
\eif\, v_1\, \ethen\, e_2\, \eelse \, e_3 \bnfalt v\, :=\, e
\\
& & \bnfalt & \enew\, e \bnfalt {*}v
\bnfalt \erelab(\lab'{::}\bot \leftarrow
                               \lab{::}\top)\,  v
\end{array}
\)
\vspace{-5pt}
\caption{Syntax of \langname}
\label{fig:lang-syntax}
\end{figure}
A policy, denoted $\pol$ is a sequence of labels specifying the
precise sequence of relabeling (declassification and endorsement) 
of the data. The example from Section~\ref{sec:example-seq} uses the
following policy:
\codebox{$(\flnTag{AlicePrivate},\bot_I)::(\bot_S, \flnTag{EncodedBal})::\bot$}
A policy always ends with
either the top element, indicating no further relabeling is allowed, or the bottom
element, indicating arbitrary relabeling is allowed. For our
application domain, the labels provided by programmers are distinct points
in the lattice that are not connected by any partial order relations
except the $\top$ and $\bot$ elements.

A simple (first-order) security type, denoted $t$, is obtained by adding policies to
ordinary types. Our core language supports integers ($\tpint$),
$\tunit$, pointers ($\tptr(s)$), and record types
($\tstruct\ T\ \{t_1,\cdots,t_k\}$ to model C structs). Here $T$ is
the defined name for a record type. To simplify our formalism, we
assume that defined type $T$ is always a record type named $T$.
Unlike ordinary information types, our information flow types use the
policy $\pol$, rather than a single label $\lab$. The meaning of an
expression of type $\tpint\ \pol$ is that this expression is evaluated
to an integer and it induces a sequence of declassification
(endorsement) operations according to the sequence of labels specified by
$\pol$. For instance, $e:\tpint\ H::L::\bot$ means that $e$ initially
is of $\tpint\ H$, then it can be given to a declassification function
to be downgraded to $\tpint\ L$, the resulting expression can be
further downgraded to bottom. $e:\tpint\ H::L::\top$ is similar except
that the last expression cannot be declassified further; i.e.. it stays at $L$ security
level.  The annotated type for \lstinline{balA} in Section~\ref{sec:example-seq}
can be similarly interpreted.

We do not have a labeled $\tunit$ type, because it is inhabited by 
one element $()$ and thus does not
contain sensitive information.
A function type is of the form
$[\pc](t_1\rightarrow t_2)^\pol$, where $t_1$ is the argument's type,
$t_2$ is the return type, $\pol$ is the security label of the function indicating who 
can receive this function, and $\pc$, called the program counter, is
the security label representing where this function can be called. For
instance a function $f$ of type $[L{::}\bot](t_1\rightarrow
t_2)^{H{::}\bot}$ cannot be called in an if branch that branches on
secrets and the function itself cannot be given to an attacker whose
label is $L{::}\bot$. 

Our expressions are reminiscent of A normal forms: all
elimination forms use only values (e.g., $v.i$, instead of
$e.i$). This not only simplifies our proofs, but also the translation
rules (presented in Section~\ref{sec:translation}). The
fragment of C that is checked in our case studies is quite similar to
this form.

Values can be variables, integers, unit, functions, records, and store
locations. Since we are modeling an imperative language, we do not
have first-class functions. Instead, all functions are predefined, and
stored in the context $\Psi$. Expressions include function calls, if
statements, let bindings, and store operations. One special expression
is the relabeling (declassification) operation, written $\erelab(\lab'{::}\bot \leftarrow
                               \lab{::}\top)\ v $. This operation changes the
label of $v$ from $\lab{::}\top$ to $\lab'{::}\bot$. Such an expression should only appear in
trusted declassification functions. For our applications, we further
restrict the relabeling to be between two labels; from one ending with the
top element to one ending with bottom element. We will explain this
later when we explain the typing rules.

The judgement for small step semantics for \langname
 is denoted $\Psi \vdash \sigma\sepidx{} e
\stepsto \sigma' \sepidx{} e'$, where $\Psi$ stores all the function code,
$\sigma$ is the store mapping locations to values and $e$ is the expression to be evaluated. 
Appendix \ref{app:polc} contains a summary of all the
operational semantic rules. 

\subsection{Typing Rules}

The type system makes use of several typing contexts. We write
$D$ to denote the context for all the type definitions. We only
consider type
definitions of record (struct) types, written $T \mapsto
\mathsf{struct}\ T\ \{t_1, \cdots, t_k\}$. 
The typing context for functions is denoted
$F$. We distinguish two types of functions: ordinary functions, and
declassification/endorsement functions whose bodies are
allowed to contain relabeling operations, written $f {:}(\dne)[\pc]t_1\rightarrow
t_2$.  $F$ does not dictate the label of a function $f$. Instead, the
context in which $f$ is used decides $f$'s label. 
\[
\begin{array}{lcll}
\textit{Type def. ctx} & D & \bnfdef &
\cdot\bnfalt D, T \mapsto \mathsf{struct}\ T\ \{t_1, \cdots, t_k\} 
\\  
\textit{Func typing ctx} & F & \bnfdef & \cdot \bnfalt F, f{:}
                                         [\pc]t_1\rightarrow t_2 
\\& & \bnfalt & F, f {:} (\dne)[\pc]t_1\rightarrow t_2
\\ 
\textit{Store Typing} & \storetp & \bnfdef & \cdot \bnfalt \storetp, \eloc: s
\end{array}
\]
We write $\storetp$ to denote the typing context for pointers. It maps
a pointer (heap location) to the type of its content. $\Gamma$ is the typing context
for variables, and $\pc$ is the security label representing the program counter.

Our type system has two typing judgments: $D;F; \Sigma; \Gamma 
\vdash v : t$ for value typing, and $D;F; \Sigma; \Gamma ; \pc
\vdash e : t$ for expression typing. Selected typing rules are shown
in Figure~\ref{fig:typing-rules-selected}; full rules are in 
Appendix~\ref{app:polc:typing}. 

\begin{figure}[t!]
%
\begin{mathpar}
\inferrule*[right=P-T-E-Val]{
  D;F; \Sigma; \Gamma  \vdash v : s
}{ 
  D;F; \Sigma; \Gamma ; \pc \vdash v: s\join\pc
}
\and
\inferrule*[right=P-T-E-Field]{
  D;F; \Sigma; \Gamma \vdash v : T\
  \pol
\\ T\mapsto \tstruct\ \{s_1,\cdots, s_n\} \in D
\\ \pc\sqsubseteq\pol
}{ 
  D;F; \Sigma; \Gamma ; \pc \vdash v.i : s_i \join\ \pol 
}
\and
\inferrule*[right=P-T-E-New]{
  D;F; \Sigma; \Gamma ; \pc \vdash e : s
\\ \pc\rhd \pol
}{ 
  D;F; \Sigma; \Gamma ; \pc \vdash \enew(e): \tptr(s)\ \pol
}
\and
\inferrule*[right=P-T-E-Deref]{
  D;F; \Sigma; \Gamma \vdash v : \tptr(s)\ \pol
\\ \pc\sqsubseteq\pol
}{ 
  D;F; \Sigma; \Gamma ; \pc \vdash *v: s\join \pol
}
\and
\inferrule*[right=P-T-E-Assign]{
  D;F; \Sigma; \Gamma  \vdash v_1 : \tptr(s)\ \pol
 \\  D;F; \Sigma; \Gamma;\pc \vdash e_2 : s
\\  \pol \rhd s
}{ 
  D;F; \Sigma; \Gamma ; \pc \vdash v_1:= e_2: \tunit
}
\and
\inferrule*[right=P-T-E-If]{
  D;F; \Sigma; \Gamma  \vdash v_1 : \tpint\ \pol
\\   D;F; \Sigma; \Gamma ; \pc\sqcup \pol  \vdash e_2 : s
\\   D;F; \Sigma; \Gamma ; \pc\sqcup \pol  \vdash e_3 : s
}{ 
  D;F; \Sigma; \Gamma ; \pc \vdash \eif\ v_1\ \ethen\ e_2\ \eelse\
  e_3 : s
}
\and
\inferrule*[right=P-T-E-DE]{
  D;F; \Sigma; \Gamma \vdash v_f: (\dne)  [\pc'] (b\ \lab_1{::}\top
  \rightarrow b\ \lab_2{::}\bot)^{\rho_f}
\\  D;F; \Sigma; \Gamma;\pc \vdash e_a :  b\ \pol
\\ \pol = \lab_1{::}\lab_2{::}\pol'
\\ \rho_f\sqcup\pc \sqsubseteq \pc'
}{ 
  D;F; \Sigma; \Gamma ; \pc \vdash v_f\ e_a: b\ \lab_2{::}\pol'
}
\and
\inferrule*[right=P-T-E-Relabel]{ 
  D;F; \Sigma; \Gamma  \vdash v : b\ \pol\\
\pc\sqsubseteq \pol'
}{ 
  D;F; \Sigma; \Gamma ; \pc \vdash \erelab(\pol'\Leftarrow\pol)\ v
  : b\ \pol'
}
\and
\inferrule*[right=P-T-E-Sub]{
  D;F; \Sigma; \Gamma ; \pc \vdash e : s'\\
 s'\leq s
}{ 
  D;F; \Sigma; \Gamma ; \pc \vdash e: s
}
\end{mathpar}
\vspace{-15pt}
\caption{Typing rules}
\label{fig:typing-rules-selected}
\end{figure}

We use a number of auxiliary definitions. First, we define the 
meaning of a policy $\pol_1$ being less strict than another, $\pol_2$, written
$\pol_1\sqsubseteq\pol_2$, as
the point-wise lifting of the label operation
$\lab_1\sqsubseteq\lab_2$. When one policy reaches its end, we use 
$\bot\sqsubseteq\pol$ or $\pol\sqsubseteq\top$. $\bot$
represents a policy that can be arbitrarily reclassified and thus is a
subtype of any policy $\pol$. On the other hand, $\top$ is the strictest
policy that forbids any reclassification; so any policy is less
strict than $\top$.


The subtyping relation $s_1\leq s_2$ is standard: most types are covariant
except function argument types, which are contravariant, and pointer content types, which are
invariant. 
$\pol\rhd t$ denotes
 $\pol$ guards $t$. It is defined as $\pol\sqsubseteq\labof(t)$. Here
 $\labof(t)$ is the outermost label of type $t$; for instance, 
 $\bot\rhd\tpint\ (\flnTag{AlicePrivate}, \bot_I)$,  $(\flnTag{AlicePrivate}, \bot_I)\rhd
 \tpint\ (\flnTag{AlicePrivate}, \bot_I)$.
Finally $s\sqcup\pol$ is the type resulting from joining the policy of
$s$ with $\pol$.
%

Most of these typing rules are standard to information flow type systems. These
rules carefully arrange the constraints on policies and the program counter so
that noninterference theorem can be proven. Due to space constraints, we only
explain the rule \rulename{P-T-E-DE}, 
which types the application of a declassification/endorsement function
and is unique to our system.
%
The first premise checks that
$v_f$ relabels data from $\lab_1$ to $\lab_2$. The second premise checks
that $e_a$'s type matches that of the argument of $v_f$; further, $e_a$'s
policy $\pol$ has $\lab_1$ and $\lab_2$ as the first two labels,
indicating that $e_a$ is currently at security level $\lab_1$ and the
result of processing $e_a$ has label $\lab_2$. Finally, the return type
of the function application has the tail of the policy $\pol$. The
policy of $e_a$ does not change; instead, the policy of the
result of the relabeling function inherits the tail of $e_a$'s
policy. Therefore, our type system is not enforcing type states of
variables as found in the Typestate system~\cite{Typestate:Strom:1986}.
These declassification and endorsement functions only rewrite one
label, not a sequence of labels. This allows us to have finer-grained
control over the stages of relabeling.

\subsection{Noninterference}

We prove a noninterference theorem for \langname's type system by adapting
the proof technique used in FlowML~\cite{Pottier:2002}. We extend our language to
include pairs of expressions and pairs of values to simulate two executions
that differ in ``high'' values. 
We only explain the key definitions for the theorem.

We first define equivalences of expressions in terms of an attacker's
observation. We assume that the attacker knows the program and can
observe expressions at the security level $\lab_A$. To be consistent,
when $\lab_A$ is not $\top$ or $\bot$,
the attacker's policy is written $\lab_A{::}\top$. Intuitively, an
expression of type $b\ \pol$ should not be visible to the attacker
if 
existing declassification functions cannot relabel data with label
$\pol$ down to $\lab_A{::}\top$. For instance, if $\pol= H{::}L{::}\bot$ and there
is no declassification function from $H$ to $L$, then an attacker at
$L$ cannot distinguish between two different integers $v_1$ and $v_2$ of type $\tpint\
\pol$. On the other hand, if there is a function
$f:_\dne \tpint\ H{::}\top\rightarrow L{::}\bot$, then $v_1$ and $v_2$ are
distinguishable by the attacker. We define when a policy $\pol$ is in
$H$ with respect to the attacker's label, the function context, and
the relabeling operations, in other words, when values of type
$b\ \pol$ are not observable to the attacker, as follows. $\pol\in H$
if $\pol$ cannot be {\em rewritten} to be a policy that is lower or
equal to the attackers' policy. 
\[\inferrule*{
\forall \pol', 
F; R\vdash \pol\leadsto\pol', \pol'\not\sqsubseteq\pol_A
}{
\pol_A; F; R\vdash \pol\in H
}
\]

Here $F; R\vdash \pol\leadsto\pol'$
holds when $\pol=\lab_1{::}\cdots{::}\lab_i{::}\pol'$ and there is a
sequence of relabeling operations in $F$ and $R$, using which 
$\pol$ can be rewritten to $\pol'$. For instance, when $\lab_A =\bot$
\[
\begin{array}{lcl}
F_1 &= &\text{\lstinline{encodeA}}:(\dne)\tpint\ (\flnTag{AlicePrivate},\bot_I)::\top \\
  ~ & ~ &\qquad\qquad\qquad\rightarrow\tpint\ (\bot_S,\flnTag{EncodedBal})::\bot
\\ 
F_2 & = & F_1, \text{\lstinline{yao\_execA}}:(\dne)
         \tpint\  (\bot_S,\flnTag{EncodedBal})::\top\rightarrow \tpint\ \bot
\end{array}
\]
\[
	\noindent\lab_A; \cdot; \cdot\vdash (\flnTag{AlicePrivate},\bot_I) \in H
\]
\[
	\lab_A; F_1; \cdot\vdash (\bot_S,\flnTag{EncodedBal})\in H
\]
\[
	\lab_A; F_2; \cdot\nvdash (\bot_S,\flnTag{EncodedBal})\in H
\]
 Our noninterference theorem is formally defined below. The theorem states
that given an expression $e$ that is observable by the attacker, and two equivalent
substitutions $\delta_1$ and $\delta_2$ for free variables in $e$,
and both $e\delta_1$ and $e\delta_2$ terminate, then they must evaluate to the
 same value. In other words, the values of sub-expressions that are not
 observable by the attacker do not influence the value of observable
 expressions. The proof can be found in Appendix~\ref{app:polc}.

\begin{thm}[Noninterference]
~\\
  If $D;F; \Gamma; \bot \vdash e : s$, $e$ does not contain any
  relabeling operations, given attacker's
  label $\lab$, and substitution
  $\delta_1$, $\delta_2$ s.t.
  $F \vdash \delta_1\approx_H \delta_2 : \Gamma$, 
and $\lab;F;\cdot\vdash \labof(s)\notin H$ and
$\Psi \vdash \emptyset\sepidx{} e\delta_1 \stepsto^*
\sigma_1\sepidx{} v_1$ and 
$\Psi \vdash \emptyset \sepidx{} e\delta_2 \stepsto^*
\sigma_2\sepidx{} v_2$, then $v_1=v_2$.
\end{thm}

It follows from Noninterference that given 
$D;F; x: \lab_1{::}\cdots{::}\lab_n{::}\bot\ \tpint \vdash e: \tpint\
\lab_n{::}\top$ where the attacker's label is
$\lab_n{::}\top$,  the attacker can only gain knowledge about the
value for $x$ if there is a sequence of declassification/endorsement
functions $f_i$s that remove label $\lab_i$ from the policy to reach $\lab_n{::}\top$. Further, if
$\lab_i\not\sqsubseteq\lab_{i+1}$, then the $f_i$s have to be applied in
the correct order, as dictated by the typing rules.


\section{Embedding in A Nominal Type System} 
\label{sec:translation}

The type system of \langname can encode  
interesting security policies and help programmers identify subtle
bugs during development. However, implementing a feature-rich language with
\langname's  type system requires
non-trivial effort. Moreover, only programmers who are willing to
rewrite their codebase in this new language can
benefit from it. 
%
%
Rather than create a new language, \sysname leverages C's type system 
to enforce policies specified by \langname's types. 

The mapping between the concrete workflow of \sysname,  \langname and
\minic, and the algorithms defined here is shown
in Figure~\ref{fig:overview}.  We first define a
simple imperative language \minic with nominal types and annotations,
which models the fragment of C that \sysname works within. We show how
the annotated types and expressions can be mapped to types and
expressions in \langname in Appendix~\ref{app:translation}. Then in
Section~\ref{sec:trans}, we show how to translate \langname programs
back to \minic. These two algorithms combined describe the core algorithm of \sysname. We prove our
translation correct in Section~\ref{sec:trans-correctness}.


\subsection{\minic and Annotated \minic}
\label{sec:minic}

Expressions in \minic are the same as those in 
\langname. The types in \minic do not have information flow
policies, which are defined below. The names of the typing
contexts remain the same. 
\[
\begin{array}{lcll}
\textit{Basic Types} & \pi & \bnfdef & T \bnfalt \tpint \bnfalt \tunit\bnfalt
                                       \tptr(\tau) 
\\
\textit{Types} & \tau & \bnfdef & \pi \bnfalt \pi_1\rightarrow\pi_2
\\
\textit{Annotation} & a & \bnfdef &
\pi \bnfalt
 T\at\pol \bnfalt \tpint \at \pol
\bnfalt  \tptr(\beta) \at \pol 
\\
\textit{Typ. Annot.} & \beta & \bnfdef & a \bnfalt a_1 \rightarrow a_2
\\
\textit{Expressions} & e & \bnfdef & \cdots\bnfalt
\elet\, x:\beta=e_1\, \ein\, e_2
\\  
\textit{Annot. typedef} & D_a & \bnfdef & \cdot \bnfalt
D_a, T\mapsto \tstruct\ T  \{a_1,\cdots, a_k\}
\\  
\textit{Annot. Func.} & F_a & \bnfdef & \cdot
\bnfalt F_a, f : a_1\rightarrow a_2
\\& & \bnfalt & F_a, f : (\dne) a_1\rightarrow a_2
\end{array}
\]
We assume that programmers will provide policy annotations, denoted
$\beta$. The annotated types $\beta$ are very similar to labeled types $s$. We
keep them separate, as programmers do not need to write out the fully
labeled types. A programmer can annotate defined record types $T\at\pol$,
integers $\tpint \at \pol$, both the content and the pointer itself
$\tptr(\beta)\at \pol$, or the record type
$\tstruct\ T \{\beta_1,\cdots, \beta_k\}$. The last case is
used to annotate type declarations in the context $D$. We extend
expressions with annotated expressions; $\elet\, x:a=e_1\, \ein\,
e_2$. We assume that let bindings, type declarations, and function
types are the only places where programmers provide annotations.
A complete account of syntax and semantics can be found in 
Appendix~\ref{app:minic} and~\ref{app:translation}.

\subsection{Translating Annotated Programs to \minic} 
\label{sec:trans}

Instead of defining an algorithm to translate an annotated \minic
program $e_a$ to another \minic program, we first define an algorithm
that maps $e_a$ into a program $e_l$ in \langname; then an algorithm that
translates $e_l$ to a \minic program.

\vspace{5pt}\noindent{\bf Mapping from annotated \minic to \langname.}
This mapping helps make
explicit all the assumptions and necessary declassification and
endorsement operations needed to interpret those annotations as
proper \langname types and programs. 

We write $\inscast{\beta}$ to denote the mapping of unannotated and
annotated \minic types to \langname types. Unannotated
types are given a special label $\un$ (unlabeled,  defined as $(\bot,\bot)$); annotated types are
translated as labeled types. All 
function types are given the pc label $\bot$, so the function body can
be typed with few restrictions. The mapping from annotated types to
\langname types is summarized in Figure~\ref{fig:inscast-type}. 

\begin{figure}[tbp]
\begin{mathpar}
\inferrule*{ 
 \pi\in\{\tpint, T\}
 }{
  \inscast{\pi}  = \pi\ \un
}
\quad
\inferrule*{   \pi\in\{\tpint, T\}
 }{
  \inscast{\pi\at\pol}  = \pi\ \pol
}
\quad
\inferrule*{ 
  \inscast{\beta} =s
 }{
  \inscast{\tptr(\beta)}  = \tptr(s)\ \un
}
\and
\inferrule*{ 
  \inscast{\beta} =s
 }{
  \inscast{\tptr(\beta)\at \pol}  = \tptr(s)\ \pol
}
\and
\inferrule*{
 \forall i\in[1,2],    \inscast{a_i}  = t_i\\
}{
\inscast{a_1 \rightarrow a_2} = [\bot](t_1\rightarrow t_2)
}
\and
\inferrule*{
 \forall i\in[1,2],    \inscast{a_i}  = t_i\\
}{
\inscast{(\dne)a_1 \rightarrow a_2} = (\dne)[\bot](t_1\rightarrow t_2)
}
\end{mathpar}
\vspace{-20pt}
\caption{Mapping annotations to types}
\label{fig:inscast-type}
\end{figure}

There are two sets of mapping rules for expressions:

$D_a;F_a;\Gamma_a; s\vdash \inscast{e} \Rightarrow\lexp$ and
$D_a;F_a;\Gamma_a \vdash \inscast{v}\Rightarrow \lv$. 

The mapping rules use the annotated typing contexts: $D_a$,
$F_a$, and $\Gamma_a$. The reading of the first judgement is that an
annotated expression $e$ is mapped to a labeled expression $\lexp$
given annotated typing contexts $D_a$, $F_a$, $\Gamma_a$, and \langname
type $s$, which $e$'s type is supposed to be.  The
second judgment is similar, except that it only applies
to values and the type of $v$ is not given. Here $\lexp$ and
$\lv$ are expressions with additional type annotations of form $@s$ to
ease the translation process from \langname to \minic. For instance,
$n@\tpint\ U$ means that $n$ is an integer and it is supposed to have
the type $\tpint\ U$. This way, we can give the same integer different
types, depending on the context under which they are used: 
$n@\tpint\ U$ and $n@\tpint\ \pol$ are translated into different terms.

A value is mapped to itself with its type annotated. For example,
integers are given $\tpint\;\un$ type, since they are unlabeled. 
\[
\inferrule*[right=V-L-Int]{
}{ 
  D_a;F_a;\Gamma_a\vdash \inscast{n} \Rightarrow\ n @\tpint\ \un
}
\]

\begin{figure}[t!]
\flushleft 
\begin{mathpar}
\mprset{flushleft}
\inferrule*[right=L-Field-U]{
D_a;F_a;\Gamma_a \vdash \inscast{v}\Rightarrow \lv
\\ \tpof(\lv) = T\ \pol
\\\\ D_a(T) =  (\tstruct\ T \{\beta_1,\cdots, \beta_n\})
\\ \forall i\in[1,n], \pol = \labof(\inscast{\beta_i})
}{ 
 D_a;F_a;\Gamma_a ; t  \vdash   \inscast{v.i} \Rightarrow \lv.i 
}
\and
\inferrule*[right=L-Field]{
D_a;F_a;\Gamma_a \vdash \inscast{v}\Rightarrow \lv
\\ \tpof(\lv) = T\ \pol
\\\\ D_a(T) =  (\tstruct\ T \{\beta_1,\cdots, \beta_n\})
\\ \exists i\in[1,n], \pol \neq \labof(\inscast{\beta_i})
}{ 
 D_a;F_a;\Gamma_a  ; t\vdash   \inscast{v.i} \Rightarrow
 \elet\ y:T\ \bot = \erelab(\bot\Leftarrow \pol)\ \lv\
~\ein\ (y@ T\ \bot).i 
}
\and
\inferrule*[right=L-Deref]{
 D_a;F_a;\Gamma_a  \vdash   \inscast{v}  \Rightarrow \lv
\\ \tpof(\lv) = b\ \pol
}{ 
    D_a;F_a;\Gamma_a ; t  \vdash \inscast{*v} \Rightarrow 
\elet\ y:b\ \bot = \erelab(\bot\Leftarrow \pol)\ \lv\ 
~\ein\ *(y@b\ \bot)
}
\and
\inferrule*[right=L-Assign]{
 D_a;F_a;\Gamma_a  \vdash   \inscast{v}  \Rightarrow \lv
\\ \tpof(\lv) =\tptr(s)\ \pol
\\ D_a;F_a;\Gamma_a ; s\vdash e \Rightarrow \lexp
}{ 
  D_a;F_a;\Gamma_a ; t\vdash   \inscast{v:= e} \Rightarrow 
\elet\ y: \tptr(s)\ \bot = \erelab(\bot\Leftarrow \rho)\ \lv\ 
\ein\  y@\tptr(s)\ \bot :=\lexp 
}
\and
\inferrule*[right=L-If]{
D_a;F_a;\Gamma_a \vdash \inscast{v_1}\Rightarrow \lv_1
\\ \tpof(\lv_1) = \tpint\ \pol 
 \\ D_a;F_a;\Gamma_a ; t \vdash  \inscast{e_2} \Rightarrow \lexp_2
 \\ D_a;F_a;\Gamma_a ; t \vdash   \inscast{e_3}  \Rightarrow \lexp_3
}{ 
D_a;F_a;\Gamma_a ; t\vdash    \inscast{\eif\ v_1\ \ethen\  e_2\ \eelse\
  e_3}
\\\\\Rightarrow \elet\ x: \tpint\ \bot = (\erelab(\bot\Leftarrow\pol)\ \lv_1)\
\ein\ \eif\ x@\tpint\ \bot\ \ethen\  \lexp_2\ \eelse\ \lexp_3
}
\end{mathpar}
\vspace{-10pt}
\caption{Mapping of expressions}
\label{fig:inscast-exp}
\end{figure}

Expression mapping rules are listed in Figure~\ref{fig:inscast-exp}.
The tricky part is mapping expressions whose typing rules in \langname
require label comparison and join operations. Obviously, the \minic
type system cannot enforce such complex rules. Instead, we add
explicit relabeling to certain parts of the expression to ensure that
the types of the translated \minic program enforce the same property
as types in the corresponding \langname program. 


There are two rules for record field access: one without explicit
relabeling (\rulename{L-Field}) and one with
  (\rulename{L-Field-U}). Rule \rulename{L-Field} applies when all the
  elements in the record have the same label as the record
  itself. Rule \rulename{L-Field-U} explicitly
  relabels the record first, so the record type changes from $T\ \pol$
  to $T\ \bot$, resulting in the field access having the same label as the
  element.  This is because when the labels
  of the elements are not the same as the record, the
  typing rule \rulename{P-T-E-Field} will join the type of the field with
  the label of the record. However, this involves label operations,
  which \minic's type system cannot handle. \rulename{L-deref} and \rulename{L-assign} are
  similar. 
%
The mapping of if statements (\rulename{L-If}) relabels the
conditional $v_1$ to have $\tpint\ \bot$ type, so the branches are
typed under the same program counter as the if expression. 
We write $\erelab(\bot\Leftarrow \pol)$ as a short hand for a sequence of
relabeling operations $\erelab(\lab::\bot\Leftarrow
\lab_n{::}\top)\cdots\erelab(\lab_i{::}\bot\Leftarrow\lab_{i-1}{::}\top)\cdots
\erelab(\lab_2{::}\bot\Leftarrow\lab_{1}{::}\top)$ where
$\pol=\lab_1{::}\cdots{::}\lab_n{::}\lab$ and $\lab$ is either $\top$
or $\bot$. The implications of inserted relabeling
operations are discussed at the end of this section.

\vspace{5pt}\noindent{\bf Translation from \langname to \minic.} The
translation of types is shown in Figure~\ref{fig:trans-type}. It
returns a \minic type and a set of new type definitions. We use a function
$\genname(t,\pol)$ to deterministically generate a string based on $t$
and $\pol$ as the identifier for a record type. It can
simply be the concatenation of the string representation of $t$ and
$\pol$, which is indeed what we implemented for C
(Section~\ref{sec:implementation}). 

\begin{figure}[t!]
\begin{mathpar}
\inferrule*{ 
\pol\in\{\un,\bot\}
 }{
  \trans{\tpint\ \pol}_D  = (\tpint, \cdot)
}
\quad
\inferrule*{ 
\pol\notin\{\un,\bot\}\\
T = \genname(\tpint,\pol)
 }{
  \trans{\tpint\ \pol}_D  = (T, T\mapsto \tstruct\ T\  \{\tpint\})
}
\and
\inferrule*{ 
\pol\notin\{\un,\bot\}
\\ T' = \genname(T,\pol)
\\\\ T\mapsto \tstruct\ T\ \{\tau_1,\cdots,\tau_n\} \in D
 }{
  \trans{T\ \pol}_D  = (T', T'\mapsto \tstruct\ T'\ \{\tau_1,\cdots,\tau_n\})
}
\end{mathpar}
\vspace{-20pt}
\caption{Type translation}
\label{fig:trans-type}
\end{figure}

We distinguish between a type with a label that is $\un$ or $\bot$ and
a meaningful label. The translation of the type $b\ \un$ is simply
$b$. This is because $b\ \un$ is mapped from an unannotated type $b$
to begin with, so the translation merely returns it to its original
type. Similarly $b\ \bot$ is generated by our relabeling operations
during the mapping process, and should be translated to its original
type $b$. On the other hand, a type annotated with a meaningful policy
$\pol$ is translated into a record type to take advantage of
nominal typing. The translation also returns the new type
definition. This would also prevent label subtyping based on the
security lattice. However, this is acceptable given our application
domain because the labels provided by programmers are distinct points
in the lattice that are not connected by any partial order relations
except the $\top$ and $\bot$ elements. Record types are translated to
record types and types for the fields of the labeled record type
$T\ \pol$ are the same as those for $T$, stored in the translated
context $D$. This works because we assume that all labeled instances
of the record type $T$ (i.e., all $T\ \pol$) share the same
definition. 

\begin{figure}[t!]


\flushleft
\begin{mathpar}
\inferrule*[right=T-App-DE]{
\tpof(\lv_f) = (\dne)[\pc](t_1\rightarrow t_2)^{\pol_f}
\\\trans{\lv_f}_D = (v_f, D_f)
\\ \tpof(\lv_a) = b\ \pol
\\ \pol = \lab_1::\lab_2::\pol'
\\    \trans{\erelab(\lab_1::\top\Leftarrow \pol) \lv_a}_D = (e', D_1)
\\  \trans{\erelab(\lab_2::\pol'\Leftarrow \lab_2::\bot) (z@b\
  \lab_2::\bot)}_D= (e'' , D_2)
\\ \trans{t_1}_D = (\tau_1, D_3)
\\ \trans{t_2}_D = (\tau_2, D_4)
}{ 
\trans{\lv_f\ \lv_a}_D = 
 (\elet\ y:\tau_1 = e'\ \ein\ 
 \elet\ z: \tau_2 = v_f\ y\ 
\\\qquad\qquad \qquad\ein\ e'',  D_f\cup D_1\cup D_2\cup D_3\cup D_4)
}
\and
\inferrule*[right=T-ReLab-N1]{
\trans{\lv}_D  = (v, D_1)
~~ \tpof(\lv) = b\ \pol~( b~\mbox{is not a struct type})
\\ \pol'\notin\{\bot,\un\} \\\pol \notin\{\bot,\un\}
\\ \trans{b\ \pol'}_D = (T, D_2)
}{ 
   \trans{\erelab(\pol'\Leftarrow\pol)\lv}_D =(\elet\ x = v.1\ \ein\
   (T)\{x\}, D_1\cup D_2)
}
\and
\inferrule*[right=T-ReLab-N2]{
\trans{\lv}_D  = (v, D_1)
~~\tpof(\lv) = b\ \pol
~ (b~\mbox{is not a struct type})
\\ \pol'\notin\{\bot,\un\} 
\\\pol \in\{\bot,\un\}
\\ \trans{b\ \pol'}_D = (T, D_2)
}{ 
   \trans{\erelab(\pol'\Leftarrow\pol)\lv}_D = ((T)\{v\}, D_1\cup D_2)
}
\and
\inferrule*[right=T-ReLab-N3]{
\trans{\lv}_D= (v, D_1)
\\ \tpof(\lv) = b\ \pol
\\ b~\mbox{is not a struct type}
\\ \pol\notin\{\bot,\un\} \\\pol'\in\{\bot,\un\}
}{ 
   \trans{\erelab(\pol'\Leftarrow\pol)\lv}_D = (v.1 D_1)
}
\and
\inferrule*[right=T-ReLab-same]{
\trans{\lv}_D  = (v, D_1)
\\ \labof(\lv) = b\ \pol
\\ \pol,\pol' \in\{\un,\bot\} 
}{ 
   \trans{\erelab(\pol'\Leftarrow\pol)\lv}_D = (v, D_1)
}
\and
\inferrule*[right=T-ReLab-Struct]{
 \pol \notin\{\bot, \un\} ~\mbox{or}~\pol'\notin\{\bot,\un\}
\\ \tpof(\lv) =  T\ \pol 
\\ \trans{T\ \pol'}_D = (T', D_1)
\\    \trans{\lv}_D  = (v, D_2)
}{ 
\trans{\erelab(\pol'\Leftarrow\pol)\lv}_D =
\elet\ x_1 = v.1\ \ein\ \cdots \elet\ x_n = v.n\ 
\\\qquad\qquad\qquad\qquad~~\ein\
(T')\{x_1,\cdots,x_n\}, D_1\cup D_2)
}
\end{mathpar}
\vspace{-10pt}
\caption{Expression translation}
\label{fig:trans-exp}
\end{figure}

Expression translation rules
recursively translate the sub-expressions. 
We present a few interesting cases in Figure~\ref{fig:trans-exp}. 
The \minic type system is not asked to do complex label checking, so
rule \rulename{T-App-De} has to insert  label conversions. The
label of the argument is cast from $\lab_1::\lab_2::\pol'$ to
$\lab_1::\top$, as required by $f$, and the result of the function is
cast from $\lab_2::\bot$ to $\lab_2::\pol'$. These operations are
different from the ones inserted during the mapping process because they
only exist to help \minic simulate the \rulename{E-App-De} typing rule in
\langname, but do not really have declassification or endorsement effects.

Next, we explain the translation of relabeling operations. Rule
\rulename{T-Relab-N1} relabels a value whose type has a meaningful
label to one with another meaningful label.  The translated expression
is a reassembled record using the fields of the original record. Rule
\rulename{T-Relab-N2} relabels an expression with a $\un$ and $\bot$
label to a meaningful label. In this case, the translated expression
is a record. Rule \rulename{T-Relab-N3} translates an expression
relabeled from a meaningful label to a $\un$ or $\bot$ label to a
projection of the record. The next rule, \rulename{T-Relab-Same}, does
not change the value itself, because we are just relabeling between
$\un$ and $\bot$ labels. The final relabeling rule,
\rulename{T-Relab-Struct}, deals with
records. In this case, we simply return the reassembled record because
record types that only differ in labels have the same types for the
fields, as shown in the last type translation rule in Figure~\ref{fig:trans-type}.


\subsection{Correctness} 
\label{sec:trans-correctness}

We prove a correctness theorem, which states that if our
translated nominal type system declares an expression $e$ well-typed,
then the labeled expression $e_l$, where $e$ is translated from, is
well-typed under \langname's type system. Formally:

\begin{thm}[Translation Soundness (Typing)]If  $D_a;F_a;\Gamma_a; s\vdash\inscast{e}= \lexp$,
$\inscast{D_a} = D_l$, $\inscast{F_a} = F_l$, 
$\inscast{\Gamma_a} = \Gamma_l$,
$\trans{D_l} = D$,
$\trans{\Gamma_l}_D = (\Gamma, D_1)$, $\trans{F_l}_D = (F, D_2)$,
 $\trans{\lexp}_D = (e', D_3)$,
and $D\cup D_1\cup D_2\cup D_3;F;\cdot;\Gamma\vdash e': \tau$
implies $D_l;F_l;\cdot;\Gamma_l \vdash \tmof(\lexp): s$ and $\trans{s} = (\tau,\_)$
\end{thm}
Here, $\tmof(\lexp)$ denotes an expression that is the same as
$\lexp$, with labels (e.g., $@\tpint\ \un$) removed. 
The proof is by induction over the derivation of
$D_a;F_a;\Gamma_a;s\vdash\inscast{e} \Rightarrow \lexp$. The
proof can be found in Appendix~\ref{app:translation:correctness}.

It not hard to see that the translated program has the same
behavior as the original program, because they have the same program
structure except that the translated program has many indirect record
constructions and field accesses.


\subsection{Discussion}
\label{embedding:discussion}
\Paragraph{Relabeling Precision}.
It is clear from the mapping algorithm that a number of powerful
relabeling operations are added. In all cases (except the if statement) we
could do better by not relabeling all the way to bottom, but to the
label of the sub-expressions. However, that would require a heavy-weight
translation algorithm that essentially does full type-checking. 

\Paragraph{Implicit Flows}.
The security guarantees of programs that require relabeling operations
to be inserted are weakened in the sense that in addition to the
special declassification and endorsement functions, these inserted
relabeling operations allow additional observation by the attacker.
This means that the resulting program can implicitly leak information via
branches, de-referencing, and record field access.


However, for our application domain we aim
to check simple data usage and function call patterns which, as seen in our
case studies, manifest errors with explicit flows.  These
policy violations are still detected if we don't have recursive types. The
reason being those operations only cause relabeling of a
\emph{smaller} type. The API sequences keep the same basic type with
changing labels. If we have recursive types, the above argument would
be invalid. See the following example.
\[\begin{array}{ll}
y: & \tstruct\ T\ \{ \tstruct\ T\ (\bot_s,\flnTag{EncodedBal})::\bot, \tpint\} \\
     &   ~~(\flnTag{AlicePrivate}, \bot_I)::
       (\bot_s,\flnTag{EncodedBal})::\bot
\end{array}
\]
$y.1$ will have the same effect as \lstinline{encodeA}, which violates the API
sequence that we try to enforce using these types. Note that C doesn't
allow this type, but we could use pointers to construct something quite
similar. In our case studies, we do not have such interaction between
policies and recursive types.


\section{Implementation}
\label{sec:implementation}
We explain how
the annotations and translation algorithms of \sysname
are implemented for C. 

\Paragraph{Translation of annotations for simple types.}
Utilizing C's nominal typing via the \lstinline{typedef} mechanism is key to realizing
\langname type system within the bounds of C's type system. The
declaration of the \langname 
type $t\ \pol$ in C will be: \lstinline[mathescape=true]|typedef struct {$t$ d;} $\tJoin{\rho}{t}$;|
Here $\tJoin{\rho}{t}$ is a string representing the type $t\ \pol$
and it is simply a concatenation of the string representation of the
policy $\pol$ and the type $t$. Consider the annotated code snippet.
\begin{lstlisting}
	#requires l1:secrecy then l2:secrecy
	int x;
\end{lstlisting}
In \langname, the type of $x$ is $\m{int}\ (\flnTag{l1}, \bot_I)::
(\flnTag{l2}, \bot_I)::\bot$. The generated C typedef is:
 \lstinline[mathescape=true]|typedef struct {int d;} l1S_l2S_int;|.
This definition contains the original type, which
allows access to the original data stored in $x$ in the transformed program. 


\Paragraph{Structures and unions.}
We allow programmers to annotate structures in two ways: an instance of a 
structure can be annotated with a particular policy, or individual
fields of an instance of a structure can be given annotations.
The names of structures hold a particular significance within C since
they are nominal types, and thus, they need to be properly handled.
Unions are treated in a parallel manner, so we omit the details.

A policy on an instance of a structure is annotated and translated following the same formula as
annotations on simple C types. Suppose we have the following
annotation and code.
\begin{lstlisting}
	#requires l1:secrecy then l2:secrecy
	struct foo x;
\end{lstlisting}
\sysname will produce the following generated type definition:
\lstinline[mathescape=true]|typedef struct {struct foo d;}|\\
\lstinline[mathescape=true]|l1S_l2S_foo;|.
This is different from the algorithm in
Section~\ref{sec:translation}, where structures are not nested and annotations
are applied to structure definitions rather than instances.
This is done in the implementation because the definition of \lstinline{foo}
might be external and therefore may not be known to the translation
algorithm, so we simply nest the entire structure inside.

The second method allows annotations on particular
fields of the structure as follows below.
\begin{lstlisting}
	#requires {f1:int, f2:int} l1:secrecy then l2:secrecy
	struct foo x;
\end{lstlisting}
The following type definition will be generated. 
\begin{lstlisting}[language=C, tabsize=2]
typedef struct { 
	l1S_l2S_int f1; l1S_l2S_int f2; foo d;
} l1S_l2S_foo;
\end{lstlisting}
Fields that have policy annotations are fields of the new struct. To
allow access to other fields in the original struct, a copy of the
original struct is nested inside this new struct. This is for the same
reason as the structure nesting in the previous case. 

 
Finally, we explain how member accesses are handled. Suppose a struct
\lstinline{foo} has members \lstinline{f1} and \lstinline{f2}, and an
annotation of policy \lstinline{p} has been placed on member \lstinline{f1},
but no annotation has been placed on member \lstinline{f2}. The
generated type definition for the structure is as follows:
 \lstinline[mathescape=true]|typedef struct { p_int f1; foo d; } p_foo;|.

Assume \lstinline{x} has type \lstinline{p_foo}. Access to \lstinline{f1} is still
\lstinline{x.f1}, since there is a copy of it in \lstinline{x}. Access to \lstinline{f2} is
rewritten to \lstinline{x.d.f2}. The field initialization is rewritten
similarly: \lstinline[mathescape=true]|foo x={.f1=1,.f2=2};| 
is transformed to this: \lstinline[mathescape=true]|foo x={.f1=1,.d={.f2=2}};|

\Paragraph{Pointers.}
%
We provide limited support for pointers.
Below is an example of how
annotations on pointers are handled. 
\begin{lstlisting}
	#requires AlicePriv:secrecy
	int* x;
\end{lstlisting}
The translated code is below; a type definition of struct
\lstinline{AlivePrivS_int} is generated:  \lstinline[mathescape=true]|AlicePrivS_int* x;|
The following function can receive \lstinline{x}  as an argument because the annotation
for its parameter matches that of \lstinline{x}. 
\begin{lstlisting}
	#param AlicePriv:secrecy
	int f(int* x) {...}
\end{lstlisting}

The annotation for pointers only annotates the content of the
pointer. Even though \langname allows policies on the pointer
themselves, we did not implement that feature. We also 
do not support pointer arithmetic, which is difficult to handle for many static
analysis tools, especially lightweight ones like ours. However, our system will flag 
aliasing of pointers across mismatched annotated types.  Our system will also
flag pointer arithmetic operations on annotated types as errors. Programmers can
encapsulate those operations in trusted functions and annotate
them to avoid such errors.

\Paragraph{Typecasts.} The C type system permits typecasts,
allowing one to redefine the type of a variable in unsound ways. Casting
of non-pointer annotated types will be flagged as an error by \sysname. This is because
our types are realized as C structures; type checkers do not allow
arbitrary casting of structures. However, our tool cannot catch typecasts 
made on annotated pointers; a policy on a pointer will be lost if a
typecast is performed.

\label{app:implementation}

\Paragraph{Void.}
\label{app:void}
In this section we will discuss the
handling of functions that have a \lstinline{void} return type. We disallow
the use of the \lstinline{#return} annotation with such
functions. The reasons for doing so will be explained below.
Given that translation and the general purpose of the \lstinline{void} type, it is clear that allowing 
an integrity annotation of a function with a \lstinline{void} return type is not valid. 
Consider the following example:
\begin{lstlisting}
	#return trusted:integrity
	void func() {...}
\end{lstlisting}

If we allowed this translation to proceed naively, the translated version of the code could look like this:
\begin{lstlisting}
	trustedI_void func() {...}
\end{lstlisting}

This is invalid for two reasons. First, as mentioned before this function is not returning anything and
 therefore an annotation on its return type is meaningless. Second, as this translation evidences, if we were to
 allow such an annotation, we would have created an invalid type, ``\lstinline{trustedI_void}". This type is 
invalid because, in order for it to be used in our annotation system, we need to generate functions that 
perform the relabeling operations to and from this type. However, no such operations can be generated,
as they would effectively take nothing and endorse it to a trusted type. 

Another case where \lstinline{void} comes into play is in implicit void pointer conversion. In the case
where a void pointer is being passed to a function for an annotated parameter, this will not be flagged
as an error by our system.

\Paragraph{Variadic Functions.}
We provide partial support for annotations on variadic functions. For example, with the following
function:
\begin{lstlisting}
	int f(int a, int b, ...) { ... }
\end{lstlisting}

Only the first two arguments can have annotations.

\Paragraph{Builtin Qualifiers}
Qualifiers are subsumed into the ``original type" that our processing algorithm extracts 
from the source code. For instance, if we encounter the code:
\begin{lstlisting}
	#requires test:secrecy
	volatile int x;
\end{lstlisting}

\noindent the qualifier \lstinline{volatile} will be considered to be part of the base 
type ``\lstinline{int}". Thus, the translation of the code will be:
\begin{lstlisting}
	__fln__testS_volatile_int x;
\end{lstlisting}

\noindent Rather than:
\begin{lstlisting}
	volatile __fln__testS_int x;
\end{lstlisting}

This approach generalizes to multiple qualifiers on a type.

\Paragraph{Builtin Operators.}
The labels we can add through our system are sometimes applied to
variables with numeric types, e.g. \lstinline{int}, \lstinline{float}, \lstinline{double},
etc. Binary and unary operations on these types are directly supported
by C. After transformation arithmetic operations do not work out of
the box on our transformed types. For instance, \lstinline{x+y} will raise a
type error if \lstinline{x} and \lstinline{y} are annotated because \lstinline{+} 
is being applied to a struct, not
an \lstinline{int}. Programmers would need to define a \lstinline{plus} function for the
annotated type to circumvent this issue.

\Paragraph{Code Generation.}
In addition to the above remarks on how specific C features are handled, we need 
to do some additional code generation and program reconstruction in order for our system
 to be straightforward for the end user to use. When processing a directory of annotated 
source files that includes one ``root" file (typically the file containing the main function), 
our system does the following:
\begin{enumerate}
 \item Recursively find and parse included files from the root
 \item Gather annotations from each file
 \item Generate header definitions for each file
 \item Stitch together the original and generated files
\end{enumerate}

Next, we explain two pieces of this process;
header generation and program reconstruction.

\Paragraph{Header Generation.} ~\\
\label{app:impl:headergen}
Header generation refers to the phase of the program transformation when all of the structure
 and function definitions for the annotated types in a particular file are generated. The generated 
structure and function definitions are collected into a single header file that is included where
 its definitions are needed during the program reconstruction phase.

To explain how the structure and function definitions are generated, let us consider the following code:
\begin{lstlisting}
	#requires AlicePriv:secrecy
	int x;
\end{lstlisting}
Previously, we explained that for a variable definition of the form $\tau~x;$ annotated with a policy 
$\rho$ we need to generate a type $\tJoin{\tau}{\rho}$. In our example, this generated type would 
be \lstinline{AlicePrivS_int}. As we explained before, to give this type concrete meaning within the C type 
system, we instantiate it in the form of a \lstinline{typedef struct}:
\begin{lstlisting}
	typedef struct {int d;} AlicePrivS_int;
\end{lstlisting}
This generated structure contains the original type as a member and interacts with the code as
 described in the subsection on structures (in section \ref{sec:implementation}).

In order to be able to convert between the original type \lstinline{int}, which we call the base 
type, and this new ``type" \lstinline{AlicePrivS_int}, which we call the policy type, two functions 
need to be generated:
\begin{lstlisting}
	privateS_int privateS_int_w(int x) {...};
	int privateS_int_r(privateS_int x) {...};
\end{lstlisting}
The first function, given a regular integer will relabel the integer to the type \lstinline{AlicePrivS_int}. 
The second function, will relabel \lstinline{AlicePrivS_int} back to a regular integer.

Thus, we have the basis for what our header generation needs to accomplish. Each annotated type
 $\tJoin{\tau}{\rho}$ can be viewed as a pair (base type, policy type). For each pair we must:
\begin{enumerate}
  \item Generate a typedef structure that has a base type member and is named $\rho\_\tau$
  \item Generate a function from the base to the policy type
  \item Generate a function from the policy to the base type
\end{enumerate}
In order to prevent the duplication of generated structure or function definitions, we 
deduplicate the list so that it consists of only unique pairs.

\Paragraph{Program Reconstruction.}
\label{app:impl:progrecon}
During the program reconstruction phase, header files that have been generated must be included
 at the right points in the program's dependency graph. If they are not included at the right points, then
 it is possible that a file containing transformed code that makes use of the generated structures and
 functions will be missing the definitions of those structures or functions and thus will not be compilable.
 In order to solve this issue, we recursively traverse the dependency graph starting from the root file. At
 each file that we visit in the graph, we include the generated header file containing the generated
 structures and definitions. 
 
\Paragraph{Pragmas.}
We have presented annotations without the \lstinline{pragma} directive prefixing them for convenience
of presentation. When using the actual implementation of \sysname we write, for instance,
\lstinline{#pragma requires AlicePriv:secrecy}. The use of the pragma directive allows C compilers to ignore
our annotations, thus allowing developers to keep annotations in their codebases without the annotations
interfering with normal compilation of the program.

\eat{
\Paragraph{Limitations.} As previously mentioned, we do not handle
pointer arithmetic. We only provide limited support for function pointers. We do
not support C's builtin operators, such as the unary
\lstinline{++}. We do not support typecasts on pointers, nor can we flag violations
due to implicit void pointer conversion. We provide partial support for variadic functions.
Finally, we do not support using \lstinline{#return} with a
function that has a \lstinline{void} return type. 
These are careful design choices we made so our tool is
lightweight and remains practical; we emphasize that our tool
is not meant for verification.
Limitations of our type system can be found in Section~\ref{embedding:discussion}.
}


\section{Case Studies}
\label{sec:case-studies}

We evaluate the effectiveness of \sysname at discovering violations of 
secrecy, integrity, and sequencing API usage policies 
on several 
open-source cryptographic libraries.
Our results are summarized in Figure~\ref{fig:results}.
We examine: Obliv-C, a compiler for dialect of C directed
at secure computation \cite{oblivc, Zahur:2015}; SCDtoObliv, a set of floating 
point circuits synthesized into C code 
\cite{scdtoobliv}; the Absentminded Crypto Kit, a library of Secure
Computation protocols and primitives \cite{absentminded,
  Doerner:2017}; Secure Mux, a secure multiplexer application
\cite{Pool:Zhu:2017}; the Pool Framework, a secure computation memory management library
\cite{poolframework, Pool:Zhu:2017}; Pantaloons RSA, the top
GitHub result for an RSA implementation in C \cite{githubrsa}; MiniAES, an
AES multiparty computation implementation \cite{miniaes,
  Damgard:2016}; Bellare-Micali OT, an implementation of the
Bellare-Micali oblivious transfer protocol \cite{Bellare:1989}; Kerberos ASN.1 Encoder,
the ASN.1 encoder module of Kerberos \cite{krbasn1}; Gnuk OpenPGP-do, a portion of the OpenPGP
module from gnuk \cite{gnuk}; Tiny SHA3, a reference implementation of SHA3 \cite{tinysha3}.
%
%
\begin{figure*}[tbp]
\centering
\begin{tabular}{ | l | c | c | c | c | c | r | c | c | }
\hline
	 \textbf{Library} & \textbf{\# Policies} & \textbf{Sec.} & \textbf{Int.} & \textbf{Seq.} & \textbf{LoA} & \textbf{$\sim$ LoC} & \textbf{Issues} & \textbf{Runtime (s)} \\
\hline
	 Obliv-C Library						& 2 & 1 & 1 & 0	 	& 11		& 80 		& 0 	& 0.04 \\
	 SCDtoObliv FP Circuits				& 4 & 4 & 0 & 0	 	& 10		& 43,000 	& 1 	& 5.55 \\
	 ACK Oqueue						& 7 & 7 & 7 & 2	 	& 19		& 700	& 0	& 0.17 \\
	 Secure Mux Application				& 4 & 3 & 4 & 0	 	& 11		& 150	& 0 	& 0.06 \\
	 Pool Framework					& 4 & 2 & 4 & 0	 	& 8		& 500	& 1 	& 0.16 \\
	 Pantaloons RSA					& 5 & 2 & 3 & 0		& 12		& 300	& 1 	& 0.11 \\
	 MiniAES							& 9 & 4 & 4 & 1	 	& 13		& 2000	& 0 	& 0.08 \\
	 Bellare-Micali OT					& 5 & 3 & 2 & 0		& 12		& 100	& 2 	& 0.05 \\
	 Kerberos	ASN.1 Encoder				& 2 & 2 & 0 & 1		& 8		& 300	& 0   & 0.12 \\
	 Gnuk OpenPGP-do					& 5 & 0 & 5 & 1		& 11		& 250	& 1   & 0.10 \\
	 Tiny SHA3						& 3 & 3 & 0 & 1		& 6		& 200	& 0   & 0.10 \\	 
\hline
\end{tabular}
\caption{Evaluation Results. Sec, Int, and Seq are the number of secrecy, integrity, and sequencing policies. LoA is lines of annotations, LoC is the lines of code. }
\label{fig:results}
\end{figure*}
We determine application-specific policies
and implement them with our annotations.

\subsection{SCDtoObliv Floating Point Circuits}
First, we show that \sysname can be used to discover flaws
in large, automatically generated segments of code that would be
very difficult for a programmer to manually analyze.

%

SCDtoObliv \cite{scdtoobliv} synthesizes floating point circuit in C via
calls to boolean gate primitives implemented in C.
While this approach produces performant floating point
circuits for secure computation applications, the resulting circuit files are hard to
interpret and debug.
%
The smallest of these generated circuit files is around 4000 lines
of C code while the largest is over 14,000 lines. 
We annotate particular wires based on the circuit function
to check that particular invariants such as which bits should be used in
the output and which bits should be flipped are maintained.
%

\sysname uncovered a flaw in the subtraction circuit. The
Obliv-C subtraction circuit actually uses an addition circuit  
to compute $A + (-B)$. The function
that does the sign bit flipping, \lstinline{__obliv_c__flipBit}, is annotated
so that it can only accept an input with the $\flnTag{needsFlipping}$ label
as follows. 

\begin{lstlisting}
	#param needsFlipping:secrecy
	void __obliv_c__flipBit(OblivBit* src)
\end{lstlisting}


Our tool reports an error; rather than the sign bit
of the second operand being given to
\lstinline{__obliv_c__flipBit} the sign bit of the \emph{first}
operand was given to \lstinline{__obliv_c__flipBit}.
Instead of computing $A + (-B)$ the
circuit computes $(-A) + B$; the result of evaluating the circuit
is negated with respect to the correct answer.

\subsection{A Potential Flaw in the Pool API}
This case study is based on Pool, a Secure Computation
tool~\cite{poolframework, Pool:Zhu:2017} and  demonstrates that \sysname
can help identify cross-module API constraints. 

The Pool framework
provides a set of APIs for users, some of which take function pointers
as arguments. As a result, user-provided functions are called inside
Pool APIs and interact with sensitive data from the framework.
The following function pointer is used-accessible.
\begin{lstlisting}
void (*Gate_Copy)(_, _, _, uint64_t indexs, _)
\end{lstlisting}
We have left most of the parameters opaque as they are unimportant
to the flaw we discovered. According to the signature, the function
pointed to by this pointer can accept \textit{any} unsigned 64-bit integer as its
fourth parameter (an index to a gate used by the Pool API). 

We would like to check the property that only valid gates are being 
used in the protocol execution
and that only trusted functions can use valid gates. We use the
label $\flnTag{valid\_gate}$ as both a secrecy and an integrity policy
to prevent APIs from using invalid gates and untrusted functions from
using valid gates.
Here is an example of that
annotation on a function that is said to produce a valid
gate: 
\begin{lstlisting}
	#return valid_gate:(secrecy, integrity)
	uint64_t Next_Gate_in_Buffer(Pool *dst)
\end{lstlisting}

An error is reported for the following code. 
\begin{lstlisting}
(*(P->Gate_Copy))(_, _, _, 
                  Next_Gate_in_Buffer(P), _);
\end{lstlisting}
Notice that the fourth argument of the \lstinline{Gate_Copy} function is
returned by the \lstinline{Next_Gate_in_Buffer} function. The flaw is
caused by the fact \lstinline{Gate_Copy} is not trusted to take a valid
gate as input, as far as can be told by its type and the project's documentation \cite{pooldocs}. 
This error is similar to bugs found in kernels that give user-supplied callback functions
private kernel data. To allow the translated code to compile, we would
have to explicitly add an annotation to the \lstinline{Gate_Copy} function to
allow it to take a valid gate as input. By doing so, we are knowingly
endorsing potentially dangerous user-supplied callback functions.





\subsection{Gnuk OpenPGP-DO}
The last case study shows that \sysname can uncover a
previously known and patched null-pointer dereferencing bug and another
potential bug in the gnuk OpenPGP-DO file, which handles OpenPGP Smart
Card Data Objects (DO). We explain the latter in the next subsection.

The function \lstinline{w_kdf} handles the reading 
or writing of DOs that support encryption via a Key Derivation
Function (KDF) in the OpenPGP-DO file. 
\begin{lstlisting}
static int rw_kdf (uint16_t tag, int with_tag, 
   const uint8_t *data, int len, int is_write)
\end{lstlisting}

If the data is being read, it is copied out to a buffer via the function \lstinline{copy_do_1}:
\begin{lstlisting}
static void copy_do_1(uint16_t tag, const uint8_t *do_data, int with_tag)
\end{lstlisting}

One invariant is that the \lstinline{do_data} pointer must point to a valid 
segment of data; it must not be null. We provide the following annotation:
\begin{lstlisting}
#param(2) check-valid-ptr:integrity
static void copy_do_1(uint16_t tag, const uint8_t *do_data, int with_tag)
\end{lstlisting}

This annotation states that the second parameter will only be accepted
if it has been endorsed by a function that returns data annotated with
the $\flnTag{check\_valid\_ptr}$ label. We provide such a function and
rewrite all nullity checks to use it.
\begin{lstlisting}
#return check_valid_ptr:integrity
const uint8_t *check_do_ptr(const uint8_t *do_ptr)
\end{lstlisting}

Returning back to the \lstinline{rw_kdf} function, when data is being read, the following call 
of \lstinline{copy_do_1} occurs: 
\begin{lstlisting}
copy_do_1(tag, do_ptr[NR_DO_KDF], with_tag);
\end{lstlisting}
%
Compilation of the transformed code results in this error:
\begin{lstlisting}
error: passing argument 2 of 'copy_do_1' from incompatible pointer type [-Werror=incompatible-pointer-types]
       copy_do_1(tag, do_ptr[1], with_tag);
                      ^~~~~~
\end{lstlisting}
The issue is \lstinline{copy_do_1} is annotated to require
a null-pointer check for parameter two, but that check was not performed.




\subsection{Length Check in Gnuk OpenPGP-DO}
We now demonstrate the discovery of a potential issue with the gnuk \lstinline{copy_do_1} 
function. 

This utility function is responsible for performing a properly sized \lstinline{memcpy} 
given a data array, in the format of a Tag-Length-Value data structure,
that contains the data to by copied as well as metadata such as the size 
of the data to be copied. We focus our analysis on the size metadata, which is captured by 
the variable \lstinline{int len}. We provide the following annotation:

\begin{lstlisting}
	#return check_len:integrity
	int len;
\end{lstlisting}

The purpose of this annotation is to ensure that this length variable will be checked before 
it is given to \lstinline{memcpy} to prevent a buffer overflow. 

The \lstinline{copy_do_1} 
function does two slightly different things depending on the value of a conditional. In the 
first case, the array element \lstinline{do_data[0]} is checked to not exceed its maximum 
size before it is assigned to \lstinline{len}. In the second case, however, no check is made.

Thus, a potential faulting path exists: if the conditional is false and \lstinline{do_data[0]}
was previously assigned a negative value causing an overflow, when \lstinline{len} is used
as the size argument to \lstinline{memcpy}, it could read past the end of the \lstinline{do_data}
array as it may not be null-terminated.

Our system alerts us to this issue:
\begin{lstlisting}
evaluations/gnuk/openpgp-do_snip__fln.c:301:9: error: incompatible types when assigning to type `__fln__check_lenI_int {aka struct <anonymous>}` from type `uint8_t {aka int}`
     len = do_data__fln_p[0];
         ^
\end{lstlisting}

We contacted the maintainer of the library who assured us that every instantiation of the 
\lstinline{do_data} array has the correct length and thus the potential issue we describe
cannot come up in practice. However, we believe that addition of a check that would fulfill
the policy we have described could be useful should a mistake be made with a
\lstinline{do_data} array.

\label{app:case-studies}
\subsection{Secure Multiplexer Application}
Pool is a secure computation framework that was released by Zhu et al. \cite{poolframework, Pool:Zhu:2017}. 
The authors provide an example application, a secure multiplexer, that makes use of the framework. We evaluate
 this application to check that the Pool API usage does not violate the secrecy or integrity properties of the garbler's 
or evaluator's data. 
We check first that the secrecy and integrity of each party's private data is maintained.
\begin{lstlisting}
	#requires AlivePriv:(secrecy, integrity)
	bool* inputA;
\end{lstlisting}

At the next step of the protocol, Alice's input is assigned her private value by way of a helper function
\begin{lstlisting}[language=C, tabsize=2]
	inputA = int2bitsA(0x01AA);
\end{lstlisting}

Given that the \lstinline{int2bitsA} function is Alice's way of assigning a value to her input, we accordingly annotate 
that it is trusted to provide integrity for the \lstinline{AlicePriv} label:
\begin{lstlisting}
	#return AlicePriv:integrity
	bool* int2bitsA(int x) {...}
\end{lstlisting}

On the side of the other party, Bob, parallel annotations can be made.
Since only the functions \lstinline{int2bitsA} and \lstinline{int2bitsB} can provide an integrity endorsement to the \lstinline{AlicePriv} and \lstinline{BobPriv} respectively, our system can check that no other code will 
modify Alice and Bob's private input.

The next annotation we provide is a check on the data structure entities representing Alice and Bob. 
Alice is an instance of a Garbler structure and Bob is an instance of an Evaluator. 
Thus we provide 
a label \lstinline{GarblerProtected} and apply both its secrecy and integrity projections to the Alice instantiation 
of Garbler:
\begin{lstlisting}
	#requires GarblerProtected:(secrecy, integrity)
	Garbler alice;
\end{lstlisting}
All Pool framework functions that need to access the Garbler's (and respectively, the Evaluator's) data thus need 
to be trusted to maintain the secrecy of the Garbler's data.
Thus, the following annotations are 
applied:
\begin{lstlisting}
	#param(1) GarblerProtected:secrecy
	#param(2) PreparedFunction:secrecy
	wire** execA(Garbler alice, wire** func, 
		wire** inpt) {...}
\end{lstlisting}

The annotation \lstinline{GarblerProtected} makes it clear that this function is trusted to read the Garbler 
structure. The annotation \lstinline{PreparedFunction} has not been explained before. Its role is specify that the 
function pointer \lstinline{wire** func} must point to a function that fulfills the policy \lstinline{PreparedFunction}.
No policy violations were found.

\subsection{Checking Initializations in Pool}
Another annotation we provide adds checks to prevent users of the Pool framework from omitting initializations. 
It is an integrity endorsement:
\begin{lstlisting}
	#return initialized_pool:integrity
	Pool* SetupPool(Pool *dst ...);
\end{lstlisting}
The reason for adding this annotation is that as the original framework code stands, there are no checks in 
functions that use the Pool structure that it is actually properly initialized. If a function uses an uninitialized Pool 
structure, the protocol evaluation could fail through an exception or could have some other undesirable behavior 
that may leak information to an attacker. By adding the above annotation as well as annotations of the form 
\lstinline{#param(i) initialized_pool:integrity} to each of the functions that uses the Pool, we are 
able to statically check for cases where an uninitialized Pool structure is used.

A similar annotation checks for initialization of the \lstinline{ServiceConfig} structure.
\begin{lstlisting}
	#return initialized_service:integrity
	ServiceConfig* SetupService(Pool *dst ...);
\end{lstlisting}
We add corresponding annotations to each function that uses the $\m{ServiceConfig}$ to only accept an initialized configuration. 

\subsection{Obliv-C Library}
\label{sec:obliv-c-lib}
We demonstrate checking a secrecy property. The annotation we provide is a 
\lstinline{oblivious} label. This policy is added to the \lstinline{OblivBit}
structure in the Obliv-C library \cite{oblivc}. The secrecy label is used to check
that oblivious data is only being handled by functions that are
trusted not to leak information about the oblivious data within the
Obliv-C library. 
The integrity label is needed to check that only trusted APIs are
allowed to generate oblivious data and update oblivious data
structures. 

We add annotations the \lstinline{OblivBit} data structure as follows:
\begin{lstlisting}
	#requires oblivious:secrecy
	OblivBit* data;
\end{lstlisting}
Functions that are trusted to process oblivious data are given an
 annotation that it is allowed to accept the oblivious data as an
 argument. See the example below.
\begin{lstlisting}
	#param oblivious:secrecy
	void __obliv_c__copyBit(OblivBit* dest,
		const OblivBit* src)
\end{lstlisting}
The use of this secrecy label also enforces the integrity
of oblivious data structures. This is because
unannotated data is assumed to have the special label $\un$, 
so it cannot be used to update structures storing data labeled with $\flnTag{oblivious}$.
We did not find any policy violations in the Obliv-C library.

\subsection{Kerberos ASN.1 Encoder}
This case study concerns enforcing an API sequencing policy in
a widely-used open-source program, Kerberos. More concretely, we consider
the Kerberos ASN.1 Encoder which makes use of two functions \lstinline{free_atype} and
\lstinline{free_atype_ptr} that work in tandem to free
memory allocated to Kerberos C objects. Objects must first be freed by the \lstinline{free_atype}
function before they are freed by the \lstinline{free_atype_ptr} function. We provide annotations for these
functions to check for violations of this sequenced behavior.

The \lstinline{free_atype} function takes as an argument a pointer to an object along
with the struct \lstinline{atype_info} containing a description of the object. We modify
the function to return this \lstinline{atype_info} struct.
\begin{lstlisting}[language=C, tabsize=2]
	const struct atype_info* free_atype(const struct atype_info *a, void *val)
\end{lstlisting}

In the function body, the appropriate freeing routine is called based
on \lstinline{atype_info}'s \lstinline{type} member. The freeing
routine can take the form of recursive calls to
\lstinline{free_atype}, calls to other specialized freeing functions,
or calls to the second freeing function \lstinline{free_atype_ptr}:

\begin{lstlisting}
	static void free_atype_ptr(const struct atype_info *a, void *val)
\end{lstlisting}

This function is constructed similarly to \lstinline{free_atype} except that it works only 
over pointer-type objects and only recursively calls itself.

We add the following annotations to those functions:
\begin{lstlisting}
#param(1) freebase:secrecy
#return freeptr:secrecy
const struct atype_info* free_atype(const struct atype_info *a, void *val)
@$\dots$@
#param(1) freeptr:secrecy
static void free_atype_ptr(const struct atype_info *a, void *val)
\end{lstlisting}

We add the following annotation to \lstinline{atype_info} structs:
\begin{lstlisting}
#requires freebase:secrecy then freeptr:secrecy
const struct atype_info* x;
\end{lstlisting}
The annotations above will check that the calling sequence invariant is maintained; 
no violations were found.

\subsection{Oblivious Queue Data Structure}
\sysname can be used to check granular
invariants of data structures. This case study
emphasizes the modularity of our approach. 
The case study is on an oblivious queue (oqueue) library~\cite{absentminded}.
The data structure is hierarchical and operations on this data structure
should maintain the following invariants \cite{Circuits:Zahur:2013}:
(1) The buffer at level $i$ has $5 \times 2^i$ data blocks. (2) The
number of non-empty blocks at buffer level $i$ is a multiple of
$2^i$. (3) Each level maintains a counter storing the next available
empty block. (4) When the buffer at level $i$ is full the last block is shifted down to level $i+1$.

Invariants (1) and (3) can be violated through incorrect modification
to the counter or the oqueue, so we should check that
modifications are only done by trusted functions. Therefore, we use the labels
$\flnTag{push\_protect}$,
$\flnTag{pop\_protect}$, and $ \flnTag{oqueue\_tail}$. 
To modify where the next element is placed in the oqueue, only functions
that are trusted to modify data labeled with $\flnTag{oqueue\_tail}$
can do so. Likewise, the binary counters \lstinline{push_time} and
\lstinline{pop_time} should only be modified within the context of the
push and pop operations. 

Each field is given an integrity label to protect its access. 
One example annotation on the oqueue data structure is:
\begin{lstlisting}
	#requires {.push_time:int} push_protect:integrity
	oqueue* this_layer;
\end{lstlisting}

Invariants (2) and (4) are checked at run time by conditional statement in the
API code. We add two sets of annotations (symmetric
for the push and pop functions) to model the checks for
ensuring that data is shifted to a lower level or raised to a
higher level in the queue when the current oqueue level is full or
empty.  We use the following label sequence policy: 
$\flnTag{oqueue\_has\_child} \rightarrow
  \flnTag{oqueue\_push\_ready}$. 
Considering just the conditional push case below,
these labels form an endorsement sequence on the oqueue data
structure. First, we endorse that the oqueue has a child via the
\lstinline{has_child} helper function that can check for the existence of a
child, then we endorse that the oqueue is ready to be pushed to via the
\lstinline{is_push_time} helper function that is trusted to access the
oqueue's $\flnTag{push\_protect}$-labeled variable. 
\begin{lstlisting}[language=C, tabsize=2]
// oqueue @$\rightarrow$@ oqueue_has_child
layer1 = has_child(layer);
if (layer1) {
	// oqueue_has_child @$\rightarrow$@ oqueue_push_ready
	layer2 = is_push_time(layer1);
	if (layer2) {
		// oqueue_push_ready @$\rightarrow$@ oqueue
		layer3 = tail_is_full(layer2);
...
\end{lstlisting}
This illustrates an
instance of a compositional check that our annotations are
providing; not only are we checking for the existence of a particular
endorsement sequence, but we also check that along the way, the
functions that act on our oqueue to provide those endorsements are only
the functions that we trust. Finally, if all of these conditions are
met, data is allowed to be shifted down to the lower level of the
oqueue.

For the case of the conditional pop, the sequence of endorsement
operations is similar, but we have another, higher-level,
compositional guarantee. We add an annotation to the pop function
itself labeling its \lstinline{layer} parameter (which is an instance of
the oqueue) as $\flnTag{oqueue\_check\_empty}$. Thus, we enforce that
only an oqueue that has been checked for emptiness can be used with
the pop operation. 

This also demonstrates the modularity of \sysname. We are able to
provide annotations at many ``levels" of the source code; in the above
example there is a general check that the oqueue is
non-empty before the \lstinline{oqueue_pop} function is entered. Then
within the body of the \lstinline{oqueue_pop} function there are additional
annotations that ``refine" our knowledge about the state of the
oqueue. These functions could come from the same library or across
several libraries from different developers. \sysname allows
policies to be collectively checked across different modules.

\subsection{Performance Evaluation}

\begin{figure}
	\includegraphics[width=0.5\textwidth]{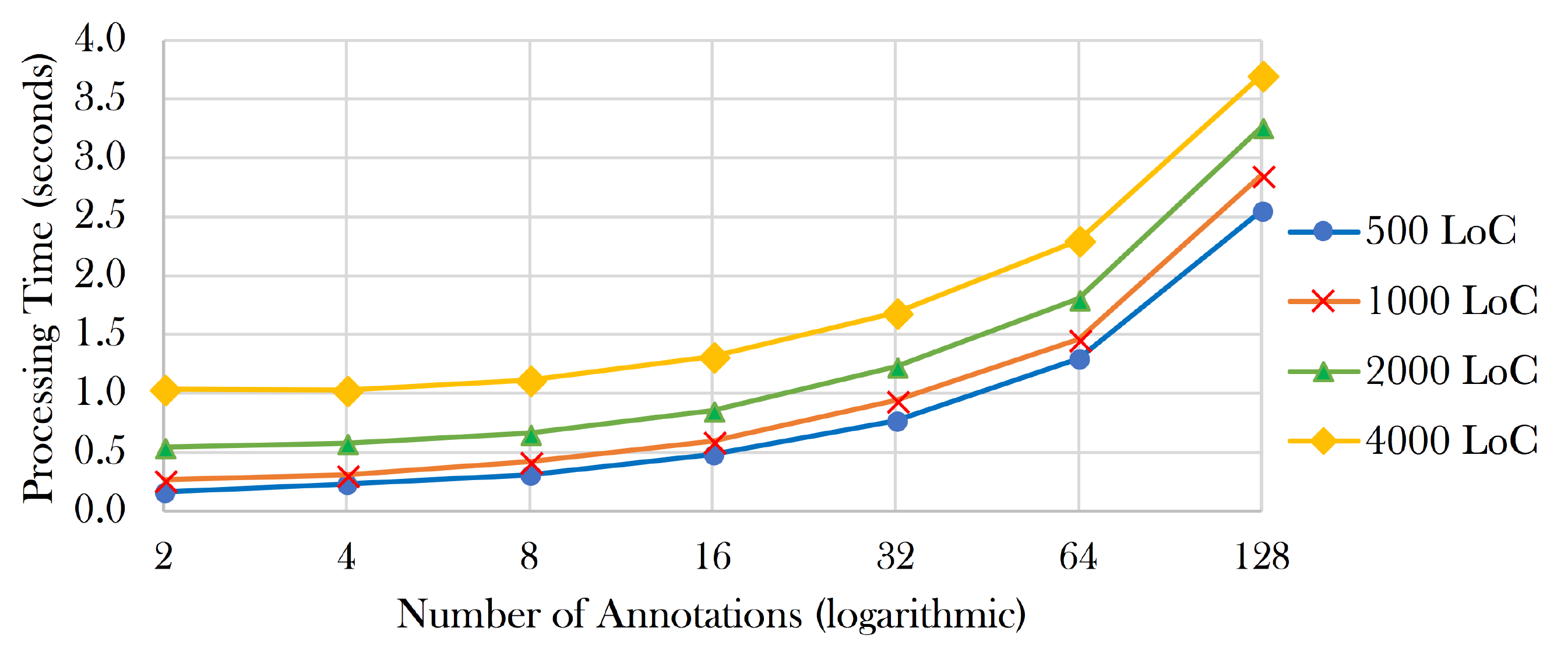}
	\includegraphics[width=0.5\textwidth]{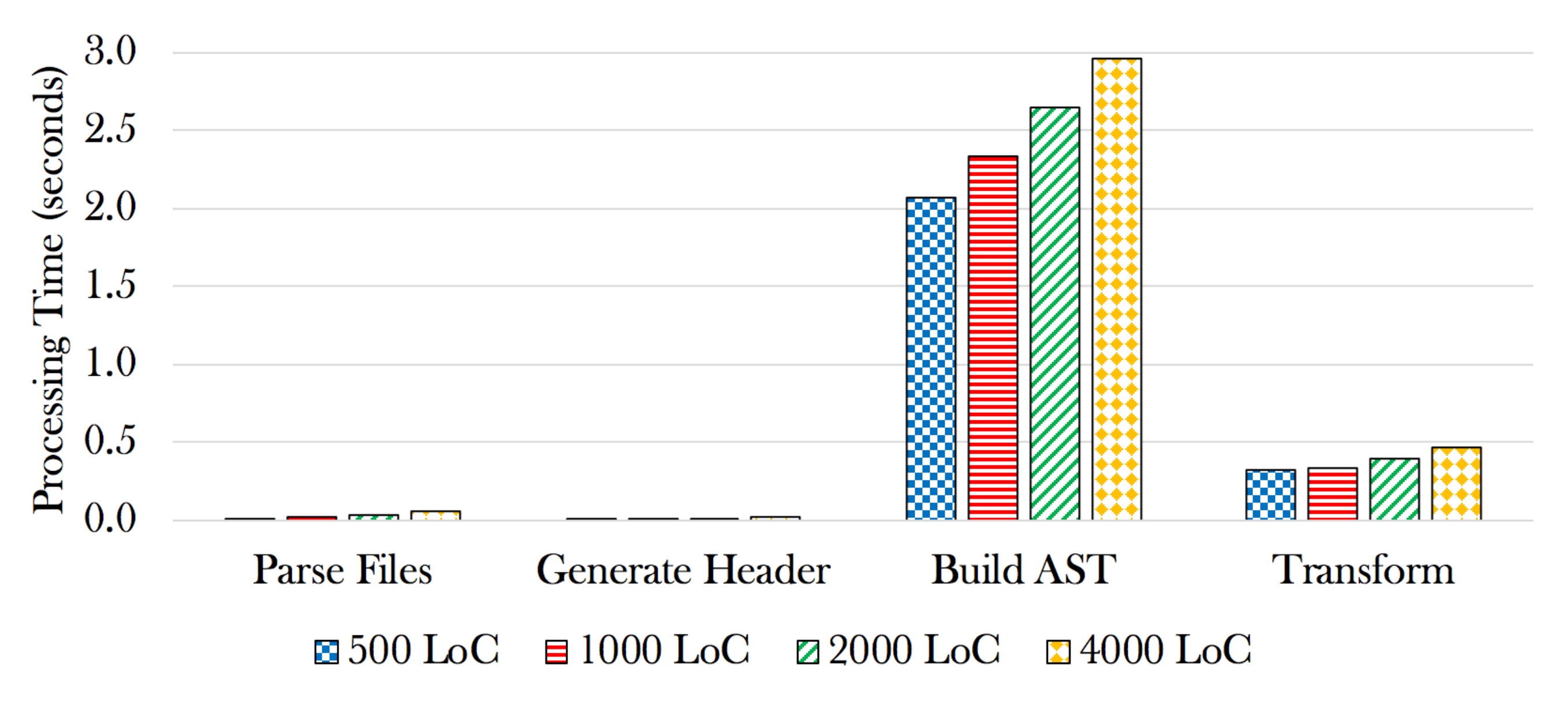}
	\caption{Processing Time vs Number of Annotations and Time Per Processing Stage}
	\label{rt-vs-loa}
	\label{t-vs-st}
\end{figure}

We evaluate the performance of \sysname on synthetically generated C
programs and annotations. The generation algorithm targets a specific
number of lines of C code and annotations. The generated annotations
include all three types of policies, \lstinline{#requires},
\lstinline{#param}, and \lstinline{#return}, in combination with
different primitive types, pointers, and structures. To elicit
worse-case behavior, the generated annotations are predominantly
sequencing annotations constructed from a set of templates
representative of common API patterns from our
case studies. The C programs are similarly generated from 
templates of our case studies. Experiments were run on a single-core 
Ubuntu 18.04 VM with 1GB of RAM, on a 2.7 GHz Intel Core i7 machine. 

First, we evaluate how the runtime of \sysname is affected by the
program size and the number of annotations. The results are summarized
in Figure~\ref{rt-vs-loa}. We evaluate the runtime of four C programs,
with 500, 1000, 2000, and 4000 lines of code respectively. For each
program, we increase the number of annotations, up to 128
annotations. 
\sysname is efficient: all the the experiments finish within
4 seconds. \sysname is intended to be run on individual modules
(libraries) that rarely exceed a couple thousand lines of code
unless they are automatically generated, like the SCDtoObliv circuit
file (14,000 LoC). Even then, \sysname finishes
within 6 seconds. 

To better understand how each component of \sysname contributes to the
processing time, we profile execution time for each part. The
results are summarized in Figure~\ref{t-vs-st}, which shows a
cross-section of Figure~\ref{rt-vs-loa} with only the samples with 128
annotations.  The four stages
of \sysname  are: ``Parse Files," where annotations are retreived;
 ``Generate Header," where the
header file containing type and structure
definitions corresponding to the transformed types is generated;
``Build AST," where the C parsing library, pycparser
\cite{pycparser} builds an abstract syntax tree from the source code;
``Transform," where the implementation of the translation
algorithm of \sysname runs.  
Most of the stages take a negligible amount of time compared to the stages Build AST and 
Transform. The majority of the overhead is due to the C parsing library we use.

We do not present the overhead added to compilation of transformed programs
because developers do not need to compile the transformed programs. Once
the transformed programs have been checked they can be discarded.


\section{Related Work}
\label{sec:related}
Related work for \sysname spans four research areas: C program
analysis tools, information flow types, linear types (type states),
and cryptographic protocol verification. 

\Paragraph{Tools for Analyzing C Programs.}
Many vulnerabilities stem from poorly written C
programs. As a result, many 
C program analysis tools have been built. Several C model
checkers~(e.g. \cite{SLAM:Ball:2002,Blast:Beyer:2007, CBMC:Clarke:2004, SeaHorn:Gurfinkel:2015, CPAChecker:Beyer:2011}) 
and program analysis
tools~\cite{Astree:Cousot:2005, FramaC:Cuoq:2012, KLEE:Cadar:2008, DrC:Machiry:2017}
are open source and readily downloadable.
Our policies can be encoded as state machines and checked by some of
the tools mentioned above, which are general purpose and
more powerful than ours but are not tuned for analyzing API usage
patterns like ours. Further, our tool is  backed by a
sophisticated information flow type system.

Closest to our work is CQual~\cite{CQual:Foster:1999}. Both 
theoretical foundations and practical applications of type qualifiers
have been investigated~\cite{Evans:1996,Scrash:Broadwell:2003,Zhang:2002,Chin:2005,Thesis:Foster:2002}. 
%
%
Our annotations are type qualifiers and our work and prior work on
type qualifiers share the same goal of producing a lightweight tool
to check simple secrecy and integrity properties. We additionally
support sequencing of atomic qualifiers, which is a novel
contribution. Further, we prove noninterference of our core calculus,
which other systems did not.  Another difference is that CQual relies
on a custom type checker, while our policies are translated
and checked using C's type system. 
%
Finally, CQual supports qualifier inference, which can reduce the
annotation burden on programmers. We do not have general qualifier
inference because to do so would be tantamount to constructing a type
checker for our system, which would defeat our goal of relying on a C
compiler's type checker.

\Paragraph{Information Flow Type Systems.}
Information flow type systems is a well-studied field. Several
projects have extended existing languages to include information flow
types (e.g.,~\cite{Pottier:2002, JFlow:Myers:1999}). Sabelfeld et
al. provided a comprehensive summary in their survey
paper~\cite{Survey:Sabelfeld:2003}. Most information flow type systems
do not deal with declassification. At most, they will include a
``declassify" primitive to allow information downgrade, similar to our
relabel operations. However, we have not seen work where the sequence
of labels is part of the information flow type like ours, except for
JRIF~\cite{Kozyri:2016}. As a result, we are able to prove a
noninterference theorem that implies API sequencing.
JRIF uses finite state automata to enforce
sequencing policies, which can entail a large runtime overhead.
%
%
%

Other projects that target enforcement of sequencing policies
similar to those we have presented rely on runtime
monitoring, not types~\cite{Chong:2008, Vachharajani:2004, Beringer:2012, Barany:2017, 
Broberg:2013, Broberg:2010}.

\Paragraph{Linear Types and Typestate.}
Our sequencing policies are tangentially related to other type systems
that aim to enforce API contracts. This line of work includes typestate and linear
types~\cite{Typestate:Strom:1986,Aldrich:2009,Vault:DeLine:2001}. The
idea is that by
using typestate/linear types one can model and check behaviors such as
files being opened and closed in a balanced manner \cite{Aldrich:2009}. However, unlike in typestate
the types on variables don't change in our system; when a part of a 
policy is fulfilled there is a new variable that ``takes on" the rest of the
policy. 

\Paragraph{Cryptographic Protocol Verification.}
Several projects have proposed languages to make verification of
cryptographic programs more feasible: Jasmine, Cryptol, Vale, Dafny,
F*, and Idris~\cite{Jasmine:Almeida:2017,Cryptol:Lewis:2003,
  Vale:Bond:2017,Dafny:Leino:2010,Swamy:2011, Idris:Brady:2013}, to
name a few.  There are also general tools for verifying
cryptographic protocols~\cite{Bhargavan:2010,Bhargavan:2008,
  Blanchet:2001,Scyther:Cremers:2008,Bengtson:2011,Cortier:2005, 
  Costanzo:2016}.
These languages and tools are general purpose and more powerful than
ours. However, none of these tools directly support checking properties of C
implementations of cryptographic libraries like we do.
Bhargavan et al.'s work uses refinement types to achieve
similar goals as ours~\cite{Bhargavan:2010}.  The annotated types can be
viewed as refinement types: $\{x:\tau\,|\, \pol\}$, where the policy is
encoded as a predicate. Their system is more powerful, however it only supports F\# code.

\section{Conclusion}
\label{sec:conclusion}
We have described \sysname, a lightweight annotation system for C that
allows programmers to specify secrecy, integrity, and sequencing
policies for their applications.  \sysname is particularly useful in
identifying errors at compile time that violate high-level 
policies in cryptographic libraries and
applications.  We have modeled our system formally and
proved a noninterference guarantee.  Finally, we have shown through a
set of detailed case studies that \sysname can express and check
complex policies for large bodies of C code and finds subtle
implementation bugs.

\bibliographystyle{plain}
\bibliography{paper}

\appendix

\section{Summary of \minic: A Core Calculus with Nominal Typing}
\label{app:minic}
We summarize the syntax, operational semantics, and typing rules for
\minic in this section. \minic represents the fragment of C that
\sysname works with.
\subsection{Syntax}
\[
\begin{array}{lcll}
\textit{Basic Types} & \pi & \bnfdef & T \bnfalt \tpint \bnfalt \tunit\bnfalt
                                       \tptr(\tau) 
\\
\textit{Types} & \tau & \bnfdef & \pi \bnfalt \pi_1\rightarrow\pi_2
\\
\textit{Values} & v & \bnfdef & x \bnfalt n \bnfalt ()\bnfalt
                                (T)\{v_1,\cdots, v_k\} \bnfalt \eloc
                                \bnfalt f
\\
\textit{Expressions} & e & \bnfdef &
 v \bnfalt e_1 \bop e_2\bnfalt  v\, e 
\bnfalt \elet\, x=e_1\, \ein\, e_2
\bnfalt v.i 
\\
& & \bnfalt & \eif\, v~ \ethen\, e_1\, \eelse \, e_2
\bnfalt  \enew(e) \bnfalt v\, :=\, e \bnfalt *v
\\
\textit{Type def. ctx} & D & \bnfdef &
\cdot\bnfalt D, T \mapsto \mathsf{struct}\ T\ \{\pi_1, \cdots, \pi_k\} 
\\  
\textit{Func typing ctx} & F & \bnfdef & \cdot \bnfalt F, f : \pi_1\rightarrow\pi_2
\\ 
\textit{Code ctx} & \Psi & \bnfdef & \cdot \bnfalt \Psi, f(x) = e
\\
\textit{Typing ctx} & \Gamma & \bnfdef & \cdot \bnfalt \Gamma, x: \tau
\\ 
\textit{Store} & \sigma & \bnfdef & \cdot \bnfalt \sigma, \eloc\mapsto v
\\ 
\textit{Store Typing} & \storetp & \bnfdef & \cdot \bnfalt \storetp, \eloc: \tau \\
\textit{Eval Ctx} & E & \bnfdef & \elet\, x=\hole \, \ein\,  e
  \bnfalt \enew(\hole)    \bnfalt v\, \hole
   \bnfalt v\, :=\,\hole \bnfalt \hole \bop e
\bnfalt v \bop \hole
\end{array}
\]

\subsection{Operational Semantics}



\begin{mathpar}
\inferrule*[right=N-E-Context]{
\Psi \vdash \sigma\sepidx{} e \stepsto  \sigma' \sepidx{}e'
}{
\Psi \vdash \sigma\sepidx{} E[e] \stepsto  \sigma' \sepidx{}E[e']
}
\and
\inferrule*[right=N-E-Bop]{
  v_1 \bop v_2 = v
}{
\Psi \vdash \sigma\sepidx{i}  v_1 \bop v_2 \stepsto  \sigma\sepidx{i}
v
}
\and
\inferrule*[right=N-E-Deref]{
}{
\Psi \vdash \sigma\sepidx{}  *\eloc \stepsto  \sigma\sepidx{}
\sigma(\eloc) 
}
\and
\inferrule*[right=N-E-Assign]{
}{
\Psi \vdash \sigma\sepidx{} \eloc :=  v 
  \stepsto  \sigma[\eloc\mapsto  v]\sepidx{} ()
}
\and
\inferrule*[right=N-E-New]{
\eloc \ \mathit{fresh}
}{
\Psi \vdash \sigma\sepidx{} \enew( v) 
\stepsto  \sigma[\eloc\mapsto v]\sepidx{} \eloc
}
\and
\inferrule*[right=N-E-Field]{
}{
\Psi \vdash \sigma\sepidx{} (\{ v_1,\cdots,  v_n\}).i \stepsto  \sigma \sepidx{}  v_i
}
\and
\inferrule*[right=N-E-App]{
\Psi = \Psi', f(x)=e
}{
\Psi \vdash \sigma\sepidx{} f\  v \stepsto  \sigma\sepidx{}
e[v/x]
}
\and
\inferrule*[right=N-E-Let]{
}{
\Psi \vdash \sigma\sepidx{} \elet\ x=  v\ \ein\ e \stepsto  \sigma\sepidx{} e[v/x]
}
\and
\inferrule*[right=N-E-If-True]{
n > 0 
}{
\Psi \vdash \sigma\sepidx{} \eif\ n\ \ethen\  e_1\ \eelse\  e_2  \stepsto  \sigma\sepidx{}   e_1
}
\and
\inferrule*[right=N-E-If-False]{
}{
\Psi \vdash \sigma\sepidx{} \eif\ 0\ \ethen\  e_1\ \eelse\  e_2  \stepsto  \sigma\sepidx{}  e_2
}
\end{mathpar}

\subsection{Typing Rules}


\noindent\framebox{$D;F;\Sigma; \Gamma \vdash e : \tau$}
\begin{mathpar}
\inferrule*[right=N-T-Int]{
}{ 
  D;F; \Sigma; \Gamma  \vdash n : \tpint
}
\and
\inferrule*[right=N-T-Loc]{
}{ 
  D;F; \Sigma; \Gamma  \vdash \eloc: \storetp(\eloc) 
}
\and
\inferrule*[right=N-T-Var]{
}{ 
  D;F; \Sigma; \Gamma  \vdash x : \Gamma(x) 
}
\and
\inferrule*[right=N-T-Fun]{
}{ 
  D;F; \Sigma; \Gamma  \vdash f : F(f) 
}
\and
\inferrule*[right=N-T-Bop]{
  D;F; \Sigma; \Gamma \vdash e_1:  \tpint
\\   D;F; \Sigma; \Gamma \vdash e_2:  \tpint
}{ 
  D;F; \Sigma; \Gamma \vdash e_1\bop e_2 : \tpint
}
\end{mathpar}
\begin{mathpar}
\inferrule*[right=N-T-Struct]{
T\mapsto \tstruct\ T\ \{\tau_1,\cdots, \tau_n\} \in D
\\\forall i,  D; F;\Sigma; \Gamma  \vdash v_i: \tau_i
}{ 
  D;F; \Sigma; \Gamma \vdash (T)\ \{v_1,\cdots,v_n\} :
  T
}
\and
\inferrule*[right=N-T-Field]{
  D;F; \Sigma; \Gamma  \vdash v : T
\\ T \mapsto \tstruct\ T \{\tau_1,\cdots, \tau_n\} \in D
}{ 
  D;F; \Sigma; \Gamma  \vdash v.i : \tau_i 
}
\and
\inferrule*[right=N-T-New]{
  D;F; \Sigma; \Gamma  \vdash e : \tau
}{ 
  D;F; \Sigma; \Gamma  \vdash \enew(e): \tptr(\tau)
}
\and
\inferrule*[right=N-T-Deref]{
  D;F; \Sigma; \Gamma  \vdash v : \tptr(\tau)
}{ 
  D;F; \Sigma; \Gamma  \vdash *v: \tau
}
\and
\inferrule*[right=N-T-Assign]{
  D;F; \Sigma; \Gamma  \vdash v : \tptr(\tau)
 \\  D;F; \Sigma; \Gamma  \vdash e : \tau
}{ 
  D;F; \Sigma; \Gamma  \vdash v:= e: \tunit
}
\and
\inferrule*[right=N-T-App]{
  D;F; \Sigma; \Gamma  \vdash v :  \pi_1 \rightarrow \pi_2
\\  D;F; \Sigma; \Gamma  \vdash e:  \pi_1
}{ 
  D;F; \Sigma; \Gamma  \vdash v\ e: \pi_2
}
\and
\inferrule*[right=N-T-Let]{
  D;F; \Sigma; \Gamma  \vdash e_1 : \tau_1
\\   D;F; \Sigma; \Gamma, x:\tau_1  \vdash e_2 : \tau_2
}{ 
  D;F; \Sigma; \Gamma  \vdash \elet\ x:\tau_1=e_1\ \ein\ e_2: \tau_2
}
\and
\inferrule*[right=N-T-If]{
  D;F; \Sigma; \Gamma  \vdash v : \tpint
\\   D;F; \Sigma; \Gamma  \vdash e_1 : \tau
\\   D;F; \Sigma; \Gamma \vdash e_2 : \tau
}{ 
  D;F; \Sigma; \Gamma  \vdash \eif\ v\ \ethen\ e_1\ \eelse\
  e_2 : \tau
}
\end{mathpar}

\section{Definitions and Meta-theory for \langname}
\label{app:polc}

\subsection{\bf \langname Operational Semantics via Pairs}

The operational semantic rules for \langname include all the rules for
\minic and the following rule for relabeling. 
\[
\inferrule*[right=P-E-Relab]{
}{
\Psi \vdash \sigma\sepidx{}  \erelab(\pol'\Leftarrow\pol) \bv
\stepsto  \sigma\sepidx{} \bv
}
\]

\subsection{Extension of Syntax with Pairs}

To prove noninterference, we define a set of operational semantic
rules that allow expression pairs, which effectively represent two
executions differing in secrets. 
The syntax for the extended values and expressions are summarized below.
\[
\begin{array}{llcl}
\textit{Ext. Values} & \extV & \bnfdef & 
v \bnfalt (T)\{\extV_1,\cdots, \extV_k\} \bnfalt \epair{v_1}{v_2} 
\\
\textit{Ext. Exprs.} & \extE & \bnfdef &  \extV \bnfalt \extV\, \extE
					\bnfalt \elet\, x=\extE_1\, \ein\, \extE_2 \bnfalt \extV.i \\
					& & \bnfalt & \eif\, \extV_1\, \ethen\, \extE_2\, \eelse \, \extE_3 \\
					& & \bnfalt & \enew(\extE) \bnfalt \extV_1\, :=\, \extE_2 \bnfalt *\extV
 					\bnfalt \epair{e_1}{e_2}
\\
\textit{Stored Values} & v^s & \bnfdef & 
\extV\bnfalt \epair{\bullet}{v_2}  \bnfalt \epair{v_1}{\bullet} 
\\
\textit{Store} & \sigma &\bnfdef & \cdot \bnfalt \sigma, \eloc\mapsto v^s
\end{array}
\]
We write $\extV$ to denote values that may include pairs and $\extE$
to denote expressions that may include pairs. The definitions disallow
nested pairs. For the rest of this section, when convient and
clear from the context, we
will write $v$ and $e$ to denote values and expressions that may
contain pairs respectively.

\subsection{Paired Operational Semantics} 
The operational semantics is summarized in
Figure~\ref{fig:op-semantics-extended}. Below are auxiliary
definitions used by those rules.

\[
\begin{array}{lcl}
	\enew~ v = v & \enew_1~ v = \epair{v}{\bullet} & \enew_2~ v = \epair{\bullet}{v} \\
	\eread ~v = v & \eread_1 ~v = \proj{v}{1} & \eread_2 ~v = \proj{v}{2} \\
	\eupdate ~v~v' = v' & \eupdate_1~v~v' = \epair{v'}{\proj{v}{2}} & 
		\eupdate_2~v~v' = \epair{\proj{v}{1}}{v'} \\
\end{array}
\]

\[
  \begin{array}{r l}
  \proj{x}{i} =				& x \\
  \proj{n}{i} =				& n \\
  \proj{()}{i} =				& () \\
  \proj{(T)\{v_1,\cdots, v_k\}}{i} =
  						& (T)\{\proj{v_1}{i},\cdots, \proj{v_k}{i}\} \\
  \proj{\eloc}{i} = 			& \eloc \\
  \proj{f}{i} = 				& f \\
  \proj{(T)\{v_1^+,\cdots, v_k^+\}}{i} =
  						&  (T)\{\proj{v_1^+}{i},\cdots, \proj{v_k^+}{i}\} \\
  \proj{\epair{v_1}{v_2}}{i} =	& v_i \\
  \proj{v^+~e^+}{i} =			&  \proj{v^+}{i}~\proj{e^+}{i} \\
  \proj{\elet\, x=e^+_1\, \ein\, e^+_2}{i} = 
  						&  \elet~x=\proj{e^+_1}{i}~\ein~\proj{e^+_2}{i} \\
  \proj{v^+.j}{i} =				&  \proj{v^+}{i}.j \\
  \proj{\eif~v^+~\ethen~e^+_1~\eelse~e^+_2}{i} =	
  						&  \eif~\proj{v^+}{i}~\ethen~\proj{e^+_1}{i}~\eelse~\proj{e^+_2}{i} \\
  \proj{\enew(e^+)}{i} =		&  \enew(\proj{e^+}{i}) \\
  \proj{v^+ := e^+}{i} =		&  \proj{v^+}{i} := \proj{e^+}{i} \\
  \proj{*v^+}{i} =				&  *\proj{v^+}{i} \\
  \proj{\epair{e_1}{e_2}}{i} =	&  e_i\\
  \end{array}
\]

\[
  \begin{array}{r l}
  x[x \Leftarrow \extV] = 				& \extV \\
  y[x \Leftarrow \extV] = 				& y \\
  n[x \Leftarrow \extV] = 				& n \\
  ()[x \Leftarrow \extV] = 				& () \\
  ((T)\{v_1, \cdots, v_n\})[x \Leftarrow \extV] =  
  								& (T)\{v_1[x \Leftarrow \extV], \cdots, 
									v_n[x \Leftarrow \extV]\} \\
  \eloc[x \Leftarrow \extV] =  			& \eloc \\
  f [x \Leftarrow \extV] =  				& f \\
  ((T)\{\extV_1, \cdots, \extV_n\})[x \Leftarrow \extV] =  
  								& (T)\{\extV_1[x \Leftarrow \extV], \cdots, 
									\extV_n[x \Leftarrow \extV]\} \\
  \epair{v_1}{v_2}[x \Leftarrow \extV] = 	& \epair{v_1[x \Leftarrow \proj{\extV}{1}]}
  										{v_2[x \Leftarrow \proj{\extV}{2}]} \\
  (\extV~\extE)[x \Leftarrow \extV_1] =  	& \extV[x \Leftarrow \extV_1] ~ \extE[x \Leftarrow \extV_1]\\
  (\elet\, y=\extE_1\, \ein\, \extE_2)[x \Leftarrow \extV] =  
  								& \elet\, y=\extE_1[x \Leftarrow \extV]\, \ein\, \extE_2[x \Leftarrow \extV]~ \\
  (\extV.i)[x \Leftarrow \extV_1] =  				& (\extV[x \Leftarrow \extV_1]).i \\
  (\eif~\extV~\ethen~\extE_1~\eelse~\extE_2)[x \Leftarrow \extV_1] =  
  								& \eif~\extV[x \Leftarrow \extV_1]~\ethen~
									\extE_1[x \Leftarrow \extV_1] \\
  								~ &\eelse~\extE_2[x \Leftarrow \extV_1] \\
  (\enew(\extE))[x \Leftarrow \extV_1] =  	& \enew(\extE[x \Leftarrow \extV_1]) \\
  (\extV := \extE)[x \Leftarrow \extV_1] =  
  								& \extV[x \Leftarrow \extV_1] := 
  								\extE[x \Leftarrow \extV_1] \\
  (*\extV)[x \Leftarrow \extV_1]  =		&  *(\extV[x \Leftarrow \extV_1]) \\
  \epair{e_1}{e_2}[x \Leftarrow \extV] =	& \epair{e_1[x \Leftarrow \proj{\extV}{1}]}
  										{e_2[x \Leftarrow \proj{\extV}{2}]}
  \end{array}
\]

\subsection{Soundness and Completeness of the Paired Semantics}
We first define projection relations.

\begin{mathpar}
\inferrule
	{~}
	{\proj{\cdot}{i} = \cdot}
\and
\inferrule
	{\proj{v^s}{i} = \bullet}
	{\proj{\sigma, \eloc \mapsto v^s}{i} = \proj{\sigma}{i}}
\and
\inferrule
	{\proj{v^s}{i} \neq \bullet}
	{\proj{\sigma, \eloc \mapsto v^s}{i} = \proj{\sigma}{i}, \eloc \mapsto \proj{v^s}{i}}
\end{mathpar}

\begin{mathpar}
\proj{\sigma \sepidx{} \extE}{i} = \proj{\sigma}{i} \sepidx{} \proj{\extE}{i}
\end{mathpar}

Next we define a number of well-formedness invariants of the runtime
configuration $\sigma \sepidx{i} e$. 

\begin{defn}[Defined Pointers]
We say that $\eloc$ is defined in $\sigma$ for execution $i$ if the
following holds 
\begin{itemize}
\item $i=\bullet$ implies $\forall j\in\{1,2\}$,
  $\proj{\sigma(\eloc)}{j} = v$
\item $i\in\{1,2\}$ implies $\proj{\sigma(\eloc)}{i} = v$
\end{itemize}
\end{defn}

\begin{defn}[In Scope Pointers]
We say that $\eloc$ is scope of $\sigma \sepidx{i} \extE$ where 
$i\in\{1,2\}$ \\if
$\eloc \in (
			(
				\bigcup_{\eloc'\in\m{dom}(\sigma)} \m{fl}(\proj{\sigma(\eloc')}{i})
			) 
			\cup 
				\m{fl}(\proj{\extE}{i})
		)$
\end{defn}

\begin{defn}[Closed Configurations]
We say that $\sigma \sepidx{i} \extE$  is closed if all of the
following holds
\begin{itemize}
\item $i=\bullet$ implies  $\forall i\in\{1,2\}$, for all $\eloc$
s.t. $\eloc$ is in sope of $\sigma \sepidx{i} \extE$, $\eloc$ is
defined in $\sigma$ for execution $i$.

\item $i\in\{1,2\}$ implies for all $\eloc$
s.t. $\eloc$ is in sope of $\sigma \sepidx{i} \extE$, $\eloc$ is
defined in $\sigma$ for execution $i$.
\end{itemize}
\end{defn}

\begin{figure}[t!]
\flushleft
\noindent{\framebox{$\Psi \vdash \sigma\sepidx{i} e \stepsto
    \sigma'\sepidx{i} e'$}}
\begin{mathpar}
\inferrule*[right=P-E-Context]{
\Psi \vdash \sigma\sepidx{i} e \stepsto  \sigma' \sepidx{i}e'
}{
\Psi \vdash \sigma\sepidx{i} E[e] \stepsto  \sigma' \sepidx{i}E[e']
}
\and
\inferrule*[right=P-E-Pair]{
\Psi \vdash \sigma\sepidx{i} e_i \stepsto  \sigma' \sepidx{i}e'_i
\\ e_j = e'_j
\\ \{i,j\} = \{1,2\}
}{
\Psi \vdash \sigma\sepidx{} \epair{e_1}{e_2} \stepsto  \sigma'
\sepidx{} \epair{e'_1}{e'_2}
}
\and
\inferrule*[right=P-E-Lift-App]{
}{
\Psi \vdash \sigma\sepidx{} \epair{v_1}{v_2}\, v \stepsto
\sigma\sepidx{} \epair{v_1\, \proj{v}{1}}{v_2\, \proj{v}{2}}
}
\and
\inferrule*[right=P-E-Lift-Deref]{
}{
\Psi \vdash \sigma\sepidx{} *\epair{v_1}{v_2} \stepsto  \sigma\sepidx{} \epair{*v_1}{*v_2}
}
\and
\inferrule*[right=P-E-Lift-Assign]{
}{
\Psi \vdash \sigma \sepidx{}\epair{v_1}{v_2} := v 
 \stepsto  \sigma\sepidx{} \epair{v_1  :=\proj{v}{1} }{v_2  := \proj{v}{2} }
}
\and
\inferrule*[right=P-E-Lift-Field]{
}{
\Psi \vdash \sigma\sepidx{} \epair{v_1}{v_2}.j \stepsto
\sigma\sepidx{}  \epair{v_1.j}{v_2.j}
}
\and
\inferrule*[right=P-E-Lift-If]{
}{
\Psi \vdash \sigma\sepidx{} \eif\ \epair{v_1}{v_2}\ \ethen\ v_t\
\eelse\ v_f
\\\\\stepsto  \sigma\sepidx{} 
\langle\eif\ v_1\ \ethen\ \proj{v_t}{1}\ \eelse\ \proj{v_f}{1}
~|~ \eif\ v_2\ \ethen\ \proj{v_t}{2}\ \eelse\ \proj{v_f}{2}\rangle
}
\and
\inferrule*[right=P-E-Lift-Relab]{
}{
\Psi \vdash \sigma\sepidx{}
\erelab(\pol'\Leftarrow\pol)\epair{v_1}{v_2} 
\\ \stepsto  \sigma\sepidx{} \epair{\erelab(\pol'\Leftarrow\pol)v_1}{\erelab(\pol'\Leftarrow\pol)v_2}
}
\and
%
\inferrule*[right=P-E-Bop]{
  v_1 \bop v_2 = v
}{
\Psi \vdash \sigma\sepidx{}  v_1 \bop v_2 \stepsto  \sigma\sepidx{}
v
}
\and
\inferrule*[right=P-E-Relab]{
}{
\Psi \vdash \sigma\sepidx{i}  \erelab(\pol'\Leftarrow\pol) v
\stepsto  \sigma\sepidx{i} v
}
\and
\inferrule*[right=P-E-Deref]{
}{
\Psi \vdash \sigma\sepidx{i}  *\eloc \stepsto  \sigma\sepidx{i}
\eread_i\ \sigma(\eloc) 
}
\and
\inferrule*[right=P-E-Assign]{
}{
\Psi \vdash \sigma\sepidx{i} \eloc:= v 
  \stepsto  \sigma[\eloc\mapsto \eupdate_i\ \sigma(\eloc)\
  v]\sepidx{i} ()
}
\and
\inferrule*[right=P-E-New]{
\eloc \ \mathit{fresh}
}{
\Psi \vdash \sigma\sepidx{i} \enew(v) 
\stepsto  \sigma[\eloc\mapsto \enew_{i}~v]\sepidx{i} \eloc
}
\and
\inferrule*[right=P-E-Field]{
}{
\Psi \vdash \sigma\sepidx{i} (\{v_1,\cdots, v_n\}).j \stepsto  \sigma \sepidx{i} v_j
}
\and
\inferrule*[right=P-E-App]{
\Psi = \Psi', f(x)=e
}{
\Psi \vdash \sigma\sepidx{i} f\ v \stepsto  \sigma\sepidx{i}
e[x\Leftarrow v][f\Leftarrow f(x)=e]
}
\quad
\inferrule*[right=P-E-Let]{
}{
\Psi \vdash \sigma\sepidx{i} \elet\ x= v\ \ein\ e \stepsto
\sigma\sepidx{i} e[x\Leftarrow v]
}
\and
\inferrule*[right=P-E-If-True]{
n > 0 
}{
\Psi \vdash \sigma\sepidx{i} \eif\ n\ \ethen\ v_1\ \eelse\ v_2  \stepsto  \sigma\sepidx{i}  v_1
}
\quad
\inferrule*[right=P-E-If-False]{
}{
\Psi \vdash \sigma\sepidx{i} \eif\ 0\ \ethen\ v_1\ \eelse\ v_2  \stepsto  \sigma\sepidx{i} v_2
}
\end{mathpar}
\caption{Operational Semantics of Extended \langname}
\label{fig:op-semantics-extended}
\end{figure}

\begin{lem}[Preservation of  Well-formednness]\label{lem:pres-wf}
~\\
\begin{enumerate}
\item For $i \in \{1, 2, \bullet\}$ if $\sigma_1 \sepidx{i} \extE_1$ is closed and 
  $\Psi \vdash \sigma_1 \sepidx{i} \extE_1 \stepsto x \sepidx{i} y$ then 
exists $\sigma_2$ and $\extE_2$ s.t. $x=\sigma_2$ and $y = \extE_2$ and
$\sigma_2 \sepidx{i} \extE_2$ is closed. 
\item For $i \in \{1, 2\}$ if $\sigma_1 \sepidx{i} e_1$ is closed and 
  $\Psi \vdash \sigma_1 \sepidx{i} e_1 \stepsto x \sepidx{i} y$ then 
exists $\sigma_2$ and $e_2$ s.t. $x=\sigma_2$ and $y = e_2$ and
$\sigma_2 \sepidx{i} e_2$ is closed. 
\end{enumerate}
\end{lem}

\begin{proofsketch}
By induction over the structure of the operational semantic rules.
\end{proofsketch}

\begin{lem}[Distributivity of Projection for Expressions]
\label{lem:proj-distributed}
$\proj{\extE[x\Leftarrow \extV]}{i}=\proj{\extE}{i}[x\Leftarrow \proj{\extV}{i}]$
\end{lem}

\begin{proof}
By induction over the structure of $\extE$. Most cases can be proven
by straightforward application of I.H., which we omit.
\begin{description}

\item [Case:] $\extE = x$
\begin{tabbing}
~~\=(1)~~~~\= $\proj{x[x \Leftarrow \extV]}{i} = \proj{\extV}{i}$
\\\>(2)\> $\proj{x}{i}[x \Leftarrow \proj{\extV}{i}] = x[x \Leftarrow
\proj{\extV}{i}] = \proj{\extV}{i}$
\\By (1) and (2)
\\\>(3) \>$\proj{x[x \Leftarrow \extV]}{i} = \proj{x}{i}[x \Leftarrow
\proj{\extV}{i}]$
\end{tabbing} 

\item [Case:] $\extE = \epair{e_1}{e_2}$
\begin{tabbing}
~~\=(1)~~~~\= $\proj{\epair{e_1}{e_2}[x \Leftarrow \extV]}{i}= 
\proj{\epair{e_1[x \Leftarrow \proj{\extV}{1}]}{e_2[x \Leftarrow
    \proj{\extV}{2}]}}{i}=  e_i[x \Leftarrow \proj{\extV}{i}]$ 
\\\>(2)\>$\proj{\epair{e_1}{e_2}}{i}[x \Leftarrow \proj{\extV}{i}] = 
     e_i[x \Leftarrow \proj{\extV}{i}]$
\\By (1) and (2)
\\\>(3)\>$\proj{\epair{e_1}{e_2}[x \Leftarrow \extV]}{i}= \proj{\epair{e_1}{e_2}}{i}[x \Leftarrow \proj{\extV}{i}]$
\end{tabbing}
\end{description}
\end{proof}

\begin{lem}\label{lem:i-exec-proj}
If for all $i\in\{1, 2\}$, $\ee::\Psi \vdash \sigma_1\sepidx{i} e_1 \stepsto
\sigma_2\sepidx{i}e_2$ where 
$\sigma_1 \sepidx{i} e_1$ is closed 
\\then $\Psi \vdash \proj{\sigma_1}{i}\sepidx{} e_1\stepsto
\proj{\sigma_2}{i}\sepidx{} e_2$ 
and  $\proj{\sigma_1}{j} =
\proj{\sigma_2}{j}$, where $\{i,j\}= \{1,2\}$
\end{lem}

\begin{proof}
Proof by induction on the structure of $\ee$. For most cases, the
store will not be updated. The proof follows directly by applying the
same rule. We will present cases of memory operations.

\begin{description}

\item [Case:] $\ee$ ends in \rulename{P-E-Deref}
\begin{tabbing}
By assumption
\\~~\=(1)~~~~\= $\Psi \vdash \sigma\sepidx{i}  *\eloc \stepsto  \sigma\sepidx{i}
(\eread_i\ \sigma(\eloc))$ 
\\ By \rulename{deref}
\\\> (2)\> $\Psi \vdash \proj{\sigma}{i} \sepidx{} *\eloc \stepsto 
				\proj{\sigma}{i} \sepidx{} \eread\
                                \proj{\sigma}{i}(\eloc)$
\\ By $\sigma \sepidx{i} *\eloc$ is closed
\\\> (3)\> $\proj{\sigma(\eloc)}{i} = v$, where $v$ is a \langname
value
\\ By definition of $\eread$:
\\\>(4)\>$\eread_i~\sigma(\eloc)=\proj{\sigma(\eloc)}{i} =v$
\\\>(5)\>$\eread\ \proj{\sigma}{i}(\eloc) =(\proj{\sigma}{i})(\eloc)=v$
\\By (4) and (5)
\\\> (6)\> $\proj{\sigma}{i}(\eloc)= \proj{\sigma(\eloc)}{i} = v$
\end{tabbing}

\item [Case:]  $\ee$ ends in \rulename{P-E-Assign} 
\begin{tabbing}
By assumption
\\~~\=(1)~~~~\= $\Psi \vdash \sigma\sepidx{i} \eloc:= v 
  \stepsto  \sigma_2 \sepidx{i} ()$ and $\sigma_2=\sigma[\eloc\mapsto \eupdate_i\ \sigma(\eloc)\
  v]$
\\ By \rulename{assign}:
\\\> (2)\>  $\Psi \vdash \proj{\sigma}{i} \sepidx{} \eloc := v \stepsto 
                                \sigma'\sepidx{} ()$ and 
$\sigma'=  	\proj{\sigma}{i}[\eloc \mapsto\eupdate~\proj{\sigma}{i}(\eloc)\ v] $
\\We show the case for when $i=1$, the case for $i=2$ can be proven similarly
\\By definition of $\eupdate$
\\\>(3)\> $\sigma' = \proj{\sigma}{i}[\eloc \mapsto v]$
\\\>(4)\> $\sigma_2
 =\sigma[\eloc \mapsto \epair{\proj{v}{1}}{\proj{\sigma(\eloc) }{2}}]$
\\By $v$ is a valid extended \langname expression
\\\>(5)\> $v$ does not contain $\bullet$ and $v=\proj{v}{1}$
\\By the definition of projection
\\\>(6)\>$\proj{\sigma_2}{1} = \proj{\sigma}{1}[\eloc \mapsto
\proj{v}{1}] = \sigma'$
\\ There are two subcases 
\\\>{\bf Subcase a.} $\eloc$ is defined in $\sigma$ for execution $2$
\\\>~~\=(a)~~\= $\proj{\sigma_2}{2}= \proj{\sigma}{2}[\eloc\mapsto
\proj{\sigma(\eloc) }{2}] = \proj{\sigma}{2}$
\\\>{\bf Subcase b.} $\eloc$ is not defined in $\sigma$ for execution
$2$
\\\>~~\=(b)~~\= $\proj{\sigma_2}{2}= \proj{\sigma[\eloc
  \mapsto \epair{\proj{v}{1}}{\proj{\sigma(\eloc) }{2}}]}{2} =
\proj{\sigma[\eloc \mapsto \epair{\proj{v}{1}}{\bullet}]}{2}= 
 \proj{\sigma}{2}$

\end{tabbing}

\item [Case:] $\ee$ ends in \rulename{P-E-New} 
\begin{tabbing}
By assumption
\\~~\=(1)~~~~\= $\Psi \vdash \sigma\sepidx{i} \enew(v) 
\stepsto  \sigma_2\sepidx{i} \eloc$ and $\sigma_2=\sigma[\eloc\mapsto \enew_i~v]$
\\By \rulename{new}
\\\>(2)\>$\Psi \vdash \proj{\sigma}{i}\sepidx{} \enew(v)
\stepsto  \sigma' \sepidx{} \eloc$ and
$\sigma'=\proj{\sigma}{i}[\eloc\mapsto \enew\ v]$
\\We show the case for when $i=1$, the case for $i=2$ can be proven similarly
\\By definition of $\enew$
\\\>(3)\> $\sigma_2 = \sigma[\eloc \mapsto\enew_1~v]  =\sigma[\eloc \mapsto \epair{v}{\bullet}]$
\\\>(4)\> $\sigma'= \proj{\sigma}{1}[\eloc\mapsto v]$
\\By $v$ is a valid extended \langname expression and the definition of projection
\\\>(6)\>$\proj{\sigma_2}{1} =\sigma'= \proj{\sigma}{1}[\eloc\mapsto
v] $
\\\>(7)\>$\proj{\sigma_2}{2} = \proj{\sigma[\eloc \mapsto
  \epair{v}{\bullet}]}{2} = \proj{\sigma}{2} $
\end{tabbing}
\end{description}
\end{proof}

\begin{lem}\label{lem:proj-context}
$\proj{E}{i}[\proj{e}{i}] = \proj{E[e]}{i}$
\end{lem}

\begin{proofsketch}
Proof by induction on the structure of $E$.
\end{proofsketch}

\begin{thm}[Soundness] 
\label{thm:pair-op-soundness}
If $\ee::\Psi \vdash \sigma_1\sepidx{} \extE_1 \stepsto
\sigma_2\sepidx{}\extE_2$ where 
$\sigma_1 \sepidx{} \extE_1$ is closed 
\\then for all $i\in\{1, 2\}$, $\Psi \vdash \proj{\sigma_1\sepidx{} \extE_1}{i} \stepsto
\proj{\sigma_2\sepidx{}\extE_2}{i}$; or $\proj{\sigma_1\sepidx{} \extE_1}{i} =
\proj{\sigma_2\sepidx{}\extE_2}{i}$.
\end{thm}

\begin{proof}
Proof by induction on the structure of $\ee$. Most cases are
straightforward. We show a few key cases below. 
\begin{description}
\item [Case:] $\ee$ ends in \rulename{Context}
\begin{tabbing}
By assumption:
\\~~\=(1)~~~~\= $\Psi \vdash \sigma\sepidx{} E[e] \stepsto  \sigma'
\sepidx{}E[e']$ 
\\\>(2)\> $\ee':: \Psi \vdash \sigma\sepidx{} e \stepsto  \sigma' \sepidx{}e'$
\\ By I.H. on $\ee'$
\\\> (3)\> $\Psi \vdash \proj{\sigma \sepidx{} e}{i} \stepsto \proj{\sigma'
  \sepidx{} e'}{i}$, $i\in\{1,2\}$
\\ By definition of projection
\\\> (4)\> $\Psi \vdash \proj{\sigma}{i} \sepidx{} \proj{e}{i} \stepsto \proj{\sigma'}{i} \sepidx{} \proj{e'}{i}$
\\ By (Context) and (4)
\\\> (5)\> $\Psi \vdash \proj{\sigma}{i} \sepidx{} \proj{E}{i}[\proj{e}{i}] \stepsto
					\proj{\sigma'}{i} \sepidx{} \proj{E}{i}[\proj{e'}{i}]$
\\ By Lemma \ref{lem:proj-context}, (5):
\\\> (6)\> $\Psi \vdash \proj{\sigma}{i} \sepidx{} \proj{E[e]}{i} \stepsto
					\proj{\sigma'}{i} \sepidx{} \proj{E[e']}{i}$
\end{tabbing}

\item [Case:] $\ee$ ends in \rulename{pair}
\begin{tabbing}
By assumption
\\~~\=(1)~~~~\= $\Psi \vdash \sigma\sepidx{} \epair{e_1}{e_2} \stepsto  \sigma'
\sepidx{} \epair{e'_1}{e'_2}$
\\\>(2)\> $\ee'::\Psi \vdash \sigma\sepidx{k} e_k \stepsto  \sigma' \sepidx{k}e'_k$, $e_j = e'_j$,
 $\{k,j\} = \{1,2\}$
\\ By Lemma~\ref{lem:i-exec-proj} and (2)
\\\> (3)\> $\Psi \vdash \proj{\sigma}{k}\sepidx{} e_k\stepsto
\proj{\sigma'}{k}\sepidx{}e'_k$ 
and 
\\\>(4)\> $\proj{\sigma}{j}=\proj{\sigma'}{j}$
\\ \>{\bf Subcase a.} $i=k$
\\\> By (3), the conclusion holds
\\ \>{\bf Subcase b.} $i=j$
\\\> By (2), 
\\\>\> (5)~~\=$\proj{\epair{e_1}{e_2} }{j} = \proj{\epair{e'_1}{e'_2}
}{j}$
\\\> By (4) and (5), the conclusion holds
\end{tabbing}

\item [Case:] $\ee$ ends in \rulename{deref}
\begin{tabbing}
By assumption
\\~~\=(1)~~~~\= $\Psi \vdash \sigma\sepidx{}  *\eloc \stepsto  \sigma\sepidx{}
(\eread\ \sigma(\eloc))$ 
\\ T.S. $\Psi \vdash \proj{\sigma \sepidx{} *\eloc}{i} \stepsto \proj{\sigma \sepidx{} \eread~\sigma(\eloc)}{i}$
\\ By \rulename{deref}
\\\> (2)\> $\Psi \vdash \proj{\sigma}{i} \sepidx{} *\eloc \stepsto 
				\proj{\sigma}{i} \sepidx{} \eread\ (\proj{\sigma}{i})(\eloc)$
\\ T.S. $\proj{\eread~\sigma(\eloc)}{i} =\eread\ (\proj{\sigma}{i} \eloc)$
\\ By definition of $\eread$:
\\\>(3)\>$\proj{\eread~\sigma(\eloc)}{i}=\proj{\sigma(\eloc)}{i}$
\\\>(4)\>$\eread\ (\proj{\sigma}{i})(\eloc) =\proj{\sigma}{i}(\eloc)$
\\ By $\sigma \sepidx{} *\eloc$ is closed:
\\\> (5)\> $\proj{\sigma(\eloc)}{i} = v$, where $v$ is a \langname
value
\\By projection definitions 
\\\> (6)\> $\proj{\sigma}{i}(\eloc)= \proj{\sigma(\eloc)}{i} = v$
\end{tabbing}

\item [Case:]  $\ee$ ends in \rulename{assign} 
\begin{tabbing}
By assumption
\\~~\=(1)~~~~\= $\Psi \vdash \sigma\sepidx{} \eloc:= v 
  \stepsto  \sigma[\eloc\mapsto \eupdate\ \sigma(\eloc)\
  v]\sepidx{} ()$
\\ T.S. $\Psi \vdash \proj{\sigma \sepidx{} \eloc := v}{i} \stepsto
			\proj{\sigma[\eloc \mapsto
                          \eupdate~\sigma(\eloc)~v] \sepidx{} ()}{i}$
\\ By \rulename{assign}:
\\\> (2)\>  $\Psi \vdash \proj{\sigma}{i} \sepidx{} \eloc := \proj{v}{i} \mapsto 
				\proj{\sigma}{i}[\eloc \mapsto \eupdate~\proj{\sigma}{i}(\eloc)~\proj{v}{i}] \sepidx{} ()$
\\ T. S. $\proj{\sigma[\eloc \mapsto \eupdate~\sigma(\eloc)~v]}{i} =
\proj{\sigma}{i}[\eloc \mapsto \eupdate~\proj{\sigma}{i}(\eloc)~\proj{v}{i}] $
\\By definition of $\eupdate$
\\\>(3)\> $\sigma[\eloc \mapsto \eupdate~\sigma(\eloc)~v] =
\sigma[\eloc \mapsto v]$
\\\>(4)\> $\proj{\sigma}{i}[\eloc \mapsto
\eupdate~\proj{\sigma}{i}(\eloc)~\proj{v}{i}] 
 =\proj{\sigma}{i}[\eloc \mapsto \proj{v}{i}] $
\\By $v$ is a valid extended \langname expression
\\\>(5)\> $v$ does not contain $\bullet$
\\By the definition of projection
\\\>(6)\>$\proj{\sigma[\eloc \mapsto v]}{i} = \proj{\sigma}{i}[\eloc \mapsto \proj{v}{i}] $
\end{tabbing}

\item [Case:] $\ee$ ends in \rulename{new} 
\begin{tabbing}
By assumption
\\~~\=(1)~~~~\= $\Psi \vdash \sigma\sepidx{} \enew(v) 
\stepsto  \sigma[\eloc\mapsto \enew~v]\sepidx{} \eloc$
\\T.S. $\Psi \vdash \proj{\sigma \sepidx{} \enew(v)}{i} \stepsto 
					\proj{\sigma[\eloc \mapsto
                                          \enew~v] \sepidx{}
                                          \eloc}{i}$
\\By \rulename{new}
\\\>(2)\>$\Psi \vdash \proj{\sigma}{i}\sepidx{} \proj{\enew(v)}{i} 
\stepsto  \proj{\sigma}{i}[\eloc\mapsto \enew~\proj{v}{i}]\sepidx{}
\eloc$
\\T.S.  $\proj{\sigma[\eloc \mapsto\enew~v]}{i} =
\proj{\sigma}{i}[\eloc\mapsto \enew~\proj{v}{i}]$
\\By definition of $\enew$
\\\>(3)\> $\sigma[\eloc \mapsto\enew~v]  =\sigma[\eloc \mapsto v]$
\\\>(4)\> $\proj{\sigma}{i}[\eloc\mapsto \enew~\proj{v}{i}] =
\proj{\sigma}{i}[\eloc\mapsto \proj{v}{i}] $
\\By $v$ is a valid extended \langname expression
\\\>(5)\> $v$ does not contain $\bullet$
\\By the definition of projection
\\\>(6)\>$\proj{\sigma[\eloc \mapsto v]}{i} = \proj{\sigma}{i}[\eloc \mapsto \proj{v}{i}] $
\end{tabbing}

\item [Case:] $\ee$ ends in \rulename{Let}
\begin{tabbing}
By assumption:
\\~~\=(1)~~~~\= $\Psi \vdash \sigma\sepidx{} \elet\ x= v\ \ein\ e \stepsto
\sigma\sepidx{} e[x\Leftarrow v]$
\\ By \rulename{Let} 
\\\> (2)\> $\Psi \vdash \proj{\sigma}{i} \sepidx{} \elet\ x= \proj{v}{i}\ \ein\ \proj{e}{i} \stepsto
\proj{\sigma}{i}\sepidx{} \proj{e}{i}[x\Leftarrow \proj{v}{i}]$
\\ By Lemma \ref{lem:proj-distributed}, 
\\\> (3)\> $\proj{e[x\Leftarrow v]}{i}=\proj{e}{i}[x\Leftarrow \proj{v}{i}]$
\end{tabbing}
\end{description}
\end{proof}

\begin{lem}[Projected run]
\label{lem:i-proj-step}
If $\ee::\Psi \vdash \proj{\sigma}{i}\sepidx{} e \stepsto \sigma'\sepidx{} e'$ where
$e$ is a core \langname constructs, $i\in\{1,2\}$, then 
$\Psi \vdash \sigma\sepidx{i} e \stepsto \sigma''\sepidx{i} e'$ and 
$\proj{\sigma''}{i}=\sigma'$. 
\end{lem}
\begin{proofsketch}
By induction over the structure of $\ee$. For all the cases, we can
apply the same evaluation rule of $\ee$. 
\end{proofsketch}

\begin{lem}[Projected execution completeness]
\label{lem:i-exec-step}
If  $\Psi \vdash \proj{\sigma\sepidx{} e}{i} \stepsto
\sigma'\sepidx{} e'$ where $i\in\{1, 2\}$,  then exists $\sigma_1$,
$e_1$, and $k\in\{1,2\}$ s.t. $\Psi \vdash \sigma\sepidx{} e \stepsto^k
\sigma_1 \sepidx{}e_1$ and $\proj{\sigma_1\sepidx{}e_1}{i} = \sigma'\sepidx{}e'$.
\end{lem}
\begin{proof}
By induction over the structure of $e$. For most cases, we consider
one of the following three cases: the \rulename{Context} rule applies,
a reduction applies, or a lift rule applies. We show one example case below. We
also show the special case when $e$ is a pair. 
\begin{description}
\item[Case:] $e=v\; e_1$ 
\begin{tabbing}
By assumption
\\~~\=(1)~~~\= $\proj{e}{i} = v_2\;e_2$ and $v_2=\proj{v}{i}$,
$\proj{e_1}{i}=e_2$
\\
\\\>{\bf Subcase a:} \rulename{context} applies
\\\>By assumption
\\\>~~\=(a1)~~~\= $\Psi \vdash \proj{\sigma}{i}\sepidx{} v_2\;e_2 \stepsto
 \proj{\sigma}{i}\sepidx{} v_2\; e'_2$ and 
\\\>\>(a2)\> $\Psi \vdash \proj{\sigma}{i}\sepidx{} e_2 \stepsto
 \proj{\sigma}{i}\sepidx{} e'_2$
\\\> By I.H. on $e_1$
\\\>\>(a3)\> exists $k$, $\sigma_2$, $e'_1$, s.t. 
 $\Psi \vdash \sigma\sepidx{} e_1 \stepsto^k
\sigma_2 \sepidx{}e'_1$ and $\proj{\sigma_2\sepidx{}e'_1}{i} =
\proj{\sigma}{i} \sepidx{}e'_2$.
\\\>By applying \rulename{context}
\\\>\>(a4)\> $\Psi \vdash \sigma\sepidx{} v\;e_1 \stepsto^k
\sigma_2 \sepidx{}v\;e'_1$
\\\>By projection and (1), (a3)
\\\>\>(a5)\>  $\proj{\sigma_2\sepidx{} v\;e'_1}{i} =
\proj{\sigma}{i} \sepidx{}v_2\;e'_2$.
\\
\\\>{\bf Subcase b:} \rulename{app} applies
\\\>By assumption
\\\>\>(b1)\>$v_2 = f$, $e_2 = v_3$, and exists $v_1$, $e_1=v_1$
\\\>\>(b2)\> $\Psi \vdash \proj{\sigma}{i}\sepidx{} v_2\;e_2 \stepsto
 \proj{\sigma}{i}\sepidx{} e_3[x\Leftarrow v_3]$ and 
\\\>\>(b3)\> $\Psi=\Psi', f(x)=e_3$
\\\>There are two cases: (I) $v=\epair{f_1}{f_2}$ and (II) $v=f$
\\\> For (II), we can apply the \rulename{app} rule, and use
Lemma~\ref{lem:proj-distributed}. 
\\\> We show details of proof of (I) below. 
\\\> By \rulename{Lift-App}  rule
\\\>\>(b4)\> $\Psi \vdash \sigma\sepidx{} \epair{f_1}{f_2}\;v_1 \stepsto
\sigma \sepidx{} \epair{f_1 \;\proj{v_1}{1} }{f_2 \;\proj{v_1}{2} }$
\\\> We show the case $i=1$ and the other case can be proven
similarly.
\\\>By (1) and (b1)
\\\>\>(b5)\> $f_1=f$ and $v_1 = v_3$
\\\>By \rulename{app}
\\\>\>(b6)\> $\Psi \vdash \sigma \sepidx{i} f\;v_3 \stepsto
 \sigma\sepidx{i} e_3[x\Leftarrow v_3]$
\\\>By \rulename{pair}
\\\>\>(b7)\> $\Psi \vdash \sigma\sepidx{} \epair{f_1 \;\proj{v_1}{1}
}{f_2 \;\proj{v_1}{2}}
\stepsto \sigma \sepidx{} \epair{ e_3[x\Leftarrow v_3]}{f_2
  \;\proj{v_1}{2}}$
\\By (b4) and (b7)
\\\> the conclusion holds
\end{tabbing}

\item[Case:] $e=\epair{e_1}{e_2}$
\begin{tabbing}
By assumption
\\~~\=(1)~~~\= $\Psi \vdash \proj{\sigma}{i}\sepidx{} e_i \stepsto
\sigma'\sepidx{} e'$
\\We prove the case when $i=1$, the other case is similar
\\By $e_1$ is a core \langname construct and Lemma~\ref{lem:i-proj-step}
\\\>(2)\>$\Psi \vdash \sigma\sepidx{1} e_1 \stepsto
\sigma''\sepidx{1} e'$ and $\sigma'= \proj{\sigma''}{1}$
\\By \rulename{pair} and (2)
\\\>(3)\> $\Psi \vdash \sigma \sepidx{} \epair{e_1}{e_2} \stepsto
\sigma''\sepidx{} \epair{e'}{e_2}$
\end{tabbing}
\end{description}
\end{proof}

\begin{thm}[Completeness] \label{thm:polc-pair-op-completeness}
If for all $i\in\{1, 2\}$, $\Psi \vdash \proj{\sigma\sepidx{} \extE}{i} \stepsto^{n_i}
\sigma_i\sepidx{} v_i$, then exists $\sigma'$, $v'$, s.t. $\Psi \vdash \sigma\sepidx{} e \stepsto^*
\sigma' \sepidx{}v'$ and for all $i\in\{1, 2\}$, $\proj{\sigma'\sepidx{}v'}{i} = \sigma_i\sepidx{}v_i$.
\end{thm}
\begin{proof}
By induction over $n_1+n_2$.
\begin{description}
\item[Base case] $n_1+n_2=0$ 
\begin{tabbing}
By assumption
\\~~\=(1)~~~\=$\proj{\extE}{i} = v_i$ for $i\in\{1,2\}$
\\By (1) and the definition of projection
\\\>(2)\> $\extE$ is a value.
\end{tabbing} 
\item[Inductive case] $n_1+n_2=k+1$
\begin{tabbing}
By assumption, at least one of the projections takes a step. 
\\We show one
case and the other can be proven similarly. 
\\~~\=(1)~~~~\=$\Psi\vdash \proj{\sigma \sepidx{} \extE}{1} \stepsto\sigma'_1
\sepidx{}e'_1\stepsto^{n_1-1}\sigma_1 \sepidx{}v_1$
\\\>(2)\> $\Psi\vdash \proj{\sigma \sepidx{} \extE}{2}\stepsto^{n_2}\sigma_2 \sepidx{}v_2$
\\By Lemma~\ref{lem:i-exec-step}, 
\\\>(3)\>exists $k\in\{1,2\}$, $\sigma'$ and $\extE_1$ s.t. $\Psi\vdash \sigma
\sepidx{} \extE \stepsto^k \sigma' \sepidx{} \extE_1$ 
\\
\\\>{\bf Subcase I:}  $k=1$
\\\>By the evaluation of a core \langname term is deterministic
\\\>~~\=(I1)~~~ $\sigma'_1 \sepidx{}e'_1 = \proj{\sigma' \sepidx{} \extE_1}{1}$
\\\>By Theorem~\ref{thm:pair-op-soundness} and (2), we have two cases
\\\>\>{\bf Subcase a:}  
\\\>\>~~\=(a1)~~\=$\Psi\vdash \proj{\sigma \sepidx{} \extE}{2} 
   \stepsto \proj{\sigma' \sepidx{} \extE_1}{2}$ 
\\\>\>By the evaluation of a core \langname term is deterministic
\\\>\>\>(a2)\> $\Psi\vdash \proj{\sigma \sepidx{} \extE}{2} \stepsto\sigma'_2
\sepidx{}e'_2\stepsto^{n_2-1}\sigma_2 \sepidx{}v_2$
\\\>\>\>(a3)\>$\sigma'_2 \sepidx{}e'_2= \proj{\sigma' \sepidx{}
  \extE_1}{2}$
\\\>\>By I.H. (1), (I1), (a2), (a3)
\\\>\>\>(a4)\> $\sigma''$, $v'$, s.t. $\Psi \vdash \sigma'\sepidx{} \extE_1 \stepsto^*
\sigma'' \sepidx{}v'$ 
\\\>\>\>(a5)\>and for all $i\in\{1, 2\}$,
$\proj{\sigma''\sepidx{}v'}{i} = \sigma_i\sepidx{}v_i$
\\\>\>By (a4) and (3), the conclusion holds
\\\>\>{\bf Subcase b:}   
\\\>\>\> (b1)\>and $\proj{\sigma \sepidx{} \extE}{2} 
   =\proj{\sigma' \sepidx{} \extE_1}{2}$
\\\>\> By I.H. (1), (I1), (2), (b1)
\\\>\>\>(b2)\> $\sigma''$, $v'$, s.t. $\Psi \vdash \sigma'\sepidx{} \extE_1 \stepsto^*
\sigma'' \sepidx{}v'$ 
\\\>\>\>(b3)\>and for all $i\in\{1, 2\}$,
$\proj{\sigma''\sepidx{}v'}{i} = \sigma_i\sepidx{}v_i$
\\\>\>By (b2) and (3), the conclusion holds
\\
\\\>{\bf Subcase II:}  $k=2$
\\\> \parbox{.7\textwidth}{The proof is similar to the previous case. We need to case on
 whether the projection of the configuration to the right execution makes a
step or remains the same. Finally invoke I.H. }
\end{tabbing}
\end{description}
\end{proof}

\subsection{Summary of Typing Rules for Paired \langname}
\label{app:polc:typing}
First we define subtyping relations and policy operations below. 
\begin{mathpar}
\inferrule
	{  S_1 \sqsubseteq_S S_2 \\
          I_1 \sqsubseteq_I I_2}
	{(S_1, I_1) \sqsubseteq (S_2, I_2)}
\and
\inferrule
	{~}
	{\bot \sqsubseteq \rho}
\and
\inferrule
	{~}
	{\rho\sqsubseteq \top}
\and
\inferrule{
  \lab_1\sqsubseteq \lab_2\\
    \pol_1 \sqsubseteq \pol_2
} {
  \lab_1::\pol_1 \sqsubseteq \lab_2::\pol_2
}
\end{mathpar}

\begin{mathpar}
%
\inferrule
	{~}
	{\rho_1 \sqcup \rho_2= \rho_2\sqcup\rho_1}
\and
\inferrule
	{~}
	{\rho \sqcup \bot= \rho}
\and
\inferrule
	{~}
	{\rho \sqcup \top = \top}
\and
\inferrule{
  \lab = \lab_1\sqcup \lab_2\\
  \pol =  \pol_1 \sqcup \pol_2
} {
  \lab_1::\pol_1 \sqcup \lab_2::\pol_2 = \lab::\pol
}
\end{mathpar}

\noindent\framebox{$b\leq b'$}~~\framebox{$t\leq t'$}~~\framebox{$s\leq s'$}
\begin{mathpar}
\inferrule*[right=$\leq$Refl]{ 
}{b \leq b} 
\and
\inferrule*[right=$\leq$Trans]{ 
  b \leq b'\\
  b'\leq b''
}{
  b \leq b''
} 
%
%
\end{mathpar}
\begin{mathpar}
\inferrule*[right=$\leq$Unit]{ 
}{
  \tunit\leq \tunit
}
%
\and
\inferrule*[right=$\leq$Pol]
	{b \leq b' \\
	 \pol_1 \sqsubseteq \pol_2}
	{b\ {\pol_1} \leq b'\ {\pol_2}}
\and
\inferrule*[right=$\leq$Fun]{
  \pc' \sqsubseteq \pc \\
 t'_1  \leq t_1\\
 t_2 \leq t'_2\\
\rho\sqsubseteq \rho'
 }{
    [\pc] (t_1 \rightarrow t_2)^\rho\leq  [\pc'] (t'_1 \rightarrow t'_2)^{\rho'}
}\end{mathpar}

\noindent\framebox{$\rho\rhd s$}
\begin{mathpar}
\inferrule*{  }{
  \rho\rhd\tunit
}
\and
\inferrule*{
  \pol'\sqsubseteq \pol
}{
  \pol' \rhd  b\ {\pol}
}
\and
\inferrule*{
\rho\sqsubseteq \rho'
 }{
    \rho\rhd  [\pc] (t_1 \rightarrow t_2)^{\rho'}
}\end{mathpar}

Figure~\ref{fig:val-typing} and~\ref{fig:exp-typing}  summarize typing rules for extended \langname.

\begin{figure*}[t!]
\flushleft
\noindent\framebox{$D;F;\Sigma; \Gamma \vdash v : s$}
\begin{mathpar}
\inferrule*[right=P-T-V-Int]{
}{ 
  D; F;\Sigma; \Gamma  \vdash n : \tpint\ \pol
}
\and
\inferrule*[right=P-T-V-Loc]{
}{ 
  D; F;\Sigma; \Gamma  \vdash \eloc: \storetp(x)\ \pol 
}
\and
\inferrule*[right=P-T-V-Var]{
}{ 
  D; F; \Sigma; \Gamma  \vdash x : \Gamma(x) 
}
\and
\inferrule*[right=P-T-V-Fun]{
}{ 
  D; F; \Sigma; \Gamma  \vdash f : F(f)\ \pol
}
\and
\inferrule*[right=P-T-V-Struct]{
T\mapsto \tstruct\ T\ \{s_1,\cdots, s_n\} \in D
\\\forall i,  D; \Sigma; \Gamma  \vdash v_i: s_i
}{ 
  D; F; \Sigma; \Gamma \vdash (T)\ \{v_1,\cdots,v_n\} :
  T\ \pol
}
\and
\inferrule*[right=P-T-V-Sub]{
  D; F;\Sigma; \Gamma  \vdash v : s'\\
 s'\leq s
}{ 
  D; F;\Sigma; \Gamma \vdash v: s
}
\and
\inferrule*[right=P-T-V-Pair]{ 
  D; F;\Sigma; \Gamma  \vdash v_1 : s
\\   D; F;\Sigma; \Gamma  \vdash v_2 : s
\\ \pol \rhd s 
\\ \pol\in\ H
}{ 
  D; F;\Sigma; \Gamma  \vdash \epair{v_1}{v_2} : s
}
\end{mathpar}
\caption{Typing Rules for Values in Extended \langname}
\label{fig:val-typing}
\end{figure*}

\begin{figure*}[t!]
\flushleft
\noindent\framebox{$D;F; \Sigma; \Gamma ; \pc \vdash e : s$}
\begin{mathpar}
\inferrule*[right=P-T-E-Val]{
  D;F; \Sigma; \Gamma  \vdash v : s
}{ 
  D;F; \Sigma; \Gamma ; \pc \vdash v: s\join\pc
}
\and
\inferrule*[right=P-T-E-Field]{
  D;F; \Sigma; \Gamma \vdash v : T\
  \pol
\\ T\mapsto \tstruct\ \{s_1,\cdots, s_n\} \in D
\\ \pc\sqsubseteq\pol
}{ 
  D;F; \Sigma; \Gamma ; \pc \vdash v.i : s_i \join\ \pol 
}
\and
\inferrule*[right=P-T-E-New]{
  D;F; \Sigma; \Gamma ; \pc \vdash e : s
\\ \pc\rhd \pol
}{ 
  D;F; \Sigma; \Gamma ; \pc \vdash \enew(e): \tptr(s)\ \pol
}
\and
\inferrule*[right=P-T-E-Deref]{
  D;F; \Sigma; \Gamma \vdash v : \tptr(s)\ \pol
\\ \pc\sqsubseteq\pol
}{ 
  D;F; \Sigma; \Gamma ; \pc \vdash *v: s\join \pol
}
\and
\inferrule*[right=P-T-E-Assign]{
  D;F; \Sigma; \Gamma  \vdash v_1 : \tptr(s)\ \pol
 \\  D;F; \Sigma; \Gamma;\pc \vdash e_2 : s
\\  \pol \rhd s
}{ 
  D;F; \Sigma; \Gamma ; \pc \vdash v_1:= e_2: \tunit
}
\and
\inferrule*[right=P-T-E-App]{
  D;F; \Sigma; \Gamma \vdash v_f :  [\pc'](t_1 \rightarrow t_2)^\pol
\\  D;F; \Sigma; \Gamma;\pc \vdash e_a :  t_1
\\ \rho\sqcup\pc \sqsubseteq \pc'
}{ 
  D;F; \Sigma; \Gamma ; \pc \vdash v_f\ e_a: t_2
}
\and
\inferrule*[right=P-T-E-Let]{
  D;F; \Sigma; \Gamma ; \pc \vdash e_1 : s_1
\\   D;F; \Sigma; \Gamma, x:s_1 ; \pc \vdash e_2 : s_2
}{ 
  D;F; \Sigma; \Gamma ; \pc \vdash \elet\ x:s_1=e_1\ \ein\ e_2: s_2
}
\and
\inferrule*[right=P-T-E-If]{
  D;F; \Sigma; \Gamma  \vdash v_1 : \tpint\ \pol
\\   D;F; \Sigma; \Gamma ; \pc\sqcup \pol  \vdash e_2 : s
\\   D;F; \Sigma; \Gamma ; \pc\sqcup \pol  \vdash e_3 : s
}{ 
  D;F; \Sigma; \Gamma ; \pc \vdash \eif\ v_1\ \ethen\ e_2\ \eelse\
  e_3 : s
}
\and
\inferrule*[right=P-T-E-DE]{
  D;F; \Sigma; \Gamma \vdash v_f: (\dne)  [\pc'] (b\ \lab_1{::}\top
  \rightarrow b\ \lab_2{::}\bot)^{\rho_f}
\\  D;F; \Sigma; \Gamma;\pc \vdash e_a :  b\ \pol
\\ \pol = \lab_1{::}\lab_2{::}\pol'
\\ \rho_f\sqcup\pc \sqsubseteq \pc'
}{ 
  D;F; \Sigma; \Gamma ; \pc \vdash v_f\ e_a: b\ \lab_2{::}\pol'
}
\and
\inferrule*[right=P-T-E-Relabel]{ 
  D;F; \Sigma; \Gamma  \vdash v : b\ \pol\\
\pc\sqsubseteq \pol'
}{ 
  D;F; \Sigma; \Gamma ; \pc \vdash \erelab(\pol'\Leftarrow\pol)\ v
  : b\ \pol'
}
\and
\inferrule*[right=P-T-E-Sub]{
  D;F; \Sigma; \Gamma ; \pc \vdash e : s'\\
 s'\leq s
}{ 
  D;F; \Sigma; \Gamma ; \pc \vdash e: s
}
\and
\inferrule*[right=P-T-E-Pair]{ 
  D;F; \Sigma; \Gamma;\pc\sqcup\pol'  \vdash e_1 : s
\\   D;F; \Sigma; \Gamma;\pc\sqcup\pol' \vdash e_2 : s
\\ \pol \rhd s 
\\ \pol\in\ H
\\ \pol'\in\ H
}{ 
  D;F; \Sigma; \Gamma; \pc  \vdash \epair{e_1}{e_2} : s
}
\end{mathpar}
\caption{Typing Rules for Expressions in Extended \langname}
\label{fig:exp-typing}
\end{figure*}

\subsection{Preservation}

Next we present the lemmas and proofs for the Preservation Theorem. We define $\Sigma\leq\Sigma'$ as $\Sigma' = \Sigma, \Sigma''$.

\begin{lem}\label{lem:guard-weakening}
If $\pol'\sqsubseteq \pol$, $\pol\rhd s$, $s\leq s'$, then $\pol'\rhd s'$
\end{lem}
\begin{proofsketch}
By examining $\pol\rhd s$ and $s\leq s'$. 
\end{proofsketch}

\begin{lem}\label{lem:polc-pc-weakening}
If $\ee::D;F;\Sigma; \Gamma; \pc \vdash  e: s$ and $\pc'\sqsubseteq \pc$
then $D;F;\Sigma; \Gamma; \pc' \vdash  e: s$.
\end{lem}
\begin{proofsketch}
By induction over the structure of $\ee$. We use
Lemma~\ref{lem:guard-weakening} in cases where $\pc$ is used in the
premises. 
\end{proofsketch}

\begin{lem}\label{lem:type-store-weakening}
\begin{enumerate}
\item If $\ee::D;F;\Sigma; \Gamma \vdash  v: s$ and $\Sigma\leq \Sigma'$
then $D;F;\Sigma'; \Gamma \vdash  v: s$.
\item If $\ee::D;F;\Sigma; \Gamma; \pc \vdash  e: s$ and $\Sigma\leq \Sigma'$
then $D;F;\Sigma'; \Gamma; \pc\vdash  e: s$.
\end{enumerate}
\end{lem}
\begin{proofsketch}
By induction over the structure of $\ee$. We use
Lemma~\ref{lem:guard-weakening} in cases where $\pc$ is used in the
premises. 
\end{proofsketch}

\begin{lem}[Projection well-typed]\label{lem:polc-projection-wf}
If $\ee::D; \Sigma;\Gamma \vdash v: s$ then $\forall i\in\{1,2\}$, $D; \Sigma;\Gamma \vdash \proj{v}{i}: s$ 
\end{lem}
\begin{proofsketch}
By induction over the structure of $\ee$. 
\end{proofsketch}

\begin{lem}[Substitution]\label{lem:polc-substitution}
~\\
\begin{enumerate}
\item If  $\ee::D;F;\Sigma;\Gamma, x{:} s \vdash v' : s'$ and 
$D;;F;\Sigma;\Gamma\vdash v : s$ 
then  $D;F;\Sigma;\Gamma \vdash v'[x\Leftarrow v] : s'$ 
\item 
If  $\ee::D;F;\Sigma;\Gamma,x{:}s ; \pc \vdash e : s'$
and $D;F;\Sigma;\Gamma\vdash v : s$ 
then  $D;F;\Sigma;\Gamma; \pc \vdash e[x\Leftarrow v] : s'$
\end{enumerate}
\end{lem}
\begin{proofsketch}
By induction over the structure of $\ee$.
\end{proofsketch}

\begin{lem}\label{lem:subtyping-form}
~\\
\begin{enumerate}
\item If $s\leq T\ \pol$ then $s=T\ \pol'$
\item If $s\leq \tptr(s')\ \pol$ then $s=\tptr(s')\ \pol'$
\item If $s\leq  [\pc_f](t_1\rightarrow t_2)^\pol$ then $s=
  [\pc'_f](t'_1\rightarrow t'_2)^{\pol'}$ and $\pc_f\sqsubseteq
  \pc'_f$, $t_1\leq t'_1$ and $t'_2\leq t_2$.
\end{enumerate}
\end{lem}
\begin{proofsketch}
By induction over the derivation $s\leq s'$.
\end{proofsketch}

\begin{lem}[Inversion]\label{lem:polc-inversion}
~\\
\begin{enumerate}
\item If $D;F;\cdot \vdash (T)\{v_1,\cdots,v_n\}: T\ \pol$, then  
 $D(T) = \{s_1,\cdots, s_2\}$ and $\forall i\in[1,n]$, $D;F;\cdot
 \vdash v_i: s_i$.
\item If  $D;F;\cdot \vdash \eloc: \tptr(s)\ \pol$, $\eloc\in\m{dom}(\sigma)$,
  and $D;F\vdash \sigma: \Sigma$, then $D;F;\cdot; \pc \vdash \sigma(\eloc): s$.
\item If  $D;F;\cdot\vdash f: [\pc_f](t_1\rightarrow t_2)^\pol$, $f(x)=e\in\m{dom}(\Psi)$
  and $D;F\vdash \Psi$, then $D;F;x: t_1; \pc_f \vdash e: t_2$.
\item If $D;F;\cdot\vdash \epair{v_1}{v_2}: s$ then $\forall
  i\in[1,2]$, $D;F;\cdot \vdash v_i: s$ and $\exists \pol$,
  s.t. $\pol\rhd s$ and $\pol\in H$.
\end{enumerate}
\end{lem}
\begin{proofsketch}
By induction over the typing derivation. 
\end{proofsketch}

\begin{lem}[Value is typed w/o PC]
\label{lem:val-no-pc}
If $\ee::D;F;\Sigma; \Gamma;\pc \vdash  v: s$ 
then $D;F;\Sigma; \Gamma \vdash  v: s$.
\end{lem}
\begin{proofsketch}
By induction over the structure of $\ee$. 
In the cases of \rulename{E-Sub} and \rulename{E-Pair}, we directly apply
I.H. and then apply the rule with the same name in value typing. 
In the case of \rulename{E-Val}, we apply \rulename{V-Sub}.
\end{proofsketch}

\begin{lem}[Store]\label{lem:store-ops}
  For all $i\in\{1,2,\bullet\}$, if $D;F;\Sigma; \Gamma \vdash v: s$
  and $i\in\{1,2\}$ implies exists $\pol\in H$ s.t. $\pol\rhd s$; then
  $D;F;\Sigma; \Gamma \vdash \enew_i\ v: s$ and and for all $v'$
  s.t. $D;F;\Sigma; \Gamma \vdash v': s$,
  $D;F;\Sigma; \Gamma \vdash \eupdate_i\ v\ v': s$.
\end{lem}
\begin{proofsketch}
By examining the definitions of these operations.
\end{proofsketch}

\begin{lem}[Value Has Flexible Label]
\label{lem:val-flex-lab}
Given a set of high labels $H$,
if $\ee::D;F;\Sigma; \cdot \vdash  v: b\ \pol$ 
and $\pol\in H$ iff $\pol' \in H$
then $D;F;\Sigma; \cdot \vdash  v: b\ \pol'$.
\end{lem}
\begin{proofsketch}
By induction on the structure of $v$. The value typing rules assign
an arbitrary $\pol$ to the type of core \langname values. In the
case of pairs, the assumption that $\pol\in H$ iff $\pol' \in H$
allows us to apply \rulename{V-Pair} rule.
\end{proofsketch}

 \begin{lem}\label{lem:pc-lower}
 If $\ee::D;F;\Sigma; \Gamma; \pc \vdash  e: s$ then $\pc\rhd s$.
 \end{lem}
\begin{proofsketch}
By induction over the structure of $\ee$.
\end{proofsketch}

\begin{lem}[Preservation]~\label{lem:polc-i-preservation}
If $\Psi \vdash \sigma\sepidx{i} e \stepsto \sigma'\sepidx{i} e'$,
$D;F\vdash\Psi$, $D;F\vdash \sigma:\Sigma$ and
$D;F;\Sigma; \cdot; \pc \vdash  e: s$, and $i\in\{1,2\}$ implies $\pc\in H$ then
exists $\Sigma'\geq\Sigma$ s.t. $D;F\vdash \sigma':\Sigma'$ and 
$D;F;\Sigma'; \cdot; \pc \vdash  e': s$.
\end{lem}
\begin{proof}
  By induction over the structure of $\ee$. The proofs are mostly
  standard and use Lemma~\ref{lem:polc-substitution} and 
  \ref{lem:polc-inversion}. We only show cases where
  information flow labels or pairs are involved.
\begin{description}

\item[Case:] $\ee$ ends in \rulename{E-ReLabel}
\begin{tabbing}
By assumption
\\~~\=(1)~~~~\= $D;F; \Sigma; \Gamma ; \pc \vdash \erelab(\pol'\Leftarrow\pol)\ v
  : b\ \pol'$
\\\>(2)\>$\ee'::D;F; \Sigma; \Gamma  \vdash v : b\ \pol$ and $\pc\sqsubseteq\pol'$
\\By examining the operational semantic rules, there are two subcases
\\{\bf Subcase a:} $v$ is not a pair
\\\>~~\=(a3)~~~~\= $\Psi \vdash \sigma\sepidx{i}  \erelab(\pol'\Leftarrow\pol) v
\stepsto  \sigma\sepidx{i} v$
\\\>\>By the definition of $H$
\\\>\>(a4)\>  $\pol\in H$ iff $\pol';\in H$
\\By Lemma~\ref{lem:val-flex-lab}, $\ee'$, and (a4)
\\\>\>(a5)\> $D;F; \Sigma'; \Gamma \vdash v: b\ \pol'$
 \\{\bf Subcase b:} $v=\epair{v_1}{v_2}$
\\\>~~\=(b3)~~~~\= $\Psi \vdash \sigma\sepidx{}  \erelab(\pol'\Leftarrow\pol) \epair{v_1}{v_2}
\stepsto  \sigma\sepidx{} \epair{ \erelab(\pol'\Leftarrow\pol)  v_1}{ \erelab(\pol'\Leftarrow\pol)  v_2}$
\\\>By Lemma~\ref{lem:polc-inversion} and $\ee'$
\\\>\>(b4)\>$\forall  i\in[1,2]$, $D;F;\cdot \vdash v_i: b\ \pol$ and
\\\>\>(b5)\> $\exists \pol''$,
  s.t. $\pol''\rhd b\ \pol$ and $\pol''\in H$ 
\\\>By (b4) and \rulename{E-Relab}
\\\>\>(b6)\> $\forall  i\in[1,2]$, $D;F;\cdot;\pc\join\pol' \vdash \erelab(\pol'\Leftarrow\pol)  v_i: b\ \pol'$ 
\\\>By the definition of $H$ and (b5)
\\\>\>(b7)\> $\pol\in H$ and $\pol'\in H$
\\\>By \rulename{E-Pair}, (b6), (b7)
\\\>\>(b8)\> $D;F; \Sigma; \Gamma ; \pc \vdash \epair{
  \erelab(\pol'\Leftarrow\pol)  v_1}{ \erelab(\pol'\Leftarrow\pol)
  v_2}: b\ \pol'$ 
\end{tabbing}

\item[Case:] $\ee$ ends in \rulename{E-If}
\begin{tabbing}
By assumption
\\~~\=(1)~~~~\= $D;F; \Sigma; \Gamma ; \pc \vdash \eif\ v\ \ethen\ e_2\ \eelse\  e_3 : s$
\\\>(2)\>$\ee':: D;F; \Sigma; \Gamma  \vdash v : \tpint\ \pol $
\\\>(3)\> and $\ee_2::D;F; \Sigma; \Gamma ; \pc\sqcup \pol  \vdash e_2
: s$
\\\>(4)\> and $\ee_3::D;F; \Sigma; \Gamma ; \pc\sqcup \pol  \vdash e_3 : s$
\\\parbox{.8\textwidth}{By examining the operational semantic rules,
  there are two subcases: $v$ is not a pair and $v$ is a pair.
We only show the case when  $v=\epair{v_1}{v_2}$}
\\\>(5)\>$\Psi \vdash \sigma\sepidx{}  \eif\ \epair{v_1}{v_2}\ \ethen\ e_2\ \eelse\  e_3 
\stepsto  \sigma\sepidx{} \epair{ \eif\ v_1\  \ethen\proj{\ e_2}{1}\ \eelse\
 \proj{e_3}{1} }{ \eif\ v_2\  \ethen\ \proj{e_2}{2}\ \eelse\  \proj{e_3}{2} }$
\\By Lemma~\ref{lem:polc-inversion} and $\ee'$
\\\>(6)\>$\forall  i\in\{1,2\}$, $D;F;\cdot\vdash v_i: \tpint\ \pol$
and 
\\\>(7)\> $\exists \pol''$,
  s.t. $\pol''\rhd \tpint\ \pol$ and $\pol''\in H$ 
\\By (7) and the definition of $H$
\\\>(8)\> $\pol\in H$
\\By Lemma~\ref{lem:pc-lower}, \ref{lem:guard-weakening} and $\ee_2$
\\\>(9)\> $\pol\rhd s$ 
\\By \rulename{E-Val} and (6)
\\\>(10)\>$\forall  i\in[1,2]$, $D;F;\cdot; \pc\join\pol \vdash v_i:
\tpint\ \pc\join\pol$
\\By Lemma~\ref{lem:polc-projection-wf} and (3), (4)
\\\>(11)\>$D;F; \Sigma; \Gamma ; \pc\sqcup \pol  \vdash \proj{e_k}{m}
: s$ where $k\in\{2,3\}$ and $m\in\{1,2\}$
\\By \rulename{E-If} and (3), (4), and (11)
\\\>(12)\>  $D;F; \Sigma; \Gamma ; \pc\join\pol \vdash \eif\ v_i\
\ethen\proj{\ e_2}{i}\ \eelse\ \proj{e_3}{i} : s$ where   $i\in\{1,2\}$
\\By \rulename{E-Pair}, (8), (9), and (12)
\\\>(13)\> $D;F; \Sigma; \Gamma ; \pc \vdash \epair{ \eif\ v_1\  \ethen\proj{\ e_2}{1}\ \eelse\
 \proj{e_3}{1} }{ \eif\ v_2\  \ethen\ \proj{e_2}{2}\ \eelse\  \proj{e_3}{2} }: s$

\end{tabbing}

\item[Case:] $\ee$ ends in \rulename{E-DE}
\begin{tabbing}
By assumption
\\~~\=(1)~~~~\= $D;F; \Sigma; \Gamma ; \pc \vdash v_f\ e_a: b\ \lab_2{::}\pol'$
\\\>(2)\>$\ee':: D;F; \Sigma; \Gamma\vdash v_f: (\dne)  [\pc'] (b\ \lab_1{::}\top
  \rightarrow b\ \lab_2{::}\bot)^{\rho_f}$
\\\>(3)\>$\ee'':: D;F; \Sigma; \Gamma;\pc \vdash e_a :  b\ \pol$
\\\>(4)\> and $\pol = \lab_1{::}\lab_2{::}\pol'$, $\rho_f\sqcup\pc
\sqsubseteq \pc'$
\\By examining the operational semantic rules, there are three
subcases
\\{\bf Subcase a:} $e_a$ is not a value. This is a standard case and
we omit. 
\\{\bf Subcase b:} $e_a=v_a$ and $v_f$ is not a pair
\\\>~~\=(b1)~~~~\= $\Psi \vdash \sigma\sepidx{i} v_f\ v_a
\stepsto  \sigma\sepidx{i} e[x\Leftarrow v_a][v_f\Leftarrow v_f(x)=e]$
\\\>By Lemma~\ref{lem:polc-inversion} and $D;F\vdash \Psi$
\\\>\>(b2)\> $ D;F; \Sigma; x:b\ \lab_1{::}\top, v_f: (\dne)  [\pc'] (b\ \lab_1{::}\top
  \rightarrow b\ \lab_2{::}\bot)^{\rho_f} ; \pc'\vdash e:  b\
  \lab_2{::}\bot$
\\\>By $\ee''$, Lemma~\ref{lem:val-no-pc} and \rulename{V-Sub}
\\\>\>(b3)\> $\ee'':: D;F; \Sigma; \Gamma \vdash v_a :  b\ \lab_1{::}\top$
\\\> By Lemma~\ref{lem:polc-substitution} (b2) and (b3)
\\\>\>(b4)\>$ D;F; \Sigma; \cdot; \pc'\vdash e[x\Leftarrow
  v_a][v_f\Leftarrow v_f(x)=e]:  b\
  \lab_2{::}\bot$
\\\>Lemma~\ref{lem:polc-pc-weakening} and (b4)
\\\>\>(b5)\>$ D;F; \Sigma; \cdot; \pc\vdash e[x\Leftarrow
  v_a][v_f\Leftarrow v_f(x)=e]:  b\  \lab_2{::}\bot$
\\\>By (b5) and \rulename{V-Sub}
\\\>\>(b6)\>$ D;F; \Sigma; \cdot; \pc\vdash e[x\Leftarrow
  v_a][v_f\Leftarrow v_f(x)=e]:  b\
  \lab_2{::}\pol'$
 \\{\bf Subcase c:} $e_a=v_a$ and $v_f=\epair{v_1}{v_2}$
\\\>~~\=(c1)~~~~\= $\Psi \vdash \sigma\sepidx{}  \epair{v_1}{v_2}\ v_a
\stepsto  \sigma\sepidx{} \epair{ v_1\ \proj{v_a}{1}}{v_2\ \proj{v_a}{2}}$
\\\>By Lemma~\ref{lem:polc-inversion} and $\ee'$
\\\>\>(c2)\>$\forall  i\in[1,2]$, $D;F;\cdot \vdash v_i: (\dne)  [\pc'] (c\ \lab_1{::}\top
  \rightarrow b\ \lab_2{::}\bot)^{\rho_f}$ and
\\\>\>(c3)\> $\exists \pol''$,
  s.t. $\pol''\sqsubseteq \rho_f$ and $\pol''\in H$ 
\\\>By  Lemma~\ref{lem:val-no-pc}, Lemma~\ref{lem:polc-projection-wf} and $\ee''$
\\\>\>(c4)\>$D;F; \Sigma; \Gamma \vdash \proj{v_a}{i} :  b\ \pol$
where $i\in\{1,2\}$
\\\>By (4) and (c3)
\\\>\>(c5)\>$\rho_f\sqcup\pc\join\pol'' 
\sqsubseteq \pc'$
\\\>By (c2), (c4), and (c5) and \rulename{E-De}
\\\>\>(c6)\> $\forall  i\in[1,2]$, $D;F;\cdot;\pc\join\pol'' \vdash
v_i\ \proj{v_a}{i}: b\ \lab_2{::}\pol'$ 
\\\>By the definition of $H$ and (c3)
\\\>\>(c7)\> $\pol_f\in H$ 
\\\>By \rulename{E-Pair}, (c6), (c7)
\\\>\>(c8)\> $D;F; \Sigma; \Gamma ; \pc \vdash \epair{ v_1\
  \proj{v_a}{1}}{v_2\ \proj{v_a}{2}}: b\ \lab_2{::}\pol'$
\end{tabbing}

\item[Case:] $\ee$ ends in \rulename{E-Pair}
\begin{tabbing}
By assumption
\\~~\=(1)~~~~\= $D;F; \Sigma; \Gamma; \pc  \vdash \epair{e_1}{e_2} :
s$
\\\>(2)\>$\ee'::D;F; \Sigma; \Gamma;\pc\sqcup\pol'  \vdash e_i : s$,
$i\in\{1,2\}$
\\\>(3)\> $\pol \rhd s$,  $\pol\in\ H$, and  $\pol'\in\ H$
\\By examining the operational semantic rules
\\\>(4)\> $\Psi \vdash \sigma\sepidx{} \epair{e_1}{e_2} \stepsto  \sigma'
\sepidx{} \epair{e'_1}{e'_2}$
\\\>(5)\>$\Psi \vdash \sigma\sepidx{i} e_i \stepsto  \sigma'
\sepidx{i}e'_i$, $e_j = e'_j$, and $\{i,j\} = \{1,2\}$
\\By I.H. on $\ee'$
\\\>(6)\> exists $\Sigma'\geq\Sigma$ s.t. $D;F\vdash \sigma':\Sigma'$ 
\\\>(7)\> and $D;F;\Sigma'; \cdot; \pc\sqcup\pol'  \vdash  e'_i: s$
\\By Lemma~\ref{lem:type-store-weakening} (2) and (6)
\\\>(8)\> $D;F;\Sigma'; \cdot; \pc\sqcup\pol'  \vdash  e'_j: s$
\\By \rulename{pair} (3), (7), and (8)
\\\>(9)\> $D;F; \Sigma'; \Gamma; \pc  \vdash \epair{e'_1}{e'_2} : s$
\end{tabbing}
\end{description}
\end{proof}

\begin{thm}[Preservation]~\label{thm:polc-preservation}
If $\Psi \vdash \sigma\sepidx{} e \stepsto \sigma'\sepidx{} e'$
and $\vdash \Psi; \sigma; e$ then $\vdash \Psi; \sigma'; e'$.
\end{thm}
\begin{proofsketch}
By the definitions of  $\vdash \Psi; \sigma; e $ and Lemma~\ref{lem:polc-i-preservation}.
\end{proofsketch}

\subsection{Noninterference}
\label{app:polc:noninterference}
Finally, we present proofs for the Noninterference Theorem for \langname.
\begin{mathpar}
\inferrule*{
F; R\vdash \pol_1\leadsto\pol_2
\\ F; R\vdash \pol_2\leadsto\pol_3
}{
F; R\vdash \pol_1\leadsto\pol_3
}
\and
\inferrule*{
\pol_1\sqsubseteq\pol_2
\\ F; R\vdash \pol_2\leadsto\pol_3
}{
F; R\vdash \pol_1\leadsto\pol_3
}
\and
\inferrule*{
\lab_2{::}\bot \Leftarrow \lab_1{::}\top \in R
}{
F; R\vdash \lab_1::\lab_2::\pol\leadsto\lab_2::\pol
}
\and
\inferrule*{
F = F', f: (\dne)  [\pc'] (b\ \lab_1{::}\top \rightarrow b\ \lab_2{::}\bot)^{\rho_f}
}{
F; R\vdash \lab_1::\lab_2::\pol\leadsto\lab_2::\pol
}
\end{mathpar}

\[\inferrule*{
\forall \pol', 
F; R\vdash \pol\leadsto\pol', \pol'\not\sqsubseteq\pol_A
}{
\pol_A; F; R\vdash \pol\in H
}
\]

\begin{lem}\label{lem:polc-h-closed}
If $\pol_A; F; R\vdash \pol_1\in H$ and $\pol_1\sqsubseteq\pol_2$ then $\pol_A; F; R\vdash \pol_2\in H$.
\end{lem} 

\begin{lem} \label{lem:polc-low-value-same}
\begin{enumerate}
\item If $D;F;\Sigma; \cdot \vdash v : \tpint\ \pol$, and $\pol\notin H$,
then $\proj{v}{1}=\proj{v}{2}$.
\item If $D;F; \Sigma; \cdot; \bot \vdash v : \tpint\ \pol$, and $\pol\notin H$,
then $\proj{v}{1}=\proj{v}{2}$.
\end{enumerate} 
\end{lem}
\begin{proofsketch}
By induction over the typing derivation of the value. 
\end{proofsketch}

\begin{defn}[Equivalent substitution]
We define $D;F \vdash \delta_1\approx_H \delta_2 : \Gamma$ iff for all
$x\in\m{dom}(\Gamma)$, $D;F; \cdot; \cdot; \vdash \delta_i(x):
\Gamma(x)$ ($i\in\{1,2\}$) and $\delta_1(x)=\delta_2(x)$ if $\labof(\Gamma(x)) \notin H$.
\end{defn}

\noindent\framebox{$\Gamma\vdash \delta_1\bowtie\delta_2 =\delta$}
\begin{mathpar}
\inferrule*{ }{
\Gamma\vdash \cdot\bowtie\cdot=\cdot
}
\and
\inferrule*{\Gamma\vdash \delta_1\bowtie\delta_2=\delta
\\ \labof(\Gamma(x))\in H
}{
\Gamma\vdash \delta_1, x\mapsto v_1 \bowtie \delta_2, x\mapsto v_2=
\delta, x\mapsto \epair{v_1}{v_2}
}
\and
\inferrule*{\Gamma\vdash \delta_1\bowtie\delta_2=\delta
\\ \labof(\Gamma(x))\notin H
}{
\Gamma\vdash \delta_1, x\mapsto v_1 \bowtie \delta_2, x\mapsto v_2=
\delta, x\mapsto v_1
}
\end{mathpar}

\begin{lem}\label{lem:polc-pair-substitutions}
If $D;F \vdash \delta_1\approx_H \delta_2 : \Gamma$ and $\Gamma\vdash
\delta_1\bowtie\delta_2 =\delta$, then $\forall x\in\m{dom}(\Gamma)$, 
$D;F; \cdot; \cdot \vdash \delta(x) : \Gamma(x)$.
\end{lem}
\begin{proofsketch}
By induction over the structure of $\Gamma$. 
\end{proofsketch}

\begin{thm}[Noninterference]
~\\
  If $D;F; \cdot; \Gamma; \bot \vdash e : \tpint\ \pol$, let $H$ be the set of
  labels not-observable by an attacker with label $\pol_A$, given substitution
  $\delta_1$, $\delta_2$ s.t.
  $\delta_1\approx_H \delta_2 : \Gamma$, 
and $\pol\notin H$ and
$\Psi \vdash \emptyset\sepidx{} e\delta_1 \stepsto^*
\sigma_1\sepidx{} v_1$ and 
$\Psi \vdash \emptyset \sepidx{} e\delta_2 \stepsto^*
\sigma_2\sepidx{} v_2$, then $v_1=v_2$.
\end{thm}
\begin{proof}

\begin{tabbing}
\\Let $\delta$ be the substitution from $\Gamma\vdash
\delta_1\bowtie\delta_2 =\delta$.
\\By Lemma~\ref{lem:polc-pair-substitutions}
\\~~\=(1)~~~\=$\forall x\in\m{dom}(\Gamma)$,  $D;F; \cdot; \cdot
\vdash \delta(x) : \Gamma(x)$.
\\By Substitution Lemma (Lemma~\ref{lem:polc-substitution})
\\\>(2)\> $D;F;\cdot;\bot\vdash e\circ\delta: \tpint\ \pol$ 
\\By $\Gamma\vdash \delta_1\bowtie\delta_2 =\delta$
\\\>(3)\> $\proj{e\circ\delta}{i} = e\delta_i$, $i\in\{1, 2\}$
\\By Completeness (Theorem~\ref{thm:polc-pair-op-completeness})
\\\>(4)\> $\Psi \vdash \emptyset\sepidx{} e\circ\delta \stepsto^*
\sigma\sepidx{} v$ and for all $i\in\{1, 2\}$, $\proj{v}{i}=v_i$
\\By Preservation (Theorem~\ref{thm:polc-preservation})
\\\>(5)\>  $D;F;\cdot;\bot\vdash v: \tpint\ \pol$ 
\\By Lemma~\ref{lem:polc-low-value-same}
\\\>(6)\> $v_1=v_2=\proj{v}{1}=\proj{v}{2}$
\end{tabbing}
\end{proof}

\section{Definitions and Proofs of Translations from annotated \minic to \minic via \langname} 
\label{app:translation}

\subsection{Mapping Annotated \minic to \langname}
We first list all the rules for mapping annotated \minic types to
\langname types. 
~\\
\framebox{$\inscast{a}=t$} \framebox{$\inscast{\beta}=s$}
\begin{mathpar}
\inferrule*{ 
 }{
  \inscast{\tunit}  = \tunit
}
\and
\inferrule*{ 
 }{
  \inscast{\tpint}  = \tpint\ \un
}
\and
\inferrule*{ 
  \inscast{\beta} =s
 }{
  \inscast{\tptr(\beta)}  = \tptr(s)\ \un
}
\and
\inferrule*{ 
 }{
  \inscast{T}  = T\ \un
}
\and
\inferrule*{ 
 }{
  \inscast{T\at\pol}  = T\ \pol
}
\and
\inferrule*{ 
 }{
  \inscast{\tpint\at \pol}  = \tpint\ \pol
}
\and
\inferrule*{ 
  \inscast{\beta} =s
 }{
  \inscast{\tptr(\beta)\at \pol}  = \tptr(s)\ \pol
}
\and
\inferrule*{
 \forall i\in[1,2],    \inscast{a_i}  = t_i\\
}{
\inscast{a_1 \rightarrow a_2} = [\bot](t_1\rightarrow t_2)
}
\and
\inferrule*{
 \forall i\in[1,2],    \inscast{a_i}  = t_i\\
}{
\inscast{(\dne)a_1 \rightarrow a_2} = (\dne)[\bot](t_1\rightarrow t_2)
}
\end{mathpar}

\framebox{$\inscast{D_a}=D$}
\begin{mathpar}
\inferrule*{
}{ \inscast{\cdot} = \cdot
 }
\and
\inferrule*{ 
 \forall i\in[1,k],    \inscast{a_i}  = t_i
 }{
  \inscast{D_a, T\mapsto \tstruct\ T\ \{a_1,\cdots,a_k\}}  = 
 \inscast{D_a}, T\mapsto \tstruct\ T\ \{t_1 \cdots,t_k\} 
}
\end{mathpar}

We write $\lv$ and $\lexp$ to denote labeled \langname values and
expressions respectively.
Values and expressions are mapped to labeled values and expressions to
facilitate the translation process later. 
\[
\begin{array}{lcll}
\textit{Labeled values} & \lv & \bnfdef & 
x@ s \bnfalt x @ (\dne) s \bnfalt n@ (b\ \pol) \bnfalt ()
\bnfalt   (T)\{\lv_1,\cdots, \lv_k\} @ (T\ \pol) 
\\ & & \bnfalt & 
f @ s \bnfalt f @ (\dne)\ s
\\
\textit{Labeled expressions} & \lexp & \bnfdef &
 \lv \bnfalt (\lexp_1 \bop \lexp_2)\bnfalt  (\lv\, \lexp)
\bnfalt (\elet\, x:s=\lexp_1\, \ein\, \lexp_2)
\bnfalt (\lv.i)   
\\
& & \bnfalt & \eif\, \lv\, \ethen\, \lexp_1\, \eelse \, \lexp_2
\bnfalt  \enew(\lexp)@ (\tptr(s)\ \pol) \bnfalt \lv\, :=\, \lexp \bnfalt *\lv
\bnfalt \erelab(\pol'\Leftarrow\pol) \lv
\end{array}
\]

Rules for mapping annotated \minic values to labeled \langname values
are as follows. 
~\\\noindent{\framebox{$D_a;F_a;\Gamma_a \vdash \inscast{v}\Rightarrow \lv$}}
\begin{mathpar}
\inferrule*[right=V-L-Unit]{
}{ 
  D_a;F_a;\Gamma_a\vdash \inscast{()} \Rightarrow\ ()
}
\and
\inferrule*[right=V-L-Int]{
}{ 
  D_a;F_a;\Gamma_a\vdash \inscast{n} \Rightarrow\ n @\tpint\ \un
}
\and
\inferrule*[right=V-L-Var]{
\Gamma_a(x) = \beta
}{ 
 D_a;F_a;\Gamma_a  \vdash \inscast{x} \Rightarrow\ x@\inscast{\beta}
}
\and
\inferrule*[right=V-L-Fun]{
}{ 
 D_a;F_a;\Gamma_a  \vdash \inscast{f} \Rightarrow\ f@\inscast{F(f)}
}
\and
\inferrule*[right=V-L-Struct]{
\forall i \in[1,n], D_a;F_a;\Gamma_a  \vdash \inscast{v_i} \Rightarrow \lv_i
}{ 
 D_a;F_a;\Gamma_a  \vdash \inscast{(T)\ \{v_1,\cdots,v_n\}}
 \Rightarrow\ (T)\{\lv_1,\cdots, \lv_n \}@T\ \un
}
\end{mathpar}

Next, we summarize rules for mapping annotated \minic expressions to labeled \langname expressions
below.
~\\ 
\noindent\framebox{$D_a;F_a;\Gamma_a; s\vdash \inscast{e} \Rightarrow \lexp$}

\begin{mathpar}
\inferrule*[right=L-Int]{
}{ 
D_a;F_a;\Gamma_a; t \vdash \inscast{n} \Rightarrow  n @ t
}
\and
\inferrule*[right=L-Var]{
\inscast{\Gamma_a(x)} = s
}{ 
 D_a;F_a;\Gamma_a ; s \vdash \inscast{x} \Rightarrow x @ s
}
\and
\inferrule*[right=L-Struct]{
 D_a(T) = \tstruct\ T\ \{\beta_1,\cdots, \beta_n\}
\\ \forall i\in[1,n], D_a;F_a;\Gamma_a ; \inscast{\beta_i} \vdash \inscast{v_i} \Rightarrow  \lv_i
}{ 
  D_a;F_a;\Gamma_a  ;\ T\ \pol\vdash \inscast{(T)\{v_1,\cdots,v_n\} } 
 \Rightarrow (T)\{\lv_1,\cdots,\lv_n\} @ (T\ \pol)
}
\and
\inferrule*[right=L-Field-U]{
D_a;F_a;\Gamma_a \vdash \inscast{v}\Rightarrow \lv
\\ \tpof(\lv) = T\ \pol
\\ D_a(T) =  (\tstruct\ T \{\beta_1,\cdots, \beta_n\})
\\ \forall i\in[1,n], \pol = \labof(\inscast{\beta_i})
}{ 
 D_a;F_a;\Gamma_a ; t  \vdash   \inscast{v.i} \Rightarrow \lv.i 
}
\and
\inferrule*[right=L-Field]{
D_a;F_a;\Gamma_a \vdash \inscast{v}\Rightarrow \lv
\\ \tpof(\lv) = T\ \pol
\\ D_a(T) =  (\tstruct\ T \{\beta_1,\cdots, \beta_n\})
\\ \exists i\in[1,n], \pol \neq \labof(\inscast{\beta_i})
}{ 
 D_a;F_a;\Gamma_a  ; t\vdash   \inscast{v.i} \Rightarrow
 \elet\ y:T\ \bot = \erelab(\bot\Leftarrow \pol)\ \lv\
\ein\ (y@ T\ \bot).i 
}
\and
\inferrule*[right=L-New]{
 D_a;F_a;\Gamma_a ; s\vdash   \inscast{e} \Rightarrow \lexp
}{ 
  D_a;F_a;\Gamma_a   ; (\tptr(s)\
  \pol)\vdash\inscast{\enew(e)} \Rightarrow \enew(\lexp)@(\tptr(s)\ \pol)
}
\and
\inferrule*[right=L-Deref]{
 D_a;F_a;\Gamma_a  \vdash   \inscast{v}  \Rightarrow \lv
\\ \tpof(\lv) = b\ \pol
}{ 
    D_a;F_a;\Gamma_a ; t  \vdash \inscast{*v} \Rightarrow 
(\elet\ y:b\ \bot = \erelab(\bot\Leftarrow \pol)\ \lv\ 
\ein\ *(y@b\ \bot)
}
\end{mathpar}
\begin{mathpar}
\inferrule*[right=L-Assign]{
 D_a;F_a;\Gamma_a  \vdash   \inscast{v}  \Rightarrow \lv
\\ \tpof(\lv) =\tptr(s)\ \pol
\\ D_a;F_a;\Gamma_a ; s\vdash e \Rightarrow \lexp
}{ 
  D_a;F_a;\Gamma_a ; t\vdash   \inscast{v:= e} \Rightarrow 
\elet\ y: \tptr(s)\ \bot = \erelab(\bot\Leftarrow \rho)\ \lv\ 
\ein\  y@\tptr(s)\ \bot :=\lexp 
}
\and
\inferrule*[right=L-App]{
D_a;F_a;\Gamma_a \vdash \inscast{v} \Rightarrow \lv
\\ \tpof(\lv) = [\bot](t_1\rightarrow t_2)^\bot
  \\  D_a;F_a;\Gamma_a ; t_1\vdash \inscast{e} \Rightarrow \lexp
}{ 
D_a;F_a;\Gamma_a ; t_2\vdash \inscast{v\ e} \Rightarrow \lv\;\lexp
}
\and
\inferrule*[right=L-App-DE]{
D_a;F_a;\Gamma_a \vdash \inscast{v_f} \Rightarrow \lv_f
\\ \tpof(\lv_f) = (\dne)[\bot](t_1\rightarrow t_2)^{\bot}
  \\  D_a;F_a;\Gamma_a \vdash \inscast{v_a} \Rightarrow \lv_a
}{ 
D_a;F_a;\Gamma_a ; t_2 \vdash \inscast{v_f\ v_a} 
\Rightarrow \lv_f\ \lv_a
}
\and
\inferrule*[right=L-Let]{
  D_a;F_a;\Gamma_a ; \inscast{\beta_1}\vdash   \inscast{e_1} \Rightarrow \lexp_1
 \\ D_a;F_a;\Gamma_a , x:\beta_1 ; t_2 \vdash  \inscast{e_2} \Rightarrow \lexp_2
}{ 
  D_a;F_a;\Gamma_a ; t_2\vdash  \inscast{\elet\ x:\beta_1 = e_1\ \ein\
    e_2}
\Rightarrow \elet\ x:\inscast{\beta_1} = \lexp_1\ \ein\ \lexp_2
}
\and
\inferrule*[right=L-If]{
D_a;F_a;\Gamma_a \vdash \inscast{v_1}\Rightarrow \lv_1
\\ \tpof(\lv_1) = \tpint\ \pol 
 \\ D_a;F_a;\Gamma_a ; t \vdash  \inscast{e_2} \Rightarrow \lexp_2
 \\ D_a;F_a;\Gamma_a ; t \vdash   \inscast{e_3}  \Rightarrow \lexp_3
}{ 
D_a;F_a;\Gamma_a ; t\vdash    \inscast{\eif\ v_1\ \ethen\  e_2\ \eelse\
  e_3}
\\\\\Rightarrow \elet\ x: \tpint\ \bot = (\erelab(\bot\Leftarrow\pol)\ \lv_1)\
 \ein\ \eif\ x@\tpint\ \bot\ \ethen\  \lexp_2\ \eelse\ \lexp_3
}
\end{mathpar}

The mapping of a function definition is as follows. To make sure
that programmers do not have to drastically change their programs, the
mapping takes care of relabeling so the parameter can be used at its
original type inside the function body. Similarly, the function body
is relabeled from the original type to the annotated type.
\[
\inferrule*{
 \inscast{a_1}= b_1\ \pol_1 
\\ \inscast{a_2}= b_2\ \pol_2 
\\\\D_a;F_a;\Gamma_a; b_2\ \pol_2\vdash \inscast{e[y/x]} \Rightarrow \lexp
}{
D_a;F_a;\Gamma_a \vdash \inscast{f(x): a_1\rightarrow a_2 = e} =
\\  f(x) = \elet\ y: t_1\ \un = \erelab(\un \Leftarrow \pol_1)\ x\
\\\\\ein\ \elet\ z: t_2\ \un = \lexp\ \ein\ \erelab (\pol_2\Leftarrow
\un)\ z
}
\]



\subsection{Translation from \langname to \minic}
We have two type translation functions, one that
does not take a type definition context as input and the one that does. The reason
is that when translating the annotated type definition context, we
need to generate new type definitions that are unknown at the
time of translation, which are mapped to $?$ as a result.  
~\\
\framebox{$\trans{t}=(\tau, D_\Delta)$}
\framebox{$\trans{s}=(\tau, D_\Delta)$}
\begin{mathpar}
\inferrule*{ 
\pol\in\{\un,\bot\}
 }{
  \trans{T\ \pol}  = (T, \cdot)
}
\and
\inferrule*{ 
\pol\notin\{\un,\bot\}
\\ T' = \genname(T,\pol)
 }{
  \trans{T\ \pol}  = (T', T'\mapsto T?)
}
\and
\inferrule*{ 
\pol\in\{\un,\bot\}
 }{
  \trans{\tpint\ \pol}  = (\tpint, \cdot)
}
\and 
\inferrule*{ 
\pol\notin\{\un,\bot\}\\
T = \genname(\tpint,\pol)
 }{
  \trans{\tpint\ \pol}  = (T, T\mapsto \tstruct\ T\ \{\tpint\})
}
\and 
\inferrule*{ 
\pol\in\{\un,\bot\}\\
  \trans{t} = (\tau, D')
 }{
  \trans{\tptr(t)\ \pol}  = (\tptr(\tau), D')
}
\and
\inferrule*{ 
\pol\notin\{\un,\bot\}\\
T = \genname(\tptr(t),\pol) \\
  \trans{t} = (\tau, D')
 }{
  \trans{\tptr(t)\ \pol}  = 
  (\tptr(\tau), (D', T\mapsto \tstruct\ T\ \{\tptr(\tau)\}))
}
\and
\inferrule*{
 \forall i\in[1,2],    \trans{t_i}  = (\tau_i, \tpdefctx_i)\\
}{
\trans{[\bot](t_1 \rightarrow t_2)^\bot} = (\tau_1\rightarrow\tau_2,\tpdefctx_1\cup\tpdefctx_2)
}
\end{mathpar}

\framebox{$\trans{t}_D=(\tau, D_\Delta)$}
\framebox{$\trans{s}_D=(\tau, D_\Delta)$}
\begin{mathpar}
\inferrule*{ 
\pol\in\{\un,\bot\}
 }{
  \trans{T\ \pol}_D  = (T, \cdot)
}
\and
\inferrule*{ 
\pol\notin\{\un,\bot\}
\\ T' = \genname(T,\pol)
\\ T\mapsto \tstruct\ T\ \{\tau_1,\cdots,\tau_n\} \in D
 }{
  \trans{T\ \pol}_D  = (T', T'\mapsto \tstruct\ T'\ \{\tau_1,\cdots,\tau_n\})
}
\and
\inferrule*{ 
\pol\in\{\un,\bot\}
 }{
  \trans{\tpint\ \pol}_D  = (\tpint, \cdot)
}
\and 
\inferrule*{ 
\pol\notin\{\un,\bot\}\\
T = \genname(\tpint,\pol)
 }{
  \trans{\tpint\ \pol}_D  = (T, T\mapsto \tstruct\ T\  \{\tpint\})
}
\and 
\inferrule*{ 
\pol\in\{\un,\bot\}\\
  \trans{t}_D = (\tau, D')
 }{
  \trans{\tptr(t)\ \pol}_D  = (\tptr(\tau), D')
}
\and
\inferrule*{ 
\pol\notin\{\un,\bot\}\\
T = \genname(\tptr(t),\pol) \\
  \trans{t}_D = (\tau, D')
 }{
  \trans{\tptr(t)\ \pol}_D  = 
  (\tptr(\tau), (D', T\mapsto \tstruct\ T\ \{\tptr(\tau)\}))
}
\and
\inferrule*{
 \forall i\in[1,2],    \trans{t_i}  = (\tau_i, \tpdefctx_i)_D\\
}{
\trans{[\bot](t_1 \rightarrow t_2)^\bot}_D = (\tau_1\rightarrow\tau_2,\tpdefctx_1\cup\tpdefctx_2)
}
\end{mathpar}

Translating the annotated type definition context needs two steps. The
first step generates new type definitions, which are not filled as
they themselves are being translated. In the second step, we fill
these undefined type definitions using the translated type definition context.
~\\\framebox{$[D] = (D'; D_\Delta)$}
\begin{mathpar}
\inferrule*{ 
 }{
 [\cdot] = (\cdot; \cdot)
}
\and
\inferrule*{ 
 \forall i\in[1,k],    \trans{t_i}  = (\tau_i, \tpdefctx_i)
\\ [D] = (D'; D_\Delta)
 }{
  [D, T\mapsto \tstruct\ T\ \{t_1,\cdots,t_k\}]  = 
 (D', T\mapsto \tstruct\ T\ \{\tau_1\cdots,\tau_k\}; \cup^k_{i=1} \tpdefctx_i\cup\Delta_\Delta)
}
\end{mathpar}

\framebox{$D\vdash \m{fill}(D_1) = D_2$}
\begin{mathpar}
\inferrule*{ 
 }{
D\vdash  \m{fill}(\cdot) = \cdot
}
\and
\inferrule*{ 
 }{
D\vdash  \m{fill}(D_1, T\mapsto \tstruct\ T\ \{\pi_1,\cdots,\pi_n\}) = 
 D\vdash \m{fill}(D_1), T\mapsto \tstruct\ T\  \{\pi_1,\cdots,\pi_n\}
}
\and
\inferrule*{ 
 D(T) = \tstruct\ T\  \{\pi_1,\cdots,\pi_n\}
 }{
D\vdash  \m{fill}(D_1, T'\mapsto \tstruct\ T?) = 
 D\vdash \m{fill}(D_1), T'\mapsto \tstruct\ T'\  \{\pi_1,\cdots,\pi_n\}
}
\end{mathpar}

\[
\inferrule*{
  [D] = (D'; D_\Delta)
\\ D' \vdash \m{fill}(D_\Delta) = D''
 }{
\trans{D} = D', D''}
\]

\begin{figure*}[t!]
\flushleft
\noindent\framebox{$ \trans{\lexp}  = e$}

\begin{mathpar}
\inferrule*[right=T-Var]{
}{ 
  \trans{x@ s}_D  = (x, \cdot)
}
\and
\inferrule*[right=T-Int]{
 \pol \in\{ \un, \bot\}
}{ 
 \trans{n@\tpint\ \pol}_D =  (n, \cdot)
}
\and
\inferrule*[right=T-Int-Pol]{
 \pol \notin\{ \un, \bot\}
\\\trans{\tpint\ \pol}_D = (T, D')
}{ 
 \trans{n@\tpint\ \pol}_D  = ((T)\{ n \}, D')
}
\and
\inferrule*[right=T-Struct]{
\forall i\in[1,n], \trans{\lv_i}_D = ( v_i, D_i)
\\ \trans{T\ \pol}_D = (T', D')
}{ 
  \trans{(T) \{\lv_1,\cdots,\lv_n\} @(T\ \pol)}_D  =
  ((T')\{v_1,\cdots,v_n\}, \cup_1^n D_i\cup D'
}
\and
\inferrule*[right=T-Field]{
\trans{\lv}_D  = (v, D')
}{ 
   \trans{\lv.i}_D  = (v.i, D')
}
\and
\inferrule*[right=T-New]{
\trans{\lexp}_D = (e, D')
\\\pol\in\{\un,\bot\}
}{ 
 \trans{\enew(\lexp)@\tptr(s)\ \pol}_D  = (\enew(e), D')
}
\and
\inferrule*[right=T-New-Pol]{
\trans{\lexp}_D = (e, D')
\\\pol\notin\{\un,\bot\}
\\\trans{\tptr(s)\ \pol}_D = T
}{ 
 \trans{\enew(\lexp)@\tptr(t)\ \pol}_D  = ((T)\{\enew(e)\}, D')
}
\and
\inferrule*[right=T-Deref]{
\trans{\lv}_D  = (v, D')
}{ 
   \trans{*\lv}_D = (*v, D')
}
\and
\inferrule*[right=T-Assign]{
\trans{\lv_1}_D  = (v_1, D_1)
\\\trans{\lexp_2}_D  = (e_2, D_2)
}{ 
  \trans{\lv_1:= \lexp_2}_D  = (v_1:=e_2, D_1\cup D_2)
}
\and
\inferrule*[right=T-Let]{
   \trans{\lexp_1}_D= (e_1, D_1)
 \\ \trans{\lexp_2}_D  = (e_2, D_2)
\\ \trans{t_1}_D = (\tau_1, D_3)
}{ 
   \trans{\elet\ x:t_1 = \lexp_1\ \ein\  \lexp_2}_D
= (\elet\ x:\tau_1 = e_1\ \ein\  e_2, D_1\cup D_2\cup D_3
}
\and
\inferrule*[right=T-If]{
   \trans{\lv_1}_D= (v_1, D_1)
 \\ \trans{\lexp_2}_D  = (e_2, D_2)
 \\ \trans{\lexp_3}_D  = (e_3, D_3)
}{ 
   \trans{\eif\ \lv_1\ \ethen\  \lexp_2\ \eelse\  \lexp_3}_D
= (\eif\ v_1\ \ethen\  e_2\ \eelse\  e_3, D_1\cup D_2\cup D_3)
}
\and
\inferrule*[right=T-App]{
\tpof(\lv) = [\pc](t_1\rightarrow t_2)^{\pol_f}
\\\trans{\lv}_D = (v, D_1)
\\\trans{\le}_D = (e, D_2)
}{ 
\trans{\lv\ \lexp}_D= (v\ e, D_1\cup D_2)
}
\and
\inferrule*[right=T-App-DE]{
\tpof(\lv_f) = (\dne)[\pc](t_1\rightarrow t_2)^{\pol_f}
\\\trans{\lv_f}_D = (v_f, D_f)
\\ \tpof(\lv_a) = b\ \pol
\\ \pol = \lab_1::\lab_2::\pol'
\\    \trans{\erelab(\lab_1::\top\Leftarrow \pol) \lv_a}_D = (e', D_1)
\\  \trans{\erelab(\lab_2::\pol'\Leftarrow \lab_2::\bot) (z@b\
  \lab_2::\bot)}_D= (e'' , D_2)
\\ \trans{t_1}_D = (\tau_1, D_3)
\\ \trans{t_2}_D = (\tau_2, D_4)
}{ 
\trans{\lv_f\ \lv_a}_D = 
 (\elet\ y:\tau_1 = e'\ \ein\ 
 \elet\ z: \tau_2 = v_f\ y\ \ein\ e'', D_f\cup D_1\cup D_2\cup D_3\cup D_4)
}
\end{mathpar}
\caption{Translation Rules}
\label{fig:translation}
\end{figure*}

\begin{figure*}[t!]
\begin{mathpar}
\inferrule*[right=T-ReLab-N1]{
\trans{\lv}_D  = (v, D_1)
\\ \tpof(\lv) = b\ \_
\\ b~\mbox{is not a struct type}
\\ \pol'\notin\{\bot,\un\} \\\pol \notin\{\bot,\un\}
\\ \trans{b\ \pol'}_D = (T, D_2)
}{ 
   \trans{\erelab(\pol'\Leftarrow\pol)\lv}_D =(\elet\ x = v.1\ \ein\
   (T)\{x\}, D_1\cup D_2)
}
\and
\inferrule*[right=T-ReLab-N2]{
\trans{\lv}_D  = (v, D_1)
\\ \tpof(\lv) = b\ \_
\\ b~\mbox{is not a struct type}
\\ \pol'\notin\{\bot,\un\} 
\\\pol \in\{\bot,\un\}
\\ \trans{b\ \pol'}_D = (T, D_2)
}{ 
   \trans{\erelab(\pol'\Leftarrow\pol)\lv}_D = ((T)\{v\}, D_1\cup D_2)
}
\end{mathpar}

\begin{mathpar}

\inferrule*[right=T-ReLab-N3]{
\trans{\lv}_D= (v, D_1)
\\ \tpof(\lv) = b\ \pol
\\ b~\mbox{is not a struct type}
\\ \pol\notin\{\bot,\un\} \\\pol'\in\{\bot,\un\}
}{ 
   \trans{\erelab(\pol'\Leftarrow\pol)\lv}_D = (v.1 D_1)
}
\and
\inferrule*[right=T-ReLab-same]{
\trans{\lv}_D  = (v, D_1)
\\ \labof(\lv) = b\ \pol
\\ \pol,\pol' \in\{\un,\bot\} 
}{ 
   \trans{\erelab(\pol'\Leftarrow\pol)\lv}_D = (v, D_1)
}
\and
\inferrule*[right=T-ReLab-Struct-Re]{
 \pol \notin\{\bot, \un\} ~\mbox{or}~\pol'\notin\{\bot,\un\}
\\ \tpof(\lv) =  T\ \pol 
\\ \trans{T\ \pol'}_D = (T', D_1)
\\    \trans{\lv}_D  = (v, D_2)
}{ 
\trans{\erelab(\pol'\Leftarrow\pol)\lv}_D =
(\elet\ x_1 = v.1\ \ein\ \cdots \elet\ x_n = v.n\ \ein\
(T')\{x_1,\cdots,x_n\}, D_1\cup D_2)
}
\end{mathpar}
\caption{Translation Rules}
\label{fig:translation-2}
\end{figure*}


\subsection{Correctness of the Translation}
\label{app:translation:correctness}
We present definitions, lemmas, and proofs for the correctness of our
translation algorithm. 

\begin{mathpar}
\inferrule*{ 
 }{
  \m{noBot}(\tunit)
}
\and
\inferrule*{ 
\pol\neq \bot
 }{
  \m{noBot}(\tpint\ \pol)
}
\and
\inferrule*{ 
  \m{noBot}(s)
\\\pol\neq \bot
 }{
  \m{noBot}(\tptr(s)\ \pol)
}
\and
\inferrule*{ 
 \pol\neq \bot
 }{
  \m{noBot}(T\at\pol)
}
\\
\inferrule*{
 \forall i\in[1,2],    \m{noBot}(t_i)
}{
\m{noBot}([\bot](t_1\rightarrow t_2)^\bot)
}
\and
\inferrule*{
 \forall i\in[1,2],    \m{noBot}(t_i) \\
}{
\m{noBot}((\dne)[\bot](t_1\rightarrow t_2)^\bot)
}
\end{mathpar}
\begin{lem}[Translation Pre-image Unique]\label{lem:trans-unique}
If  $\m{noBot}(s)$ and $\m{noBot}(s')$ and
$\m{fst}(\trans{s})=\m{fst}(\trans{s'})$ 
or $\m{fst}(\trans{s}_D)=\m{fst}(\trans{s'}_D)$ then $s = s'$.
\end{lem}
\begin{proofsketch}
By induction over the structure of $s$.
\end{proofsketch}

\begin{lem}\label{lem:nobot}
If $\beta$ does not include $\bot$, then $\m{noBot}(\inscast{\beta})$.
\end{lem}
\begin{proofsketch}
By induction over the structure of $\beta$. The translation rules does
not insert $\bot$ except for functions.
\end{proofsketch}

\begin{lem}[Value Translation Soundness]
\label{lem:val-translation-soundness}
If  $\ee::D_a;F_a;\Gamma_a\vdash\inscast{v}= \lv$,
$\tpof(\lv)=s$,
$\inscast{D_a} = D_l$, $\inscast{F_a} = F_l$, 
$\inscast{\Gamma_a} = \Gamma_l$,
$\trans{D_l} = D$,
$\trans{\Gamma_l}_D = (\Gamma, D_1)$, $\trans{F_l}_D = (F, D_2)$,
 $\trans{\lv}_D = (v', D_3)$,
and $D\cup D_1\cup D_2\cup D_3;F;\cdot;\Gamma\vdash v': \tau$
implies $D_l;F_l;\cdot;\Gamma_l \vdash \tmof(\lv): s$ and $\trans{s} = (\tau,\_)$.
\end{lem}
\begin{proof}
By induction over the structure of $\ee$.
\begin{description}
\item[Case:] $\ee$ ends in \rulename{V-L-Int} rule. 
\begin{tabbing}
By assumption: 
\\~~\=(1)~~~~\= $D_a;F_a;\Gamma_a\vdash \inscast{n} \Rightarrow\ n
@\tpint\ \un$
\\ By examining the translation rules, only \rulename{T-Int} applies
\\\> (2) \> $\trans{n @\tpint\ \un}_D = (n, \cdot)$,
\\ By typing rules 
\\\>(3)\> $D\cup D_1\cup D_2;F;\cdot;\Gamma\vdash n: \tpint$
\\ By typing rule \rulename{V-Int} 
\\\>(4)\> $D_l;F_l;\cdot;\Gamma_l \vdash n: \tpint\ \un$ 
\\ By type translation
\\\>(5)\> $\trans{\tpint\ \un}_D = (\tpint,\_)$
\end{tabbing}
\item[Case:] $\ee$ ends in \rulename{V-L-Var} rule. 
\begin{tabbing}
By assumption: 
\\~~\=(1)~~~~\= $D_a;F_a;\Gamma_a\vdash \inscast{x} \Rightarrow\ x
@ \inscast{\beta} $ and $\Gamma_a(x)=\beta$
\\ By examining the translation rules, only \rulename{T-Var} applies
\\\> (2) \> $\trans{x @\inscast{\beta} }_D = (x, \cdot)$,
\\ By typing rules 
\\\>(3)\> $D\cup D_1\cup D_2;F;\cdot;\Gamma\vdash x: \Gamma(x)$
\\ By typing rule \rulename{V-Var} 
\\\>(4)\> $D_l;F_l;\cdot;\Gamma_l \vdash x: \Gamma_l(x)$ 
\\ By assumption that $\inscast{\Gamma_a} = \Gamma_l$ and
   $\trans{\Gamma_l}_D = (\Gamma, D_1)$
\\\>(5)\> $\trans{\inscast{\beta}}_D=\trans{\inscast{\Gamma_a(x)}}_D
=\trans{\Gamma_l(x)}_D = (\Gamma(x),\_)$
\end{tabbing}
\item[Case:] $\ee$ ends in \rulename{V-L-Fun} rule. 
~\\This case can be proved similarly as the previous case.
\item[Case:] $\ee$ ends in \rulename{V-L-Struct} rule. 
\begin{tabbing}
By assumption: 
\\~~\=(1)~~~~\= $D_a;F_a;\Gamma_a\vdash \inscast{(T)\ \{v_1,\cdots,v_n\}}
 \Rightarrow\ (T)\{\lv_1,\cdots, \lv_n \}@T\ \un$
\\\>(2)\> and $\forall i \in[1,n], \ee_i::D_a;F_a;\Gamma_a  \vdash \inscast{v_i} \Rightarrow \lv_i$
\\ By examining the translation rules, only \rulename{T-Struct} applies
\\\> (3) \> $ \trans{(T) \{\lv_1,\cdots,\lv_n\} @(T\ \un)}_D  =
((T')\{v'_1,\cdots,v'_n\}, D_3)$
\\\>(4) \> and $\trans{\lv_i}_D =  (v'_i, D_{3i})$ and $\trans{T\ \un}
= (T', D_{3ii})$, and $D_3=D_{3i}\cup D_{3ii}$
\\By type translation rules
\\\>(5)\> $T'=T$ and $D_{3ii}=\cdot$
\\ By typing rules 
\\\>(6)\> $D\cup D_1\cup D_2\cup D_3;F;\cdot;\Gamma\vdash
(T)\{v'_1,\cdots,v'_n\}: T$
\\\>(7)\> and $(D\cup D_1\cup D_2\cup D_3)(T) =\tstruct\
T\{\tau_1,\cdots,\tau_n\} $,
\\\>(8)\> and $D\cup D_1\cup D_2\cup D_3;F;\cdot;\Gamma\vdash v'_i: \tau_i$
\\By $T\in \m{dom}(D_a)$ 
\\\>(9) \> $T\in \m{dom}(D)$
\\By I.H. on $\ee_i$
\\\>(10)\> $D_l;F_l;\cdot;\Gamma_l \vdash \tmof(\lv_i): s_i$ and
$\trans{s_i}_D = (\tau_i,\_)$
\\By well-formedness constraints,
\\\>(11)\> $D_l(T) = \tstruct\ T \{s'_1,\cdots,s'_n\}$
\\By $\trans{D_l} = D$ and $D(T) = \tstruct\ T\{\tau_1,\cdots,\tau_n\}
$
\\\>(12)\> $\trans{s'_i} = \tau_i$
\\By Lemma~\ref{lem:trans-unique}, (10) and (12)
\\\>(13)\> $s_i=s'_i$
\\By \rulename{V-Struct} and (10) and (13)
\\\>(14)\> $D_l;F_l;\cdot;\Gamma_l \vdash \tmof((T)
\{\lv_1,\cdots,\lv_n\} @(T\ \un)): T\ \un$
\\By type translation rules
\\\>(15)\> $\trans{T\ \un}_D = (T,\_)$
\end{tabbing} 

\end{description}
\end{proof}

\begin{lem}[Relabel translation is sound]\label{lem:relab-sound}
If $\tpof(\lv) = b\ \pol$,
$e=\m{fst}(\trans{\erelab(\pol'\Leftarrow\pol)\lv}_D)$,
$v=\m{fst}(\trans{\lv}_D)$, and 
$D',D;F;\Gamma\vdash v: \m{fst}(\trans{b\ \pol}_D)$
then $D;F;\Gamma\vdash e: \m{fst}(\trans{b\ \pol'}_D)$.
\end{lem}
\begin{proofsketch}
By examining  the translation rules for $\trans{\erelab(\pol'\Leftarrow\pol)\lv}_D$.
\end{proofsketch}

\begin{lem}\label{lem:relab-has-tp}
Given $e=\m{fst}(\trans{\erelab(\pol'\Leftarrow\pol)\lv}_D)$,
then
$D,D';F;\Gamma\vdash e: \tau$
implies exists $\tau'$ s.t.  $D,D';F;\Gamma\vdash \m{fst}(\trans{\lv}_D) : \tau'$.
\end{lem}
\begin{proofsketch}
By examining the translation rules for $\trans{\erelab(\pol'\Leftarrow\pol)\lv}_D$.
\end{proofsketch}

\begin{thm}[Expression Translation Soundness]
If  $\ee::D_a;F_a;\Gamma_a; s\vdash\inscast{e}= \lexp$,
$\inscast{D_a} = D_l$, $\inscast{F_a} = F_l$, 
$\inscast{\Gamma_a} = \Gamma_l$,
$\trans{D_l} = D$,
$\trans{\Gamma_l}_D = (\Gamma, D_1)$, $\trans{F_l}_D = (F, D_2)$,
 $\trans{\lexp}_D = (e', D_3)$,
and $D\cup D_1\cup D_2\cup D_3;F;\cdot;\Gamma\vdash e': \tau$
implies $D_l;F_l;\cdot;\Gamma_l \vdash \tmof(\lexp): s$ and $\trans{s} = (\tau,\_)$
\end{thm}

\begin{proof}
By induction over the structure of $\ee$.
\begin{description}
\item[Case:] $\ee$ ends in \rulename{L-Int} rule. 
\begin{tabbing}
By assumption: 
\\~~\=(1)~~~~\= $D_a;F_a;\Gamma_a; s\vdash \inscast{n} = n @ s$
\\ By examining the translation rules, there are two subcases
\\\> {\bf subcase i.} \rulename{T-Int} applies
\\\>~~\=(i2)~~~~\=$s=\tpint\ \un$ and $\trans{n @\tpint\ \un}_D = (n, \cdot)$,
\\\> By typing rules 
\\\>\>(i3)\> $D\cup D_1\cup D_2;F;\cdot;\Gamma\vdash n: \tpint$
\\\> By typing rule \rulename{P-T-V-Int} and \rulename{P-T-E-Val}
\\\>\>(i4)\> $D_l;F_l;\cdot;\Gamma_l;\bot \vdash n: \tpint\ \un$ 
\\\> By type translation
\\\>\>(i5)\> $\trans{\tpint\ \un}_D = (\tpint,\_)$
\\\> {\bf subcase ii.} \rulename{T-Int-Pol} applies
\\\>\> (ii2)\> $s=\tpint\ \pol$ and $\trans{n @\tpint\ \pol}_D =
((T)\{n\},D_e)$, $\trans{\tpint\ \pol}_D = (T, D_3)$
\\\> By (ii2) and type translation rules 
\\\>\>(ii3)\> $D_3 = T\mapsto \tstruct\ T\{\tpint\}$
\\\> By typing rules 
\\\>\>(ii4)\> $D\cup D_1\cup D_2\cup D_3;F;\cdot;\Gamma\vdash (T)\{n\}: T$
\\\> By typing rule \rulename{P-T-V-Int} and \rulename{P-T-E-Val}
\\\>\>(ii5)\> $D_l;F_l;\cdot;\Gamma_l;\bot \vdash n: \tpint\ \pol$ 
\\\> By type translation
\\\>\>(ii6)\> $\trans{\tpint\ \pol}_D = (T,\_)$
\end{tabbing}

\item[Case:] $\ee$ ends in \rulename{L-Var} rule. 
\begin{tabbing}
By assumption: 
\\~~\=(1)~~~~\= $D_a;F_a;\Gamma_a; s\vdash \inscast{x} \Rightarrow\ x
@s $ and $\inscast{\Gamma_a(x)}=s$
\\ By examining the translation rules, only \rulename{T-Var} applies
\\\> (2) \> $\trans{x @s }_D = (x, \cdot)$,
\\ By typing rules 
\\\>(3)\> $D\cup D_1\cup D_2;F;\cdot;\Gamma\vdash x: \Gamma(x)$
\\ By typing rule \rulename{P-T-V-Var} and \rulename{P-T-E-Val}
\\\>(4)\> $D_l;F_l;\cdot;\Gamma_l;\bot \vdash x: \Gamma_l(x)$ 
\\ By assumption that $\inscast{\Gamma_a} = \Gamma_l$ and
   $\trans{\Gamma_l}_D = (\Gamma, D_1)$
\\\>(5)\> $\trans{s}_D=\trans{\inscast{\Gamma_a(x)}}_D
=\trans{\Gamma_l(x)}_D = (\Gamma(x),\_)$
\end{tabbing}

\item[Case:] $\ee$ ends in \rulename{L-Struct} rule. 
\begin{tabbing}
By assumption: 
\\~~\=(1)~~~~\= $D_a;F_a;\Gamma_a; T\ \pol \vdash \inscast{(T)\ \{v_1,\cdots,v_n\}}
 \Rightarrow\ (T)\{\lv_1,\cdots, \lv_n \}@T\ \pol$
\\\>(2) \> $D_a(T) =\tstruct\ T\{\beta_1,\cdots,\beta_n\}$
\\\>(3)\> and $\forall i \in[1,n], \ee_i::D_a;F_a;\Gamma_a;\inscast{\beta_i} \vdash \inscast{v_i} \Rightarrow \lv_i$
\\ By examining the translation rules, only \rulename{T-Struct} applies
\\\> (4) \> $ \trans{(T) \{\lv_1,\cdots,\lv_n\} @(T\ \pol)}_D  =
((T')\{v'_1,\cdots,v'_n\}, D_3)$
\\\>(5) \> and $\trans{\lv_i}_D =  (v'_i, D_{3i})$ and $\trans{T\
  \pol}_D = (T', D')$
\\By type translation rules
\\\>(6)\> $T'=\genname(T,\pol)$ and $D' = T'\mapsto \tstruct\ T' \{
\tau_1,\cdots, \tau_n\}$ and $D(T) = \tstruct\ T \{\tau_1,\cdots, \tau_n\}$
\\ By typing rules 
\\\>(7)\> $D\cup D_1\cup D_2\cup D_3;F;\cdot;\Gamma\vdash
(T')\{v'_1,\cdots,v'_n\}: T'$
\\\>(8)\> and $(D\cup D_1\cup D_2\cup D_3)(T') =\tstruct\
T'\{\tau_1,\cdots,\tau_n\} $,
\\\>(9)\> and $D\cup D_1\cup D_2\cup D_3;F;\cdot;\Gamma\vdash v'_i: \tau_i$
\\By I.H. on $\ee_i$
\\\>(10)\> $D_l;F_l;\cdot;\Gamma_l;\bot \vdash \tmof(\lv_i): \inscast{\beta_i}$ and
$\trans{\inscast{\beta_i}}_D = (\tau_i,\_)$
\\By \rulename{P-T-V-Struct} and \rulename{P-T-E-Val} and (6) and (10)
\\\>(11)\> $D_l;F_l;\cdot;\Gamma_l;\bot \vdash \tmof((T)
\{\lv_1,\cdots,\lv_n\} @(T\ \pol)): T\ \pol$
\\By (5)
\\\>(12)\> $\trans{T\ \pol}_D = (T',\_)$
\end{tabbing} 

\item[Case:] $\ee$ ends in \rulename{L-Field-U} rule. 
\begin{tabbing}
By assumption: 
\\~~\=(1)~~~~\= $D_a;F_a;\Gamma_a; s \vdash \inscast{v.i}
 \Rightarrow\ \lv.i$
\\\>(2) \> $\ee'::D_a;F_a;\Gamma_a \vdash \inscast{v}\Rightarrow \lv$
and $\tpof(\lv) = T\ \pol$
\\\>(3)\> and $D_a(T) =  (\tstruct\ T \{\beta_1,\cdots,
\beta_n\})$
\\\>(4)\> and $\forall i\in[1,n], \pol = \labof(\inscast{\beta_i})$
\\By examining the translation rules, only \rulename{T-Field} applies 
\\\>(5)\> $\trans{\lv.i}_D =  (v'.i, D_3)$  and $\trans{\lv}_D =  (v', D_3)$
\\By assumption  and typing rules
\\\>(6)\> $D\cup D_1\cup D_2\cup D_3;F;\cdot;\Gamma\vdash v'.i : \tau_i$
\\By inversion of (6)
\\\>(7)\> $D\cup D_1\cup D_2\cup D_3;F;\cdot;\Gamma\vdash v' : T'$ 
\\\>(8)\> and   $(D\cup D_1\cup D_2\cup D_3)(T') = \tstruct\ T'\{\tau_1,\cdots,\tau_n\}$
\\By Lemma~\ref{lem:val-translation-soundness} on $\ee'$, (5) and (8)
\\\>(9)\> $D_l;F_l;\cdot;\Gamma_l \vdash \tmof(\lv): T\ \pol$ and $\trans{T\ \pol}_D = (T', \_)$
\\By (9), \rulename{P-T-E-Val}
\\\>(10)\> $D_l;F_l;\cdot;\Gamma_l;\bot \vdash \tmof(\lv): T\ \pol$ 
\\By (3) and $\inscast{D_a}=D_l$
\\\>(11)\> $D_l(T) = \tstruct\ T\ \{\inscast{\beta_1},\cdots,\inscast{\beta_n}\}$
\\By (10), (11), and \rulename{P-T-E-Field}
\\\>(12)\>$D_l;F_l;\cdot;\Gamma_l;\bot \vdash \tmof(\lv.i):
\inscast{\beta_i}\join\pol$ 
\\By (4) and (12)
\\\> (13) \>$D_l;F_l;\cdot;\Gamma_l;\bot \vdash \tmof(\lv.i):
\inscast{\beta_i}$ 
\\By (9) and $\trans{D_l} = D$ 
\\\>(14)\> $D(T) = \tstruct\ T\{\tau_1,\cdots,\tau_n\}$ and $D(T') =
\tstruct\ T'\{\tau_1,\cdots,\tau_n\}$
\\By (11) and (14)
\\\>(15)\> $\trans{\inscast{\beta_i}}_D = (\tau_i,\_)$
\end{tabbing}

\item[Case:] $\ee$ ends in \rulename{L-Field} rule. 
\begin{tabbing}
By assumption: 
\\~~\=(1)~~~~\= $D_a;F_a;\Gamma_a; s \vdash \inscast{v.i}
 \Rightarrow\ \lexp$ and $\lexp=\elet\ y:T\ \bot = \erelab(\bot\Leftarrow \pol)\ \lv\
\ein\ (y@ T\ \bot).i $
\\\>(2) \> $\ee'::D_a;F_a;\Gamma_a \vdash \inscast{v}\Rightarrow \lv$
and $\tpof(\lv) = T\ \pol$
\\\>(3)\> and $D_a(T) =  (\tstruct\ T \{\beta_1,\cdots,
\beta_n\})$
\\By examining the translation rules, only \rulename{T-Let} applies 
\\\>(4)\> $\trans{\lexp}_D =  (e', D'_3\cup D''_3)$  and 
     $e' = \elet\ y:T = e_1\ \ein\ y.i$, 
\\\>(5)\> and    $\trans{\erelab(\bot\Leftarrow \pol)\ \lv}_D = (e_1,
D_3)$
\\By assumption and typing rules 
\\\>(6)\>  $D\cup D_1\cup D_2\cup D_3;F;\cdot;\Gamma\vdash e':\tau_i$
\\By inversion of (6)
\\\>(7)\>  $D\cup D_1\cup D_2\cup D_3;F;\cdot;\Gamma\vdash e_1: T$
\\\>(8)\>  $D\cup D_1\cup D_2\cup D_3;F;\cdot;\Gamma,y:T \vdash y.i:
\tau_i$ and $$
\\\>(9)\> $D_l(T) = \tstruct\ T\ \{\inscast{\beta_1},\cdots,\inscast{\beta_n}\}$
\\By (9) \rulename{P-T-E-Field}
\\\>(10)\>$D_l;F_l;\cdot;\Gamma_l, y:T\ \bot;\bot \vdash y.i: \inscast{\beta_i}$ 
\\By (9) and $\trans{D_l} = D$ 
\\\>(11)\> $D(T) = \tstruct\ T\{\tau_1,\cdots,\tau_n\}$ 
\\By (9) and (11)
\\\>(12)\> $\trans{\inscast{\beta_i}}_D = (\tau_i,\_)$
\\By Lemma~\ref{lem:relab-has-tp} and (7)
\\\>(13)\>  $D\cup D_1\cup D_2\cup D_3;F;\cdot;\Gamma\vdash
\m{fst}(\trans{\lv}_D): \tau'$
\\By Lemma~\ref{lem:val-translation-soundness} on $\ee'$, (5), (7),
(13)
\\\>(14)\> $D_l;F_l;\cdot;\Gamma_l \vdash \tmof(\lv): T\ \pol$
\\By (14), \rulename{P-T-E-Relab}
\\\>(15)\> $D_l;F_l;\cdot;\Gamma_l;\bot \vdash
  \tmof(\erelab(\bot\Leftarrow \pol)\ \lv): T\ \bot$ 
\\\> By \rulename{P-T-E-Let}, (15), (10),
\\\>(16)\> $D_l;F_l;\cdot;\Gamma_l;\bot \vdash \lexp: \inscast{\beta_i}$ 
\end{tabbing}
\item[Case:] $\ee$ ends in \rulename{L-New} rule. 
\begin{tabbing}
By assumption: 
\\~~\=(1)~~~~\= $D_a;F_a;\Gamma_a   ; (\tptr(s)\
  \pol)\vdash\inscast{\enew(e)} \Rightarrow \enew(\lexp)@(\tptr(s)\ \pol)$
\\\>(2) \> $\ee'::  D_a;F_a;\Gamma_a ; s\vdash   \inscast{e} \Rightarrow \lexp $
\\By examining the translation rules, there are two subcases
\\\>{\bf Subcase i:}  \rulename{T-New} applies 
\\\>~~\= (i3)~~~~\=$\pol\in\{\un,\bot\}$, 
 $\trans{\enew(\lexp)@(\tptr(s)\ \pol)}_D =  (\enew(e'), D_3)$  and $\trans{\lexp}_D =  (e', D_3)$
\\\>By assumption  and typing rules
\\\>\>(i4)\> $D\cup D_1\cup D_2\cup D_3;F;\cdot;\Gamma\vdash \enew(e') : \tptr(\tau)$
\\\>By inversion of (i4)
\\\>\>(i5)\> $D\cup D_1\cup D_2\cup D_3;F;\cdot;\Gamma\vdash e' : \tau$ 
\\\>By I.H. on $\ee'$, (i5) and (i3)
\\\>\>(i6)\> $D_l;F_l;\cdot;\Gamma_l;\bot \vdash \tmof(\lexp): s$ 
   and $\trans{s}_D = (\tau, \_)$
\\\>By (i6), \rulename{P-T-E-New}
\\\>\>(i7)\> $D_l;F_l;\cdot;\Gamma_l;\bot \vdash \tmof(\enew(\lexp)):
\tptr(s)\ \pol$ 
\\\>\>(i8)\> $\trans{\tptr(s)\ \pol}_D = (\tptr(\tau),\_)$
\\\>{\bf Subcase ii:}  \rulename{T-New-Pol} applies 
\\\>~~\= (ii3)~~~~\=$\pol\notin\{\un,\bot\}$, 
 $\trans{\enew(\lexp)@(\tptr(s)\ \pol)}_D =  (T\{\enew(e')\}, D_3)$ 
\\\>\>(ii4)\>  and $\trans{\lexp}_D =  (e', D_3)$ and $\trans{\tptr(s)\ \pol} = T$
\\\>By assumption  and typing rules
\\\>\>(ii5)\> $D\cup D_1\cup D_2\cup D_3;F;\cdot;\Gamma\vdash (T)\{\enew(e')\} : T$
\\\>By inversion of (i5)
\\\>\>(ii6)\> $D\cup D_1\cup D_2\cup D_3;F;\cdot;\Gamma\vdash e' :
\tau$ 
\\\>\>(ii7)\> and $(D\cup D_1\cup D_2\cup D_3)(T) = \tstruct\
T\{\tptr(\tau)\}$ and $\trans{s}_D = (\tau,\_)$
\\\>By I.H. on $\ee'$, (ii6) and (ii4)
\\\>\>(ii8)\> $D_l;F_l;\cdot;\Gamma_l;\bot \vdash \tmof(\lexp): s$ 
   and $\trans{s}_D = (\tau, \_)$
\\\>By (ii8), \rulename{P-T-E-New}
\\\>\>(ii9)\> $D_l;F_l;\cdot;\Gamma_l;\bot \vdash \tmof(\enew(\lexp)):
\tptr(s)\ \pol$ 
\end{tabbing}

\item[Case:] $\ee$ ends in \rulename{L-Deref} rule. 
\begin{tabbing}
By assumption: 
\\~~\=(1)~~~~\= $D_a;F_a;\Gamma_a; s \vdash \inscast{*v} 
 \Rightarrow\ \lexp$ and $\lexp=\elet\ y:b\ \bot = \erelab(\bot\Leftarrow \pol)\ \lv\ 
\ein\ *(y@b\ \bot) $
\\\>(2) \> $\ee'::D_a;F_a;\Gamma_a \vdash \inscast{v}\Rightarrow \lv$
and $\tpof(\lv) = b\ \pol$
\\By examining the translation rules, only \rulename{T-Let} applies 
\\\>(3)\> $\trans{\lexp}_D =  (e', D'_3\cup D''_3)$  and 
     $e' = \elet\ y:b = e_1\ \ein\ *y$, 
\\\>(4)\> and    $\trans{\erelab(\bot\Leftarrow \pol)\ \lv}_D = (e_1,
D_3)$
\\By assumption and typing rules 
\\\>(5)\>  $D\cup D_1\cup D_2\cup D_3;F;\cdot;\Gamma\vdash e':\tau$
\\By inversion of (5)
\\\>(6)\>  $D\cup D_1\cup D_2\cup D_3;F;\cdot;\Gamma\vdash e_1: \tau_y$
\\\>(7)\> $\tau_y = \tptr(\tau)$, and $\trans{b}_D=(\tau_y,\_)$
\\\>(8)\>  $D\cup D_1\cup D_2\cup D_3;F;\cdot;\Gamma,y:\tau_y \vdash *y:
\tau$ 
\\By (7)
\\\>(9)\> $b=\tptr(s)$ and $\trans{s}_D=(\tau,\_)$
\\By \rulename{P-T-E-Deref}
\\\>(10)\>$D_l;F_l;\cdot;\Gamma_l, y:b\ \bot;\bot \vdash *y: s$ 
\\By Lemma~\ref{lem:relab-has-tp}, (4) and (6)
\\\>(13)\>  $D\cup D_1\cup D_2\cup D_3;F;\cdot;\Gamma\vdash
\m{fst}(\trans{\lv}_D): \tau'$
\\By Lemma~\ref{lem:val-translation-soundness} on $\ee'$, (5), (7),
and (13)
\\\>(14)\> $D_l;F_l;\cdot;\Gamma_l \vdash \tmof(\lv): b\ \pol$
\\By (14), \rulename{P-T-E-Relab}
\\\>(15)\> $D_l;F_l;\cdot;\Gamma_l;\bot \vdash
  \tmof(\erelab(\bot\Leftarrow \pol)\ \lv): b\ \bot$ 
\\ By \rulename{P-T-E-Let}, (15),(9) (10),
\\\>(16)\> $D_l;F_l;\cdot;\Gamma_l;\bot \vdash \lexp: s$ 

\end{tabbing}
%
%
\item[Case:] $\ee$ ends in \rulename{L-Assign} rule. 
\begin{tabbing}
By assumption: 
\\~~\=(1)~~~~\= $D_a;F_a;\Gamma_a; s \vdash \inscast{v:=e} 
 \Rightarrow\ \lexp$ 
\\\>\> and $\lexp=\elet\ y: \tptr(s)\ \bot = \erelab(\bot\Leftarrow \rho)\ \lv\ 
\ein\  y@\tptr(s)\ \bot :=\lexp_2 $
\\\>(2) \> $\ee'::D_a;F_a;\Gamma_a \vdash \inscast{v}\Rightarrow \lv$
and $\tpof(\lv) = \tptr(s)\ \pol$
\\\>(3)\> $\ee'':: D_a;F_a;\Gamma_a ; s\vdash e \Rightarrow \lexp_2$
\\By examining the translation rules, only \rulename{T-Let} applies 
\\\>(4)\> $\trans{\lexp}_D =  (e', D'_3\cup D''_3)$  and 
     $e' = \elet\ y:\tau_y = e_1\ \ein\ y:= e_2$, 
\\\>(5)\> and    $\trans{\erelab(\bot\Leftarrow \pol)\ \lv}_D = (e_1,
D'_3)$ 
\\\>(6)\> and    $\trans{\lexp_2}_D = (e_2, D''_3)$, 
\\\>(7)\>and $\trans{\tptr(s)\ \bot}_D=(\tau_y,\_)$, 
\\By assumption and typing rules 
\\\>(8)\>  $D\cup D_1\cup D_2\cup D_3;F;\cdot;\Gamma\vdash e':\tau'$
\\By inversion of (7)
\\\>(9)\>  $D\cup D_1\cup D_2\cup D_3;F;\cdot;\Gamma\vdash e_1: \tau_y$
\\\>(10)\> $\tau_y = \tptr(\tau)$, and $\tau'=\tunit$,
\\\>(11)\>  $D\cup D_1\cup D_2\cup D_3;F;\cdot;\Gamma,y:\tau_y \vdash e_2: \tau$ 
\\By (7) and (10)
\\\>(12)\> $\trans{s}_D=(\tau,\_)$
\\\> By I.H. on $\ee''$, (3), (6), (11)
\\\>(13)\> $D_l;F_l;\cdot;\Gamma_l;\bot \vdash \tmof(\lexp_2): s$
\\\> By \rulename{T-Assign}, (13), 
\\\>(14)\>$D_l;F_l;\cdot;\Gamma_l, y:\tau_y\ \bot;\bot \vdash
y@\tptr(s)\ \bot :=\lexp_2 :\tunit$ 
\\By Lemma~\ref{lem:relab-has-tp}, (5) and (9)
\\\>(15)\>  $D\cup D_1\cup D_2\cup D_3;F;\cdot;\Gamma\vdash
\m{fst}(\trans{\lv}_D): \tau'$
\\By Lemma~\ref{lem:val-translation-soundness} on $\ee'$, (2), (15),
\\\>(16)\> $D_l;F_l;\cdot;\Gamma_l \vdash \tmof(\lv): \tptr(s)\ \pol$
\\By (16), \rulename{P-T-E-Relab}
\\\>(17)\> $D_l;F_l;\cdot;\Gamma_l;\bot \vdash
  \tmof(\erelab(\bot\Leftarrow \pol)\ \lv): \tptr(s)\ \bot$ 
\\ By \rulename{P-T-E-Let}, (17) (14),
\\\>(18)\> $D_l;F_l;\cdot;\Gamma_l;\bot \vdash \lexp: \tunit$ 

\end{tabbing}
%
%
\item[Case:] $\ee$ ends in \rulename{L-If} rule. 
\begin{tabbing}
By assumption: 
\\~~\=(1)~~~~\= $D_a;F_a;\Gamma_a; s \vdash 
\inscast{\eif\ v_1\ \ethen\  e_2\ \eelse\  e_3} 
 \Rightarrow\ \lexp$ 
\\\>\> and $\lexp=\elet\ x: \tpint\ \bot = (\erelab(\bot\Leftarrow\pol)\ \lv_1)\
 \ein\ \eif\ x@\tpint\ \bot\ \ethen\  \lexp_2\ \eelse\ \lexp_3$
\\\>(2) \> $\ee'::D_a;F_a;\Gamma_a \vdash \inscast{v}\Rightarrow \lv$
and $\tpof(\lv) = \tpint\ \pol$
\\\>(3)\> $\ee'':: D_a;F_a;\Gamma_a ; s\vdash \inscast{e_2}
\Rightarrow \lexp_2$
\\\>(4)\> $\ee''':: D_a;F_a;\Gamma_a ; s\vdash \inscast{e_3} \Rightarrow \lexp_3$
\\By examining the translation rules, only \rulename{T-Let} applies 
\\\>(5)\> $\trans{\lexp}_D =  (e', D'_3\cup D''_3)$  and 
     $e' = \elet\ x:\tpint= e_1\ \ein\ \eif\ x\ \ethen\  e'_2\ \eelse\ e'_3$, 
\\\>(6)\> and    $\trans{\erelab(\bot\Leftarrow \pol)\ \lv}_D = (e_1,
D'_3)$ 
\\\>(7)\> and  $\trans{\lexp_2}_D = (e'_2, D''_3)$,   and  $\trans{\lexp_3}_D = (e'_3, D'''_3)$, 
\\By assumption and typing rules 
\\\>(8)\>  $D\cup D_1\cup D_2\cup D_3;F;\cdot;\Gamma\vdash e':\tau$
\\By inversion of (8)
\\\>(9)\>  $D\cup D_1\cup D_2\cup D_3;F;\cdot;\Gamma\vdash e_1: \tpint$
\\\>(10)\>  $D\cup D_1\cup D_2\cup D_3;F;\cdot;\Gamma,x:\tpint\vdash
e'_2: \tau$ 
\\\>(11)\>  $D\cup D_1\cup D_2\cup D_3;F;\cdot;\Gamma,x:\tpint \vdash e'_3: \tau$ 
\\ By I.H. on $\ee''$, (3), (7), (10)
\\\>(12)\> $D_l;F_l;\cdot;\Gamma_l;\bot \vdash \tmof(\lexp_2): s$
 and $\trans{s}_D = (\tau, \_)$
\\ By I.H. on $\ee'''$, (4), (7), (11)
\\\>(13)\> $D_l;F_l;\cdot;\Gamma_l;\bot \vdash \tmof(\lexp_3): s$
\\ By \rulename{T-If}, (12), (13)
\\\>(14)\>$D_l;F_l;\cdot;\Gamma_l, x:\tpint\ \bot;\bot
  \vdash \tmof(\eif\ x@\tpint\ \bot\ \ethen\  \lexp_2\ \eelse\ \lexp_3) :s$
\\By (5) there are two subcases
\\\> {\bf Subcase $\pol\in\{\bot,\un\}$}
\\\> \rulename{T-ReLab-Same} applies
\\\>~~\=(i1)~~~~\= 
    $e_1=v_1$  and $\trans{\lv}_D=(v_1, D_3)$
\\\>By Lemma~\ref{lem:val-translation-soundness} on $\ee'$, (6),
and (i11)
\\\>\>(i2)\> $D_l;F_l;\cdot;\Gamma_l \vdash \tmof(\lv): \tpint\ \pol$
\\\>By (i2), \rulename{P-T-E-Val}
\\\>\>(i3)\> $D_l;F_l;\cdot;\Gamma_l;\bot \vdash \tmof(\lv): \tpint\ \pol$ 
\\\>\>(i4)\> $D_l;F_l;\cdot;\Gamma_l;\bot \vdash
  \tmof(\erelab(\bot\Leftarrow \pol)\ \lv): \tpint\ \bot$ 
\\\> By \rulename{P-T-E-Let}, (i4), (15),
\\\>\>(i5)\> $D_l;F_l;\cdot;\Gamma_l;\bot \vdash \lexp: s$ 
\\\> {\bf Subcase $\pol\notin\{\bot,\un\}$}
\\\>\rulename{T-ReLab-N3} applies 
\\\>\>(ii1) $e_1=v'.1$   and $\trans{\lv}_D=(v', D_3)$
\\\> By inversion of (9)
\\\>\>(ii2)\> 
  $D\cup D_1\cup D_2\cup D_3;F;\cdot;\Gamma\vdash v': T'$, 
\\\>\>(ii3)\> $(D\cup D_1\cup D_2\cup D_3)(T') = \tstruct\ T'\{\tpint\}$
\\\>By Lemma~\ref{lem:val-translation-soundness} on $\ee'$, (2), (ii3),
and (ii1)
\\\>\>(ii4)\> $D_l;F_l;\cdot;\Gamma_l \vdash \tmof(\lv): \tpint\ \pol$
 and $\trans{\tptr(s)\ \pol}_D = (T', D'_3)$
\\\>By (ii4), \rulename{P-T-E-Val}
\\\>\>(ii5)\> $D_l;F_l;\cdot;\Gamma_l;\bot \vdash \tmof(\lv): \tpint\ \pol$ 
\\\>\>(ii6)\> $D_l;F_l;\cdot;\Gamma_l;\bot \vdash
  \tmof(\erelab(\bot\Leftarrow \pol)\ \lv): \tpint\ \bot$ 
\\\> By \rulename{P-T-E-Let}, (ii6) (15),
\\\>\>(ii7)\> $D_l;F_l;\cdot;\Gamma_l;\bot \vdash \lexp: s$ 
\end{tabbing}
%
%
\item[Case:] $\ee$ ends in \rulename{L-Let} rule. 
\begin{tabbing}
By assumption: 
\\~~\=(1)~~~~\= $D_a;F_a;\Gamma_a; s \vdash 
\inscast{\elet\ x:\beta_1 = e_1\ \ein\
    e_2} \Rightarrow\ \lexp$ 
\\\>\> and $\lexp=\elet\ x:\inscast{\beta_1} = \lexp_1\ \ein\ \lexp_2$
\\\>(2) \> $\ee':: D_a;F_a;\Gamma_a ; \inscast{\beta_1}\vdash
\inscast{e_1} \Rightarrow \lexp_1$ 
\\\>(3)\> $\ee'':: D_a;F_a;\Gamma_a , x:\beta_1 ; t_2 \vdash  \inscast{e_2} \Rightarrow \lexp_2$
\\By examining the translation rules, only \rulename{T-Let} applies 
\\\>(4)\> $\trans{\lexp}_D =  (e', D'_3\cup D''_3)$  and 
     $e' = \elet\ x:\tau= e'_1\ \ein\ e'_2\ $
\\\>(5)\> $\trans{\inscast{\beta_1}}_D= (\tau,\_)$
\\\>(6)\> and  $\trans{\lexp_1}_D = (e'_1, D'_3)$,   and  $\trans{\lexp_2}_D = (e'_2, D''_3)$, 
\\By assumption and typing rules 
\\\>(7)\>  $D\cup D_1\cup D_2\cup D_3;F;\cdot;\Gamma\vdash e':\tau'$
\\By inversion of (7)
\\\>(8)\>  $D\cup D_1\cup D_2\cup D_3;F;\cdot;\Gamma\vdash e'_1: \tau$
\\\>(9)\>  $D\cup D_1\cup D_2\cup D_3;F;\cdot;\Gamma,x:\tau\vdash
e'_2: \tau'$ 
\\ By I.H. on $\ee'$, (2), (6), (8)
\\\>(10)\> $D_l;F_l;\cdot;\Gamma_l;\bot \vdash \tmof(\lexp_1): \inscast{\beta_1}$
 and $\trans{\inscast{\beta_1}}_D = (\tau, \_)$
\\ By I.H. on $\ee'''$, (3), (6), (9)
\\\>(11)\> $D_l;F_l;\cdot;\Gamma_l, x: \inscast{\beta_1};\bot \vdash \tmof(\lexp_2): s$
\\ By \rulename{T-Let}, (10), (11)
\\\>(12)\>$D_l;F_l;\cdot;\Gamma_l;\bot \vdash \lexp:s$
\end{tabbing}
%
%
\item[Case:] $\ee$ ends in \rulename{L-App} rule. 
\begin{tabbing}
By assumption: 
\\~~\=(1)~~~~\= $D_a;F_a;\Gamma_a ; t_2\vdash \inscast{v\ e}
\Rightarrow \lv\ \lexp$
\\\>(2) \> $\ee'::D_a;F_a;\Gamma_a \vdash \inscast{v}\Rightarrow \lv$
and $\tpof(\lv) = [\bot](t_1\rightarrow t_2)^\bot$
\\\>(3)\> $\ee'':: D_a;F_a;\Gamma_a ; t_1\vdash e \Rightarrow \lexp$
\\By examining the translation rules, only \rulename{T-App} applies 
\\\>(4)\> $\trans{\lv\ \lexp}_D =  (v'\; e', D'_3\cup D''_3)$  and 
\\\>(5)\> and    $\trans{\lv}_D = (v', D'_3)$ 
\\\>(6)\> and    $\trans{\lexp}_D = (e', D''_3)$, 
\\By assumption and typing rules 
\\\>(7)\>  $D\cup D_1\cup D_2\cup D_3;F;\cdot;\Gamma\vdash v'\;e':\tau_2$
\\By inversion of (7)
\\\>(8)\>  $D\cup D_1\cup D_2\cup D_3;F;\cdot;\Gamma\vdash v': \tau_1\rightarrow\tau_2$
\\\>(9)\>  $D\cup D_1\cup D_2\cup D_3;F;\cdot;\Gamma \vdash e': \tau_1$ 
\\By I.H. on $\ee''$, (6), (9)
\\\>(10)\> $D_l;F_l;\cdot;\Gamma_l;\bot \vdash \tmof(\lexp): t_1$ and
$\trans{t_1}_D = (\tau_1,\_)$
\\By Lemma~\ref{lem:val-translation-soundness} on $\ee'$, (5), (8)
\\\>(11)\> $D_l;F_l;\cdot;\Gamma_l \vdash \tmof(\lv):
[\bot](t_1\rightarrow t_2)^\pol$
and $\trans{[\bot](t_1\rightarrow t_2)^\pol}_D= \tau_1\rightarrow\tau_2$
\\By (11), \rulename{P-T-E-Val}
\\\>(12)\> $D_l;F_l;\cdot;\Gamma_l;\bot \vdash \tmof(\lv): [\bot](t_1\rightarrow t_2)^\bot$
\\ By \rulename{P-T-E-App}, (10), (12),
\\\>(13)\> $D_l;F_l;\cdot;\Gamma_l;\bot \vdash \lv\;\lexp: t_2$ 
\\By (11)
\\\>(14)\> $\trans{t_2}_D= (\tau_2,\_)$
\end{tabbing}

%
%
\item[Case:] $\ee$ ends in \rulename{L-App-De} rule. 
\begin{tabbing}
By assumption: 
\\~~\=(1)~~~~\= $D_a;F_a;\Gamma_a ; t_2 \vdash \inscast{v_f\ v_a} 
\Rightarrow \lv_f\ \lv_a$ 
\\\>(2) \> $\ee'::D_a;F_a;\Gamma_a \vdash \inscast{v_f}\Rightarrow \lv_f$
and $\tpof(\lv_f) = (\dne)[\bot](t_1\rightarrow t_2)^\bot$
\\\>(3)\> $\ee'':: D_a;F_a;\Gamma_a \vdash v_a \Rightarrow \lv_a$
\\By examining the translation rules, only \rulename{T-App-De} applies 
\\\>(4)\> $\trans{\lv_f\ \lv_a}_D =  (e', D'_3\cup D''_3\cup D'''_3\cup D''''_3\cup D'''''_3)$  and 
     $e' = \elet\ y:\tau_1 = e_a\ \ein\  \elet\ z: \tau_2 = v'_f\ y\ \ein\
     e_2$
\\\>(5)\> $\tpof(\lv_a) = b\ \pol$ and  $\pol = \lab_1::\lab_2::\pol'$
\\\> (6)\> 
 $\trans{\erelab(\lab_1::\top\Leftarrow \pol) \lv_a} = (e_a,D'_3)$
\\\>(7)\> $\trans{\erelab(\lab_2::\pol'\Leftarrow \lab_2::\bot) (z@b\ \lab_2::\bot)}= (e_2,D''_3)$ 
\\ \>(8)\>$\trans{t_1} = (\tau_1, D'''_3)$ and $\trans{t_2} = (\tau_2,
D''''_4)$
\\\>(9)\> and    $\trans{\lv_f}_D = (v'_f, D'''''_3)$
\\By assumption, 
\\\>(10)\>  $D\cup D_1\cup D_2\cup D_3;F;\cdot;\Gamma\vdash e':\tau'$
\\By inversion of (10)
\\\>(11)\>  $D\cup D_1\cup D_2\cup D_3;F;\cdot;\Gamma\vdash e_a: \tau_1$
\\\>(12)\>  $D\cup D_1\cup D_2\cup D_3;F;\cdot;\Gamma,y:\tau_1 \vdash
v'_f\; y: \tau_2$ 
\\\>(13)\>  $D\cup D_1\cup D_2\cup D_3;F;\cdot;\Gamma,y:\tau_1,
z:\tau_2 \vdash e_2: \tau'$ 
\\By inversion of (12)
\\\>(14) \>$D\cup D_1\cup D_2\cup D_3;F;\cdot;\Gamma,y:\tau_1 \vdash
v'_f:\tau_1\rightarrow \tau_2$ 
\\By  Lemma~\ref{lem:val-translation-soundness} on $\ee'$, and
(14)
\\\>(15)\> $D_l;F_l;\cdot;\Gamma_l;\bot \vdash \tmof(\lv_f):
(\dne)[\bot](t_1\rightarrow t_2)^{\bot}$ 
\\By Lemma~\ref{lem:relab-has-tp}, (5), and (11)
\\\>(16) \>$D\cup D_1\cup D_2\cup D_3;F;\cdot;\Gamma\vdash
\m{fst}(\trans{\lv_a}_D): \tau'$ and $\tau'=\m{fst}(\trans{b\ \pol}_D)$
\\By Lemma~\ref{lem:val-translation-soundness} 
\\\>(17)\> $D_l;F_l;\cdot;\Gamma_l \vdash \tmof(\lv_a): b\ \pol$ and 
\\By Lemma~\ref{lem:relab-sound} on (6) (11) (16)
\\\>(18)\> $\tau_1=\m{fst}(\trans{b\ \lab_1::\top}_D)$
\\By Lemma~\ref{lem:trans-unique} on (8) (18)
\\\>(19)\> $t_1 =b\ \lab_1::\top $
\\By Lemma~\ref{lem:relab-has-tp} and (7)
\\\>(20)\>$D\cup D_1\cup D_2\cup D_3;F;\cdot;\Gamma\vdash
\m{fst}(\trans{z@b\ \lab_2::\bot}_D): \tau''$ and
$\tau''=\m{fst}(\trans{z@b\ \lab_2::\bot}_D)$
\\By Lemma~\ref{lem:trans-unique} on and $z$ has type $\tau_2$ (8) (20)
\\\>(21)\> $t_2 =b\ \lab_2::\bot $
\\By \rulename{P-T-E-App-De}
\\\>(22)\> $D_l;F_l;\cdot;\Gamma_l;\bot \vdash \tmof(\lv_f\ \lv_a): b\
\lab_2::\pol'$ 
\end{tabbing}
\end{description}
\end{proof}


\end{document}